\documentclass{article} % DO NOT CHANGE THIS
\usepackage[utf8]{inputenc}
\usepackage[margin=1in]{geometry}
\usepackage{authblk}

\usepackage{times}  % DO NOT CHANGE THIS
\usepackage{helvet} % DO NOT CHANGE THIS
\usepackage{courier}  % DO NOT CHANGE THIS
\usepackage[switch]{lineno}  
\usepackage[hyphens]{url}  % DO NOT CHANGE THIS
\usepackage{graphicx} % DO NOT CHANGE THIS
\usepackage[T1]{fontenc}    % use 8-bit T1 fonts
\usepackage{hyperref}       % hyperlinks

\usepackage{booktabs}       % professional-quality tables
\usepackage{amsfonts,mathtools}       % blackboard math symbols
\usepackage{nicefrac}       % compact symbols for 1/2, etc.
\usepackage{microtype}      % microtypography
\usepackage{xcolor}         % colors
\usepackage{amsmath,bm}
\usepackage{amssymb}
\allowdisplaybreaks
\usepackage{amsmath,amsfonts,amsthm} % Math packages

\DeclareMathOperator*{\argmin}{arg\,min}
\usepackage{color}
\usepackage{lipsum} % Used for inserting dummy 'Lorem ipsum' text into the template

\DeclareMathOperator*{\E}{\mathbb{E}}

\usepackage{enumerate}
\usepackage{enumitem}
\usepackage{setspace}

\usepackage[open]{bookmark}

\makeatletter
\newcommand{\RN}[1]{\textup{\uppercase\expandafter{\romannumeral#1}}}

\newcommand{\Rmnum}[1]{\expandafter\@slowromancap\romannumeral #1@}
\makeatother

\newtheorem{theorem}{Theorem}

\newtheorem{lemma}{Lemma}
\newtheorem{proposition}{Proposition}
\newtheorem{corollary}{Corollary}

\newtheorem{definition}{Definition}
\newtheorem{assumption}{Assumption}
\newtheorem{remark}{Remark}

\usepackage{graphicx}

\usepackage{tikz}
\usepackage{subfig}
\usepackage{mathtools}

\newcommand{\red}[1]{\textcolor{red}{#1}}

\newcommand{\blue}[1]{\textcolor{blue}{#1}}

\definecolor{dark}{rgb}{0.35, 0.15, 0.13}

\DeclarePairedDelimiter\ceil{\lceil}{\rceil}
\newcommand{\R}{\mathbb R}
\newcommand{\T}{\mathcal T}
\newcommand{\Esafe}{\mathcal E_{\text{safe}}}
\newcommand{\ind}{\mathrel{\text{\scalebox{1.07}{$\perp\mkern-10mu\perp$}}}}

\newcommand{\D}{\mathcal D}

\newcommand{\U}{\mathbb U}
\newcommand{\W}{\mathbb W}
\newcommand{\A}{\mathcal A}
\newcommand{\F}{\mathcal F}

\newcommand{\X}{\mathbb X}
\newcommand{\Pb}{\mathbb P}
\newcommand{\B}{\mathbb B}

\newcommand{\M}{\mathcal M}
\newcommand{\K}{\mathcal K}

\newcommand{\Kb}{\mathbb K}
\newcommand{\mb}{\mathbf}
\usepackage{dsfont}
\newcommand{\one}{\mathds{1}}
\newcommand{\id}{\mathbb{I}}
\newcommand{\nbf}{\noindent\textbf}

\usepackage{algorithm}
\usepackage{algorithmic}
\usepackage[algo2e,ruled,noend]{algorithm2e}

\SetCommentSty{mycommfont}

\usepackage[round]{natbib}

\usepackage{comment}
\excludecomment{myspecial}
\excludecomment{mycitation}

\title{Safe Adaptive Learning-based Control for Constrained Linear Quadratic Regulators with Regret Guarantees}
\author[1]{Yingying Li}
\author[2]{Subhro Das}
\author[1]{Jeff Shamma}
\author[3]{Na Li}
\affil[1]{University of Illinois at Urbana-Champaign}
\affil[2]{MIT-IBM Watson AI Lab, IBM Research}
\affil[3]{SEAS Harvard University}
\date{} 
\begin{document}

\maketitle
\begin{abstract}
%\red{to revise}

We study the adaptive control of an unknown linear system with a quadratic cost function subject to safety constraints on both the states and actions. The challenges of this problem arise from the tension among safety, exploration, performance, and computation. To address these challenges, we propose a polynomial-time   algorithm that guarantees  feasibility and constraint satisfaction  with high probability under proper conditions. Our algorithm is implemented on a single trajectory and does not require system restarts. Further, we analyze the regret of our learning algorithm compared to the optimal safe linear controller with known model information. The proposed algorithm can achieve a $\tilde O(T^{2/3})$ regret, where $T$ is the number of stages and $\tilde O(\cdot)$ absorbs some logarithmic terms of $T$.
\end{abstract}

\section{Introduction}
\begin{myspecial}
	\red{Q: do I advertise recursive feasibility, or feasibility?}
	
	\red{Q: should I explain recursive feasibility somewhere?}
\end{myspecial}

%\red{Lina: the flow of the intro looks good but repetitive}

%\red{when you revise the intro, pay attention when you explain that we are doing "one trajectory" learning.
%need to have that introduced before and in "our contributions"}

Recent years have witnessed great interest in learning-based, and a lot of results have been developed for \textit{unconstrained} systems \citep{fazel2018global,dean2018regret,dean2017sample,mania2019certainty,simchowitz2018learning,simchowitz2020improper,cohen2019learning}. However,  practical systems usually face \textit{constraints} on the states and control inputs, especially in safety-critical applications \citep{campbell2010autonomous,vasic2013safety}. For example, drones are not supposed to visit certain locations to avoid collision and the thrusts of drones are usually bounded. Therefore, it is crucial to study \textit{safe} learning-based control for \textit{constrained} systems.

%learning to i.% systems to ensure  safety. 

As a starting point, this paper considers a linear quadratic regulator (LQR) with linear constraints on the states and actions, i.e., 
\begin{align}\label{equ: state and action constraints}
	D_x x_t \leq d_x, \quad D_u u_t \leq d_u.
\end{align}
We consider a linear  system $	x_{t+1}=A_* x_t + B_* u_t +w_t$ with bounded system disturbances $w_t \in \W=\{w: \|w\|_\infty \leq w_{\max}\}$ and unknown model $(A_*, B_*)$. We aim to design an adaptive  control algorithm to  minimize the quadratic cost $\E[x_t^\top Q x_t +u_t^\top R u_t]$ with \textit{safety} guarantees during the learning process, i.e. satisfying all the constraints for any  disturbances $w_t\in\W$.

% constraints. Specifically, we consider we consider a  linear system
% \begin{equation}\label{equ: linear system A* B*}
% 	x_{t+1}=A_* x_t + B_* u_t +w_t
% \end{equation}
% with linear constraints on the states and actions 
% \begin{align}\label{equ: state and action constraints}
% 	D_x x_t \leq d_x, \quad D_u u_t \leq d_u,
% \end{align}
% and i.i.d. system disturbances with bounded support, i.e., $w_t \in \W=\{w: \|w\|_\infty \leq w_{\max}\}$. We consider unknown system parameters $(A_*, B_*)$ and aim to minimize the quadratic costs $\E[x_t^\top Q x_t +u_t^\top R u_t]$ by  designing an adaptive learning-based control algorithm with constraint satisfaction guarantees

%out violating the constraints for bounded disturbances $w_t\in\W$ (safety). 

%. We aim to design an adaptive learning-based controller to minimize the quadratic costs $\E[x_t^\top Q x_t +u_t^\top R u_t]$ without violating the constraints \eqref{equ: state and action constraints} for any possible disturbances $w_t\in \W$.

The constraints on LQR bring great difficulties even when the model is known. Unlike unconstrained LQR, which enjoys  closed-form optimal policies \citep{lewis2012optimal}, there is no computationally efficient method to solve the optimal policy for constrained LQR \citep{rawlings2009model}.\footnote{Efficient algorithms exist for some special cases, e.g. when $w_t=0$, the optimal controller is piecewise-affine and can be computed as in \cite{bemporad2002explicit}.} Thus, most literature  sacrifices optimality for computation efficiency by  designing policies with certain structures, e.g. linear policies \citep{dean2019safely,li2020online}, piecewise-affine polices in robust model predictive control (RMPC) \citep{bemporad1999robust,rawlings2009model}, etc.

Therefore, when the model is unknown, a reasonable goal  is to learn and  achieve what can be obtained with perfect model information. In this paper, we adopt the optimal safe linear policy as our benchmark/target and leave the discussions on RMPC as future work.

The current literature on learning the optimal safe linear policies adopts  an offline/non-adaptive learning approach, which does not improve the policies until the learning terminates \citep{dean2019safely}. To improve the control performance during learning, adaptive/online learning-based  control algorithms should be  designed. However, though adaptive learning for unconstrained LQR can be designed by direct conversions from offline algorithms  (see e.g., \citep{simchowitz2020naive,mania2019certainty,dean2018regret}), it is  much more challenging for the  constrained case because  direct conversions may cause infeasibility and/or constraint violation  for single-trajectory adaptive learning as noted in \cite{dean2019safely}.

\nbf{Our contributions.} In this paper, we propose a single-trajectory adaptive learning algorithm for constrained LQR with feasibility and constraint satisfaction guarantees.  Our algorithm estimates the model with a least-square-estimator (LSE) and updates the policies with improved model estimations based on certainty-equivalence (CE) with robust constraint satisfaction against model uncertainties. To ensure safe  policy updates, we propose a SafeTransit algorithm by extending the slow-variation trick in \cite{li2020online} for a known model to the case with model uncertainties and varying model estimations. Our algorithms  can be implemented in polynomial time at each stage.

%Notice that  \cite{li2020online} requires perfect model information but we extend their trick to handle model uncertainties and time-varying model estimations.

%Theoretically, we establish feasibility and constraint satisfaction of our algorithm under proper conditions.

Further,  we provide performance guarantees for  our learning algorithm by discussing the regret  compared with the optimal safe linear policy with perfect model information.
 We obtain a sublinear regret bound of order $\tilde O(T^{2/3})$. Interestingly,   our regret bound also holds when compared against an RMPC algorithm (RMPC)  proposed in \cite{mayne2005robust}. Discussions on more general regret benchmarks are left for the future.

Lastly, when developing our theoretical results, we provide a model estimation error bound for general and possible nonlinear policies. This is to handle the potential nonlinearity in our designed controllers when the model errors are non-negligible. Our error bound extends  the existing results for linear policies  in \cite{dean2017sample,dean2019safely} and can be useful by its own.
%maybe "on its own"

\nbf{Related work.}  \textit{Constrained LQR with linear policies} is studied in \cite{dean2019safely,li2020online}.  \cite{dean2019safely} consider  an \textit{unknown} model and  propose an offline learning method with sample complexity guarantees.  In contrast, \cite{li2020online} study online constrained LQR with a \textit{known} model and  adopt a slow-variation trick for safe policy updates. However, it remains open how to ensure safe adaptive control under  model uncertainties.

%\red{explain rmpc achieves feasibility, but not regret.}

\textit{Constrained LQR by model predictive control (MPC).} MPC and its variants are popular methods for constrained control, e.g. RMPC designed for hard constraints \citep{mayne2005robust,limon2010robust,rawlings2009model}, and stochastic MPC methods for soft constraints \citep{mesbah2016stochastic,oldewurtel2008tractable}. With model uncertainties, robust adaptive MPC (RAMPC) algorithms are proposed to learn the model and updates the policies \citep{zhang2020adaptive,bujarbaruah2019adaptive,kohler2019linear,lu2019robust}. Most RAMPC algorithms guarantee recursive feasibility and constraint satisfaction but lack non-asymptotic performance guarantees compared with the known-model case. In contrast, there are some recent results on non-asymptotic regret bounds by sacrificing feasibility and/or constraint satisfaction, e.g., \cite{wabersich2020performance} establish a regret bound for an adaptive MPC algorithm that requires restarting the system  to some safe feasible state, \cite{muthirayan2020regret} provides a regret bound for an adaptive algorithm without considering state constraints.

%only considering control constraints.

%with $o(T)$  regret but does not consider state constraints. %Further, (cite??) designs a single-trajectory  adaptive algorithm with both state and action constraints and obtains a $\tilde O(\sqrt T)$ regret bound, but their  regret  benchmark is weaker than  RMPC with a known model because their benchmark policy is robustly safe for all  models in some initial set instead of only being safe for the true model.

\textit{Learning-based unconstrained LQR}  enjoys rich literature, so we only review the most related papers below. Firstly, our algorithm is related with the CE-based
% the Grammarly show it shoud be related to
adaptive control  \citep{dean2018regret,mania2019certainty,cohen2019learning,simchowitz2020naive}, and this approach is shown to be optimal for the unconstrained LQR \citep{mania2019certainty,simchowitz2020naive}.
Further, similar to \cite{agarwal2019online,agarwal2019logarithmic,plevrakis2020geometric}, we adopt the disturbance-action policies to approximate  linear policies.% in a computationally efficient manner.

%Among the rich literature on learning-based unconstrained LQR, CE-based adaptive control is most related to this paper \cite{dean2018regret,mania2019certainty,cohen2019learning,simchowitz2020naive}.

%This enjoys a rich literature \citep{fazel2018global,dean2018regret,dean2017sample,mania2019certainty,simchowitz2018learning,simchowitz2020improper,cohen2019learning}. Same with \citet{agarwal2019online,agarwal2019logarithmic}, this paper utilizes the disturbance-action policies to approximate linear policies. Further, similar to \citet{dean2018regret,mania2019certainty,simchowitz2020naive}, this paper leverages LSE and CE for adaptive control design. Such approaches are shown to be optimal in unconstrained LQR \citep{mania2019certainty,simchowitz2020naive}.

%our design utilizes LSE for model estimation and certainty equivalence for policy updates, which is related with \citet{dean2018regret,mania2019certainty,simchowitz2020naive}.  %which has been shown to be optimal for learning-based control without constraints \citep{mania2019certainty,simchowitz2020naive}.  %Hence, it is natural  to study the performance of certainty equivalent controllers in the constrained setting. Lastly,
%The robust stability guarantee is also studied in \cite{dean2018regret,dean2019sample,chow2018lyapunov}. 
%\red{add safe learning, with constraints, or asymp. guarantees.}\lina{yingying, you have safe RL. put them there. I am happy with the current four topics in related work. }
%\red{say many safe algorithm design, but both safe and optimality is under-explored!}

\noindent\textit{Safe reinforcement learning (RL).} Safety in RL has different definitions \citep{mihatsch2002risk,garcia2015comprehensive}. This paper is related 
% the Grammarly show it shoud be related to
with RL with  constraints, which enjoys a lot of research  but has limited results on both safety and non-asymptotic optimality guarantees \citep{marvi2021safe,leurent2020robust,fisac2018general,garcia2015comprehensive,cheng2019end,fulton2018safe}. %Besides, safety in RL 

%enjoys a lot of research with different safety requirements \citep{mihatsch2002risk,garcia2015comprehensive}. This paper is most relevant to safety  defined  by state and action constraints \citep{marvi2021safe,leurent2020robust,fisac2018general,garcia2015comprehensive,cheng2019end,fulton2018safe}, but non-asymptotic optimality guarantees for safe RL are under-explored. %Most safe RL algorit lack non-asymptotic optimality guarantees.%Most safe RL papers 

\textit{Model estimation for nonlinear systems.} There are model estimation guarantees for general nonlinear systems \citep{foster2020learning,sattar2020non}, but our estimation error bound leverages the special structure of our problem, i.e., nonlinear policies on a linear system, to obtain better guarantees.

%For general no

%Our estimation error bound holds for linear systems with nonlinear policies, which induces a special type of closed-loop nonlinear systems. There are other  estimation error bounds for  general nonlinear systems,   \citep{foster2020learning,sattar2020non}.{\color{red} comment on the difference or the weakness?}
%Many papers consider soft constraints by only aiming for sublinear number of stages when constraint violation happen \cite{paternain2019safe,yang2021wcsac,wabersich2020performance}. 

\nbf{Notations.}
For a distribution $\D_\eta$, we write $\eta \overset{\text{ind.}}{\sim} \bar \eta \D_\eta$ if $\eta/\bar \eta $ is generated with distribution $\D_\eta$ and is independent from other random variables in the context. By $\| . \|_F$ we denote the Frobenius norm.  Define $
\B(\hat{\theta},r)=\{\theta: \|\theta-\hat\theta\|_F\leq r\}$. Let $\one_n$ be an all-one vector in $\R^n$. Let $\mathbb I_E$ be an indicator function on set $E$.

\vspace{-5pt}

\section{Problem formulation}
We consider the following constrained LQR problem,
\vspace{-5pt}
\begin{equation}\label{equ: J(pi)}
	\begin{aligned}
	&	\min_{u_0, u_1, \dots}  \lim_{T\to +\infty} \frac{1}{T}\sum_{t=0}^{T-1}\E [l(x_t,u_t)]\\
		\text{s.t. } & x_{t+1}=A_*x_t+B_* u_t +w_t, \ \forall \, t\geq 0,\\
		&D_x x_t \leq d_x,  D_u u_t\leq d_u, \  \forall  \, \{w_k \in \W\}_{k\geq 0}.
	\end{aligned}
\end{equation}
where  $l(x,u)=x^\top Q x+ u^\top R u$, $Q$ and $R$ are positive definite matrices,  $x_t\in \R^n$ is the state with a given initial state $x_0$, $u_t\in \R^m$ is the action, and $w_t$ is the disturbance. The parameters $D_x, d_x$, and $D_u, d_u$ determine the constraint sets of the state and action respectively, where $d_x\in \R^{k_x}, d_u\in \R^{k_u}$. Further, the constraint sets on the state and action are assumed to be bounded with $x_{\max}=\sup_{D_x x \leq d_x}\|x\|_2$ and $ u_{\max}=\sup_{D_u u \leq d_u}\|u\|_2$. Besides, denote $\theta_*\coloneqq (A_*, B_*)$ and $\theta\coloneqq (A, B)$ for simplicity. The model parameters $\theta_*$ are unknown but other parameters are known. 

An algorithm/controller is called `safe' if its induced states and actions satisfy  the constraints for all $t$ under any possible  disturbances $w_t\in\W$, which is also called robust constraint satisfaction under   disturbances $w_t$.

 %In this paper, we consider \textit{unknown} model parameters $(A_*, B_*)$ and consider other parameters as known.
Notice that even with known model $\theta_*$, the optimal policy to   problem \eqref{equ: J(pi)} cannot be computed efficiently, but there are  efficient methods to compute  sub-optimal policies, e.g.    optimal safe linear policies by quadratic programs \citep{dean2019safely,li2020online} and piecewise affine policies by RMPC \citep{mayne2005robust,rawlings2009model}. In this paper, we set the optimal safe linear policy as our learning  goal  and leave  RMPC for future. We aim to achieve our learning goal by designing safe adaptive learning-based control. Further, we  consider single-trajectory  learning, which is more challenging since the system cannot be restarted to ensure feasibility and constraint satisfaction.

%case where there is not restart for the system dynamics. That is, the adaptive learning is performed  on a single trajectory of the dynamics.

%design an  adaptive learning-based control algorithm on a single trajectory with both safety and 

%guarantees on both safety and 

%guarr achieve our goal under model uncertainties. Further, we consider single-trajectory learning, which is more challenging than the learning with restarts in the constrained cases.

%aim to learn to  achieve the optimal safe linear policy and leave the stud

%In this paper, we  goal of this work is  learning to achieve 

%to design  adaptive learning-based control  to achieve the performance of the optimal safe \textit{linear} policy computed with perfect model information while ensuring safety during the entire learning process. We leave the discussions on RMPC for future. {\color{red} We also focus on the more challenging case where there is not restart for the system dynamics. That is, the adaptive learning is performed online on a single trajectory of the dynamics.}  % the first approach and  aims  to achieve the performance by the optimal linear policy with \textit{safe} adaptive learning algorithm design. We call an algorithm `safe' if its induced states and actions satisfy  the constraints for all $t$ under any possible disturbances $w_t \in \W$. 

For simplicity, we assume $x_0=0$. Our results can be generalized to $x_0$  in a small neighborhood around 0.\footnote{Footnote \ref{footnote: generalize small nonzero x0} will provide more discussions on  nonzero $x_0$ and generalization to a small $x_0$. Here, we discuss the implication of considering a small $x_0$. Remember that state 0 represents a desirable equilibrium point of the system.  By considering $x_0$ close to 0, this paper focuses on how to safely optimize the performance around the  equilibrium   instead of safely driving a distant state back to the equilibrium. As an example of applications,  this paper studies how to safely maintain a drone around a target  in the air  despite  wind disturbances with minimum battery consumption, instead of flying the drone safely to the target from a distance. In practice, one can first apply existing algorithms such as \citet{mayne2005robust} to safely steer the system to around 0 and then apply our algorithm to achieve optimality and safety around 0.}

%LQR arises in many applications such as navigation and path tracking. The zero state   corresponds to an equilibrium point that we would like the system to reach and stay. By starting 

%For $x_0$ far away from 0, linear policies are not  good choices for constraint satisfaction and more general policies should be adopted \citep{rawlings2009model}. For $x_0$ close to 0, this paper studies how to optimize the performance around  without violating the constraints, instead of safely steering a distant initial state back to around 0. For applications, consider  maintaining a system (e.g., a drone) around a target state (e.g., a position)
%\red{to revise!??}
%optimally and safely under environment disturbances. If $x_0$ is very distant from $0$, in practice, one can first apply other algorithms, e.g. RMPC (cite??), to safely steer the system to around the equilibrium point  then apply our algorithm to optimize the performance around the equilibrium.} 
%Based on the discussions above, we can formally define the performance metric in this paper below. 

\nbf{Regret  and benchmark.} Roughly speaking, we measure the performance of our adaptive learning algorithm by comparing with 
% comparing it with
the optimal safe linear policy $u_t=-K^* x_t$ computed with perfect model information.

To formally define the performance metric, we first define a quantitative characterization of matrix stability as in e.g., \citet{agarwal2019online,agarwal2019logarithmic,cohen2019learning}.
\begin{definition}
	For $\kappa\geq 1$, $\gamma\in [0,1)$,
	a matrix $A$ is called $ (\kappa, \gamma)$-stable if 
	$\|A^t\|_2\leq \kappa (1-\gamma)^t, \forall  t\geq 0$.\footnote{In some literature, e.g. \cite{agarwal2019online}, this property is  called $(\sqrt \kappa, \gamma)$-strong stability.}
\end{definition}
Consider the following  benchmark policy set:
\begin{align*}
	\mathcal K=\{ &K: (A_*-B_*K)\text{ is }(\kappa, \gamma) \text{-stable,  $ \|K\|_2\leq \kappa$} \\
	& D_x x^K_t \leq d_x, D_u u^K_t\leq d_u, \forall\, t, \forall\,\{w_k \in \mathbb W\}_{k\geq 0}\},
\end{align*}
where $x_t^K, u_t^K$ are generated by policy $u_t=-Kx_t$. 

For any safe learning algorithm/controller $\A$, we measure its performance by `regret' as defined below:
\begin{align*}
	\text{Regret}=\sum_{t=0}^{T-1}l(x^{\A}_t,u^{\A}_t)-T\min_{K\in \K} J(K)
\end{align*}
where $x_t^{\A}, u_t^{\A}$  are generated by the algorithm $\A$ and $J(K)=\lim_{T\to +\infty} \frac{1}{T}\sum_{t=0}^{T-1}\E[l(x^{K}_t,u^{K}_t)]$.

Next, we provide and discuss the assumptions.%  in this paper.

\nbf{Assumptions.} Firstly, though the model $\theta_*$ is not perfectly known, we assume there is some prior knowledge, which is captured by a bounded model uncertainty set $\Theta_{\text{ini}}$ that contains $\theta_*$. It is widely acknowledged that without such prior knowledge,  hard constraint satisfaction is extremely difficult, if not impossible \citep{dean2019safely}. We also assume that $\Theta_{\text{ini}}$ is small enough so that there exists a universal linear controller $u_t=-K_\text{stab}x_t$ to stabilize any  system in $\Theta_{\text{ini}}$. This is a common assumption in constrained LQR with model uncertainties \citep{kohler2019linear,lu2019robust} and $K_\text{stab}$ can be  computed by, e.g., linear matrix inequalities (LMIs) (see \citet{caverly2019lmi} as a review). 

\begin{assumption}\label{ass: Theta0 known}
	There is a known model uncertainty set $\Theta_{\textup{ini}}=\{\theta: \|\theta-\hat \theta_{\textup{ini}}\|_F\leq r_{\textup{ini}}\}$\footnote{The symmetry of  $\Theta_\textup{ini}$ is assumed for  simplicity and not restrictive. We only need $\Theta_{\text{ini}}$ to be small and contain  $\theta_*$.} for some $0<r_{\textup{ini}}<+\infty$ such that (i) $\theta_*\in \Theta_\textup{ini}$, and (ii) there exist $\kappa\geq 1, \gamma\in [0,1)$, and $K_\textup{stab}$ such that for any  $(A, B)\in \Theta_\textup{ini}$, $A-BK_\textup{stab}$ is $(\kappa,\gamma)$-stable.
	%$\Theta^{(0)}= \{(A,B): \|A-\hat A^{(0)}\|_F\leq \epsilon_A^{(0)},\ \|B-\hat B^{(0)}\|_F\leq \epsilon_B^{(0)}\}$ for some $\epsilon_A^{(0)}, \epsilon_B^{(0)}<+\infty$ such that $\theta_*\in \Theta^{(0)}$. We denote $\epsilon_{\theta}^{(0)}=(\epsilon_A^{(0)},\epsilon_B^{(0)})$.
\end{assumption}
Next, to further simplify the technical exposition, we impose a technical assumption  	$K_\textup{stab}=0$, which basically requires that $A$ is open-loop stable for any  $A$ in $\Theta_\textup{ini}$. This assumption can be removed by a standard pre-stabilizing procedure: consider inputs as $u_t=-K_\textup{stab} x_t +v_t$ and focus on designing $v_t$ (see e.g., \citet{agarwal2019logarithmic}).
%We include this assumption here for technical simplicity.
\begin{assumption}[Technical assumption]\label{ass: open loop stable}
	$K_\textup{stab}=0$.
\end{assumption}
\vspace{-5pt}
% Besides, we assume bounded feasible sets of  states and actions, which is also a standard assumption in the constrained LQR literature (cite??).
% \begin{assumption}
% 	The set $\X=\{x:D_x x\leq d_x\}$ is bounded with  $x_{\max}=\max_{x\in \X} \|x\|_2$. The set $\U=\{u:D_u u\leq d_u\}$ is bounded with  $u_{\max}=\max_{u\in \U} \|u\|_2$.
% \end{assumption}

Further, we  need to assume  a feasible linear policy exists for our constrained LQR \eqref{equ: J(pi)}, otherwise our regret benchmark is not well-defined. Here, we impose a slightly stronger assumption of strict feasibility to allow   approximation errors in our control design. This assumption can also be verified by LMIs \citep{caverly2019lmi}. 
\begin{assumption}\label{ass: KF epsilonF}
	There exists $K_F\in \mathcal K$ and $\epsilon_{F,x}>0,\epsilon_{F,u}>0$ such that  $D_x x_t^{K_F}\leq d_x -\epsilon_{F,x}\one_{k_x} $ and $D_u u_t^{K_F}\leq d_u -\epsilon_{F,u}\one_{k_u} $ for all $t\geq 0$ under all  $w_k \in \W$.
\end{assumption}

Lastly, we impose assumptions on disturbance $w_t$. We define a certain anti-concentration property around 0 as in \citep{abeille2017linear}, which essentially requires a random vector $X$ to have  large enough probability  at a distance from 0 on all directions.% and there are positive uniform lower bounds on the distances and probabilities \citep{abeille2017linear} (cite??). 

\begin{definition}[Anti-concentration]\label{def: w anticoncentration}
	A random vector $X\in \R^n$ satisfies  $(s, p)$-anti-concentration  for some $s>0, p\in (0,1)$ if  $\Pb(\lambda^\top w\geq s)\geq p$  for any  $\|\lambda\|_2=1$.
\end{definition}
%Notice that this definition essentially requires $X$ has positive probability on all directions and there is a positive lower bound on the probability of each direction. 

%\red{how to explain this??how to explain this? stats people will ask a lot about this, how does this compare with moments bound? what are stat literature citations?} \blue{Lina: let's not worry about it now...}
We assume anti-concentrated $w_t$ below to provide excitation for learning the model,  which  is crucial for our estimation error bound under general policies.
\begin{assumption}\label{ass: w Sigma}
	$w_t \in \mathbb W$ is i.i.d., $\sigma_{sub}^2$-sub-Gaussian,   zero mean, and  $(s_w, p_w)$-anti-concentration.\footnote{Notice that by $\mathbb W=\{w: \|w\|_\infty \leq w_{\max}\}$, we have $\sigma_{sub}\leq \sqrt n w_{\max}$. .} 
\end{assumption}

%\subsection{Preliminaries: Approximate linear policies with Disturbance-action Policies}

\begin{mycitation}
Further,

the first few stages.

 when learning starts, if not impossib

 let $x_t^{\A}, u_t^{\A}$ denotes the

For clarification, we formalize the definition of safety in this paper.
\begin{definition}[Safety]
	We call an algorithm to be 
	\begin{enumerate}
		\item[(i)] \textup{safe} if $x_t\in \X, u_t\in \U$  for all $t$ and all $w_k\in \W$ under the true system $\theta_*$; 
		\item[(ii)] \textup{$\epsilon$-strictly safe}, where $\epsilon=(\epsilon_x, \epsilon_u)>0$, if  $D_x x_t \leq d_x -\epsilon_x \one_{k_x}, D_u u_t \leq d_u -\epsilon_u \one_{k_u}$ for all $t$ and all $w_k\in \W$ under the true system $\theta_*$;
		\item[(iii)] \textup{robustly safe} on  an uncertainty set $\Theta$ if $x_t\in \X, u_t\in \U$  for all $t$ and all $w_k\in \W$ under any  system with parameter $\theta\in \Theta$.
	\end{enumerate}

\end{definition}

The focus of this paper is important in many applications, e.g. optimizing the fuel consumption while maintaining the speed at desirable range
% a desirable range

 In practice, one can first apply other algorithms, e.g. RMPC (cite??), to safely steer the system to around 0, then apply our algorithm to optimize the performance around 0. 

This paper will focus on

 Hence, this paper follows the literature (cite??) and focuses on the optimal linear policy for problem \eqref{equ: J(pi)}.

follows the literature (cite??) and

\red{where to explain: robust constraint satisfaction, but expected cost. Also, we can do min-max, but we need the randomness to explore. this problem is borrowed from Sarah Dean.}

Consider the following optimal control with a constrained linear system and bounded disturbances.

\begin{equation}\label{equ: J(pi)}
\begin{aligned}
	\min_{\pi}  &\ J(\pi)= \lim_{T\to +\infty} \frac{1}{T}\sum_{t=0}^{T-1}\E l(x_t, u_t)\\
	\text{s.t. } & x_{t+1}=A_*x_t+B_* u_t +w_t, \ \forall \, t\geq 0,\\
	&D_x x_t \leq d_x,  D_u u_t\leq d_u, \, \forall  \{w_t\!:\!\|w_t\|_\infty \leq w_{\max}\}
\end{aligned}
\end{equation}
where $\pi$ denotes a control policy,  $J(\pi)$ denotes the infinite-horizon averaged cost of $\pi$, $l(x_t,u_t)=x_t^\top Q x_t +u_t^\top R u_t$ for positive definite $Q$ and $R$,  $x_t\in \R^n, u_t\in \R^m, d_x\in \R^{k_x}, d_u\in \R^{k_u}$, $x_0$ is given. 	We define some shorthand notations: $\mathbb X:=\{x: D_x x\leq d_x\}$, $  \mathbb U:=\{u: D_u u\leq d_u\}$,
$\mathbb W:=\{w\in \R^n : \|w\|_\infty \leq w_{\max}\}$,
$ \theta_*:=(A_*, B_*)$, $ \theta :=(A, B)$. We consider bounded $\mathbb X$ and $\mathbb U$, i.e., there exist $x_{\max}, u_{\max}$ such that
$\|x\|_2 \leq x_{\max}, \forall x\in \mathbb X, \|u\|_2\leq u_{\max}, \forall u\in \mathbb U$.  For simplicity, we consider $x_0=0$ and discuss non-zero $x_0$ in the supplementary file.\footnote{Roughly, if $x_0$ is sufficiently small such that it admits a safe linear controller, then our algorithm can directly be applied. If $x_0$ is large, we leverage  RMPC in \cite{mayne2005robust} to steer the state to be small enough. For too large $x_0$, the constrained control can be infeasible.} 
For ease of illustration and discussion, we define four versions of safety.
\begin{definition}[Safety]
We call an algorithm to be 
\begin{enumerate}
	\item[(i)] \textup{safe} if $x_t\in \X, u_t\in \U$  for all $t$ and all $w_k\in \W$ under the true system $\theta_*$; 
	\item[(ii)] \textup{$\epsilon$-strictly safe}, where $\epsilon=(\epsilon_x, \epsilon_u)>0$, if  $D_x x_t \leq d_x -\epsilon_x \one_{k_x}, D_u u_t \leq d_u -\epsilon_u \one_{k_u}$ for all $t$ and all $w_k\in \W$ under the true system $\theta_*$;
	\item[(iii)] \textup{robustly safe} on  an uncertainty set $\Theta$ if $x_t\in \X, u_t\in \U$  for all $t$ and all $w_k\in \W$ under any  system with parameter $\theta\in \Theta$; 
%	\item[(iv)] \textup{$\epsilon$-strictly robustly safe} on set $\Theta$ if $D_x x_t \leq d_x -\epsilon_x \one_{k_x}, D_u u_t \leq d_u -\epsilon_u \one_{k_u}$ for all $t$ and all $w_k\in \W$ under any  system with parameter $\theta\in \Theta$.
\end{enumerate}
%(i) \textup{safe} if $x_t\in \X, u_t\in \U$  for all $t$ and all $w_k\in \W$ under the true system $\theta_*$; 
%(ii) \textup{$\epsilon$-strictly safe}, where $\epsilon=(\epsilon_x, \epsilon_u)>0$, if  $D_x x_t \leq d_x -\epsilon_x \one_{k_x}, D_u u_t \leq d_u -\epsilon_u \one_{k_u}$ for all $t$ and all $w_k\in \W$ under the true system $\theta_*$;
%(iii) \textup{robustly safe} on a model uncertainty set $\Theta$ if $x_t\in \X, u_t\in \U$  for all $t$ and all $w_k\in \W$ under any  system with parameter $\theta\in \Theta$; and (iv) \textup{$\epsilon$-strictly robustly safe} on $\Theta$ if $D_x x_t \leq d_x -\epsilon_x \one_{k_x}, D_u u_t \leq d_u -\epsilon_u \one_{k_u}$ for all $t$ and all $w_k\in \W$ under any  system with parameter $\theta\in \Theta$.

\end{definition}

%This definition is also closely related with the strong stability in (cite??). Compared with strong stability, we don't require an explicit decomposition, but our definitions are equivalent. Notice that for any asymptotically stable matrix $A$, there exists $(\kappa, \gamma)$ such that $A$ is $ (\kappa, \gamma)$-stable. \lina{I am not sure about the meaning of ``strong'' stability here. I think based on your last sentence, this stability is equivalent to asymptotically stability for matrix A. So revise this paragraph to make it more clear?}

To demonstrate theoretical rigor, we introduce a quantitative version of matrix stability.
\begin{definition}
	For $\kappa\geq 1$, $\gamma\in [0,1)$,
	a matrix $A$ is called $ (\kappa, \gamma)$-stable\footnote{In some literature, e.g. \cite{agarwal2019online}, this property is  called $(\sqrt \kappa, \gamma)$-strong stability.} if 
	 $\|A^t\|_2\leq \kappa (1-\gamma)^t, \forall  t\geq 0$.
\end{definition}

%\red{not finished, }

In this work, we consider that the system parameters $\theta_*=(A_*, B_*)$ are unknown but  the constraints ($\X, \U, \W$) and cost functions  ($Q, R$,) are known and $x_t$ can be observed. Though the true model $(A_*,B_*)$ is unknown, we assume some prior knowledge on the system dynamics is available, i.e., a model uncertainty set $\Theta_{\text{ini}}$ that satisfies the following assumption.

\red{add more explanation and discussion to the assumptions}
\begin{assumption}\label{ass: Theta0 known}
	There is a known model uncertainty set $\Theta_{\textup{ini}}=\{\theta: \|\theta-\hat \theta_{\textup{ini}}\|_F\leq r_{\textup{ini}}\}$\footnote{Here, $\Theta_\textup{ini}$ is symmetric on all directions, which may not be the case in practice. This is not  restrictive and only assumed for technical simplicity. What we really need is that $\Theta^{(0)}$ is a  compact set containing $\theta_*$.} for some $0<r_{\textup{ini}}<+\infty$ such that (i) $\theta_*\in \Theta_\textup{ini}$, and (ii) there exist $\kappa\geq 1, \gamma\in [0,1)$ such that for any  $(A, B)\in \Theta_\textup{ini}$, $A$ is $(\kappa,\gamma)$-stable.
	%$\Theta^{(0)}= \{(A,B): \|A-\hat A^{(0)}\|_F\leq \epsilon_A^{(0)},\ \|B-\hat B^{(0)}\|_F\leq \epsilon_B^{(0)}\}$ for some $\epsilon_A^{(0)}, \epsilon_B^{(0)}<+\infty$ such that $\theta_*\in \Theta^{(0)}$. We denote $\epsilon_{\theta}^{(0)}=(\epsilon_A^{(0)},\epsilon_B^{(0)})$.
\end{assumption}
%Since $\Theta_\text{ini}$ is compact, we denote $\kappa_B\geq 1$ such that $\kappa_B=\max_{(A, B)\in \Theta_\text{ini}} \|B\|_2$.
\begin{remark}
Assumption (i) is standard in the literature \cite{dean2018regret,mayne2005robust}. Though assumption (ii) seems restrictive, with a pre-stabilization trick (cite??), this condition can be relaxed to a common assumption in the constrained LQR literature \cite{kohler2019linear,lu2019robust}, i.e. (iii): there exists $ K_{pre}$ such that $A-B K_{pre}$ is $(\kappa,\gamma)$-stable  for any  $(A,B)\in \Theta_\text{ini}$. Essentially, Assumption (iii) requires a small enough model uncertainty set such that any system in this set can be stabilized by a common  policy matrix $K_{pre}$.  The  matrix $K_{pre}$ can be computed by e.g., linear matrix inequalities (LMI)  \cite{kohler2019linear}. The pre-stabilization trick is by considering $u_t=-K_{pre} x_t + v_t$ and focusing on choosing $v_t$, which leads to a closed-loop system $x_{t+1}=(A-BK_{pre})x_t+Bv_t+w_t$ with $(A-BK_{pre})$ satisfying the assumption (ii). This paper imposes Assumption (ii) instead of (iii) for technical simplicity. Most results can be obtained under Assumption (iii) with more complicated analysis.
\end{remark}

We note that although $A_*$ is stable, implementing zero control may not be safe, i.e., violating the constraints, so it calls for 
%%% add "a"
a more careful control design to ensure constraint satisfaction.

Next, we impose assumptions on disturbance $w_t$. We introduce anti-concentration property \cite{abeille2017linear}.
\begin{definition}[Anti-concentration]\label{def: w anticoncentration}
	A random vector $X\in \R^n$ is said to satisfy $(s, p)$-anti-concentration properties for some $s>0, p\in (0,1)$ if  for any $\lambda\in \R^n$, $\|\lambda\|_2=1$, $\Pb(\lambda^\top X\geq s)\geq p.$
\end{definition}
Notice that this definition essentially requires $X$ has positive probability on all directions and there is a positive lower bound on the probability of each direction. 

\red{explain what this means on the covariance matrix.}

\begin{assumption}\label{ass: w Sigma}
	$w_t \in \mathbb W$ is i.i.d., $\sigma_{sub}^2$-sub-Gaussian,   zero mean, and  $(s_w, p_w)$-anti-concentration.\footnote{Notice that by $\mathbb W=\{w: \|w\|_\infty \leq w_{\max}\}$, we have $\sigma_{sub}\leq \sqrt n w_{\max}$. .} 
\end{assumption}

%	\red{Q: this anti-concentration, can it be proved by $\Sigma_w\geq \alpha_w I_n$?? Think about this later, our goal is to make assumptions look less restrictive by writing.
%		Note: I think I can do it, I don't need anti-concentration because I can prove it by $\Sigma_w\geq \alpha_w I_n$ and $\E w=0$. $\eta$ can also be any distribution bounded, p.d. variance, zero mean.}
%	%Notice that many distributions satisfy anti-concentration properties, e.g. truncated gaussian, uniform distribution on a ball, etc.
%	

%
%Next, we assume the disturbance $w_t$ follows truncated Gaussian distribution.
%	\begin{assumption}[Technical Assumption on $w_t$]
%	
%	For  ease of analysis, we  assume $w_t$ follows truncated Gaussian with zero mean, variance $\sigma_w^2 I_n$ and infinity norm bound $w_{\max}$. 
%	
%	
%	(detailed definition for  truncated gaussian: $w_{t,i}=\mathring w_{t,i}\mid \|\mathring w\|_\infty \leq w_{\max}$, where $\mathring w_{t,i}$ i.i.d. for all $t,i$ with Gaussian/normal distribution $\mathcal N(0, \sigma_{\mathring w}^2)$.)
%	
%\end{assumption}
%
%

%\red{lina: I move this sentence here as a transition to discuss the benchmark policy etc}

\red{add discussion on effect of $x_0$ and how we choose our benchmark with a different $x_0$!}

\red{in the main text, we consider initial state is x0=0, in the supplementary, we will discuss non-zero x0. Specifically, our benchmark is still x0=0, because it is meaningful, and because it has good connection with RMPC. Otherwise, we just do similar things as x0=0, every step: solve the same x0.}

\red{in DAP-constraint satisfaction: we have x0=0 in AAAI paper, and we explain when x0 not=0, we impose constraints on the transient.}

\red{comparison with Sarah, add a remark: 1).. use FIR on ut to represent xt, instead of using FIR xt and impose a equality constraints. 2).. directly implement FIR, instead of solving dynamical controller K. Similar to Sarah's FIR formulation. We do this because, ut does not depend on all the history, but a finite memory, easier to adapt, while u=Kx more complicated at transient. Note: this transient additional constraint also shows up in their FIR formulation, see equation (??).}
	
\red{Note: our selling point, we have an adaptive way to achieve optimal linear policy, and our linear policy design is different from Sarah's method, because we want to have cleaner truncation/forgetting pattern. That being said, we think our result can still be applied to Sarahs, with more careful and involved design.}

\red{transient: constraint violation: because new initial state is no longer x0. One method to handle this is by imposing additional constraints based on the new intiial state. This means we need to revise the constraints on the transient, and this is what is done in Sarah, but in an adaptive way,  but this may cause infeasibility, as mentioned in (cite??). We provide an example below.}

\blue{ for nonzero initial state, we need to bound transient states, here, we still treat wt=0 for t<0, some times people define w-1=x0, we borrow their idea, both are okay.}
\nbf{Regret definition and Benchmark Policy.} 
%\lina{we need to rewrite this part to make the story flow. Plus I think we should do something similar to our previous NeurIPS papers: summerize main steps of our learning problem and then introduce the regret definition. One suggestion is below.}
% Consider an infinite-horizon optimal control with robust constraint satisfaction as described below:
% \begin{align}\label{equ: J(pi)}
% \min_{\pi} J(\pi)=\lim_{T\to +\infty} \frac{1}{T}\sum_{t=0}^{T-1}\E l(x_t^{\pi},u_t^{\pi}), \quad \text{ subject to } x_t^{\pi}\in \X, \ u_t^{\pi}\in \U, \ \forall\, \{w_k\in \ W\}_{k\geq 0}.
% \end{align}
% Even with known model $\theta_*$, how to solve \eqref{equ: J(pi)} efficiently is still an open problem. Many attempts have been made to reformulate \eqref{equ: J(pi)} as a tractable problem by considering certain policy classes and optimizing the policy parameters, for example, \cite{sarah...} considers linear controllers with memory, and robust MPC \cite{} considers piecewise affine controllers without memory. 
In this paper, we aim to design an adaptive algorithm $\A$ that learns the system parameters and updates the control policies in an online fashion to improve the system performance and guarantee constraint satisfaction for all $t\geq 0$ under any $w_t\in \W$. Solving  \eqref{equ: J(pi)} efficiently with constraints $\mathbb X$ and $\mathbb U$ is still an open problem even when the model $\theta_*$ is known. People reformulate \eqref{equ: J(pi)} as a tractable problem by limiting the policy class. For example, \cite{dean2019safely} considers linear controllers with memory, and robust MPC  considers piece-wise affine controllers without memory \cite{mayne2005robust}. In this paper, we will consider disturbance-action policy (DAP) control as introduced later in Section~\ref{sec: prelim}. 

To evaluate the performance of the online algorithm $\A$, besides ensuring safety, we analyze the regret of the algorithm $\A$ while bench-marking it with some optimal safe policy obtained by assuming $\theta_*$ to be known. Given the difficulty in designing safe policies even with known $\theta_*$,  we consider static linear policy class as our benchmark policies in this paper. But our results can be extended to other benchmarks such as linear policy with memory and a certain type of RMPC with PWA controllers. We discuss these extensions in the supplementary while leaving more general benchmarks as future work. In particular,
we consider $u_t=-K x_t$ as benchmark. 
Define 
\begin{align*}
	\mathcal K=\{ & \ K: (A_*-B_*K)\text{ is }(\kappa, \gamma) \text{ stable and } \\
	& \ \|K\|_2 \leq \kappa, x_t 
	%%%% should this be x_t^K or x_t
	\in \mathbb X, u_t \in \mathbb U, \forall\{w_k \in \mathbb W\}_{k\geq 0}\}.
\end{align*} Let $J^*=\min_{K\in \mathcal K} J(K)$ denote the optimal control cost provided by policy $\K$ when the model is known. We measure the performance of online algorithm $\A$ by policy regret, which is defined as
\vspace{-4pt}
\begin{center}
$\text{Regret}=\sum_{t=0}^{T-1} l(x^{\A}_t,u^{\A}_t)-TJ^*$
\end{center}

To make the regret well-defined, we need to assume that $\mathcal K$ is not empty. For technical reasons, we will impose a stronger assumption that there exists a strictly robustly safe  controller.
\begin{assumption}\label{ass: KF epsilonF}
	There exists $K_F\in \mathcal K$ and $\epsilon_F=(\epsilon_{F,x},\epsilon_{F,u})>0$ such that  $K_F$ is $\epsilon_F$-strictly robustly safe on model uncertainty set $\Theta$.
\end{assumption}
A sufficient condition to verify the existence of $K_F$ is by LMI reformulation. In the supplementary file, we provide one such reformulation.

\begin{remark}
	Unlike unconstrained LQR, the initial state $x_0$ affects the optimal linear controller to the constrained LQR. This is because a contr
	
\end{remark}

\red{add a brief review of DAP}

\end{mycitation}

\section{Preliminaries}\label{sec: prelim}
%Note: use $g_i^x, g_j^u$ for time-invariant $\bf M$, and $\tilde g_i^x$ for time-varying $\bf M$? No, not a good idea.
%\red{cut to 1 page}

%\red{Q: how to compare DAP and SLS?}

%\red{explain why we use DAP: for computation efficiency, another way is SLS, closely related, but SLS recovers K, DAP does not. Future work to extend to SLS.}

This section reviews the existing results in \cite{li2020online} on how to approximate the optimal safe linear policy computationally efficiently and how to do safe policy updates when assuming a known model.
\subsection{Approximation of Optimal Linear Polices with Disturbance-action Policies}
When the model is known, a computationally efficient way to approximate the optimal linear policy for problem \eqref{equ: J(pi)} is to rewrite the action as a linear combination of the  history disturbances, i.e., 
\vspace{-5pt}
\begin{equation}\label{equ: def DAP}
	u_t=\sum_{k=1}^H M[k]w_{t-k},
\end{equation}  and solving the optimal policy matrices $\bm M=\{M[1], \dots, M[H]\}$ by a linearly constrained quadratic program. When the model is known, we can compute the disturbance by $w_t=x_{t+1}-A_* x_t -B_* u_t$. The policy \eqref{equ: def DAP} is called a disturbance-action policy (DAP) in \citet{agarwal2019logarithmic,li2020online} and $H$ is called the policy's memory length.\footnote{\eqref{equ: def DAP} is also called 
	 finite-impulse response in \citet{dean2019safely} and affine-disturbance policy in \citet{mesbah2016stochastic}.}
 There are different ways to construct the linearly constrained quadratic program (see e.g., \citep{li2020online,dean2019safely,mesbah2016stochastic}). This paper adopts the method in \citet{li2020online}, which is briefly reviewed below.

 Under DAP, the state  can be approximated by $\tilde x_t$, which is an affine function of $\mb M$ as below.
 \begin{proposition}[\cite{agarwal2019online}]\label{prop: state formula with known system}
	Under a time-invariant policy $\mb M$, we have
	$x_t=A_*^H x_{t-H}+\tilde x_t(\mb M;\theta_*)$, where $ \tilde x_t(\mb M;\theta_*)=\sum_{k=1}^{2H} \Phi_k^x(\mb M;\theta_*)w_{t-k}$ and
 $
	 \Phi_k^x(\mb M;\theta_*)= A_*^{k-1} \id_{(k\leq H)}+\sum_{i=1}^H A_*^{i-1} B_* M[k-i] \id_{(1\leq k-i \leq H)}
	$. 

\end{proposition}

\nbf{Safe policy set.} \citet{li2020online} reformulated the linear constraints on the states and actions to polytopic constraints on the policy parameters as follows.
\begin{align}
	&g_i^x(\mb M;\theta_*)\leq d_{x,i}\!-\!\epsilon_x,\ \forall 1\leq i \leq k_x \label{equ: gix}\\
	&g_j^u(\mb M)\leq d_{u,j}-\epsilon_u, \forall 1\leq j \leq k_u\label{equ: giu}
\end{align}
where $g_i^x(\mb M;\theta_*)$ and $g_j^u(\mb M)$ represents each line of the state and action constraints as follows:
\begin{align*}
		g_i^x(\mb M;\theta_*)& \coloneqq \sup_{w_k\in \W} D_{x,i}^\top \tilde x_t(\mb M;\theta_*)\\
		&=\sum_{k=1}^{2H}\|D_{x,i}^\top \Phi_k^x(\mb M;\theta_*)\|_1w_{\max}\\
		g_j^u(\mb M)&\coloneqq\!\sup_{w_k\in \W} D_{u,j}^\top u_t\!=\! \sum_{k=1}^{H}\|D_{u,j}^\top M[k]\|_1w_{\max}.
\end{align*}
Notice that $g_i^x(\cdot)$  is defined based on the approximate state $\tilde x_t$ in Proposition \ref{prop: state formula with known system}, so a constraint-tightening term $\epsilon_x \geq \epsilon_{H}(H)=\|D_x \|_\infty \kappa x_{\max}(1-\gamma)^{H}$ is introduced in \eqref{equ: gix} to account for the approximation error. $g_j^u(\cdot)$ is defined on the actual action $u_t$ so  $\epsilon_u=0$ here.

%For notation simplicity, we define

In summary, \citet{li2020online} construct a safe policy set:\footnote{The policy set \eqref{equ: Omega def} includes safe DAPs starting from $x_0=0$. With a non-zero $x_0$, additional linear constraints should be imposed to ensure state constraint satisfaction for $t<H$ because Proposition \ref{prop: state formula with known system} only holds for $t\geq H$ with a nonzero $x_0$. Due to the additional constraints, only small $x_0$ can ensure the existence of a  safe DAP/linear policy. For larger $x_0$, other types of policies should be considered to ensure safety, e.g.,  \cite{mayne2005robust}. For too large $x_0$, it is possible that no policy can ensure safety. See e.g., \cite{rawlings2009model}, for more discussions.\label{footnote: generalize small nonzero x0}} 
%the state approximation in Proposition \ref{prop: state formula with known system}  for $t<H$ should be revised to incorporate 
\begin{align}\label{equ: Omega def}
    \Omega(\theta_*, \epsilon_x,\epsilon_u)=\{\mb M\in \M_H: \eqref{equ: gix}, \eqref{equ: giu}.\}, 
\end{align}
where $\theta_*$ is  used in \eqref{equ: gix}, 
$\epsilon_x\geq \epsilon_H(H), \epsilon_u\geq 0$, and  $	\M_H=\{\mb M: \|M[k]\|_\infty \leq 2\sqrt n\kappa^2(1-\gamma)^{k-1},  \forall\, 1\leq k\leq H\}
$ is included for technical simplicity without losing generality. Notice that set \eqref{equ: Omega def} is a polytope.

\nbf{Quadratic program (QP) for optimal safe DAP.} The optimal safe DAP $\mb M^*$ can be solved  by QP below:
\begin{equation}\label{equ: optimal DAP}
	\begin{aligned}
		&\min_{\mb M} f(\mb M; \theta_*)\\
		\text{s.t. } &\mb M\in \Omega( \theta_*, \epsilon_x^*, \epsilon_u^*),\epsilon_x^*= \epsilon_H(H), \epsilon_u^*=0.
	\end{aligned}
\end{equation}
where  $ f(\mb M; \theta_*)=\E\left[ l(\tilde x_t(\mb M, \theta_*), u_t\right]$ is defined by the expected cost of the approximate state and the action, which is a quadratic convex function of $\mb M$.  \citet{li2020online} show that the optimal safe DAP $\mb M^*$ approximates the optimal safe linear policy $K^*$  in the sense of $J(\mb M^*)-J(K^*)\leq \tilde O(1/T)$ for large enough $H$.

\begin{mycitation}
%guarantees safety and its induced cost $J(\bm M_*)$ is close to the cost of the optimal linear policy $J(K^*)$  in the sense of $J(\mb M^*)-J^*\leq \tilde O(1/T)$ for large enough $H$.

%on $\mb M$ $f(\mb M;\theta_*)$ by the expected cost of  approximate state and action: $ f(\mb M; \theta_*)=\E_{w_k}\left[ l(\tilde x_t(\mb M, \theta_*), u_t(\mb M))\right]$, which is a convex quadratic  function of $\mb M$.
\citet{li2020online} constructs a polytopic safe policy set by establishing safety constraints on the  policy parameters $\mb M$ below.
\begin{align}
    \Omega(\theta_*, \epsilon_x, \epsilon_u)
\end{align}

for safe implementation. For

safe policy set by using Proposition 

By imposing the constraints on the approximate state $\tilde x_t$, \citet{li2020online} establishes state constraints on the policy $\mb M$:

%\red{to do: write the review based on what is needed in our algo.}

%\red{please rewrite the explanation below, to emphasize what is needed to be known, and leave out unnecessary details.}
 %on the  parameters $\bm M=\{M[1], \dots, M[H]\}$. There are different ways to formulate the quadratic programs \citep{li2020online,dean2019safely}. This paper adopts the method in \citet{li2020online}, which is briefly reviewed below.

%and briefly reviewed below. 

%Our  adaptive control is built upon disturbance-action policy (DAP) \cite{agarwal2019online}. To introduce our  algorithm, we provide a review on DAP and its application to constrained control with \textit{known} model, which will be critical for developing our online algorithm.

%\nbf{Disturbance-action policies  and their properties.}
%\begin{definition}[Disturbance-action control (DAP)]
%	Consider  memory length $H\geq 1$ and  policy parameters $\mb M=\{M[k]\}_{k=1}^H$  for $M[k]\in\R^{m \times n}$, DAP selects
%$		u_t=\sum_{k=1}^H M[k]w_{t-k}, 
%$	where  $w_{t}=x_{t+1}-A_*x_t-B_*u_t$ can be computed when $\theta_*$ is known and  $w_{t}=0$ for $t<0$.\footnote{DAP is also called 
%%
%%??? add "the"?
%the finite-impulse response \cite{dean2019safely} or affine-disturbance policy \cite{mesbah2016stochastic}.}
%\end{definition}

As in \cite{li2020online}, we  will work with a convex polytopic constraint set $\mathcal{M}_H$ on admissible $\mb M$ for technical simplicity:
$	\M_H=\{\mb M: \|M[k]\|_\infty \leq 2\sqrt n\kappa^2(1-\gamma)^{k-1}, \ \forall\, 1\leq k\leq H\}
$. Notice that $u_t$ is a linear function with respect to $\mb M$. Further, the next proposition shows that $x_t$ can be approximated by $\tilde x_t(\mb M;\theta_*)$, which is an affine function on $\mb M$. 

\begin{proposition}[\cite{agarwal2019online}]\label{prop: state formula with known system}
	When implementing time-invariant policy $\mb M$, we have
	$x_t=A_*^H x_{t-H}+\tilde x_t(\mb M;\theta_*)=A_*^H x_{t-H}+\sum_{k=1}^{2H} \Phi_k^x(\mb M;\theta_*)w_{t-k}$, 
	where $
	 \Phi_k^x(\mb M;\theta_*)= A_*^{k-1} \one_{(k\leq H)}+\sum_{i=1}^H A_*^{i-1} B_* M[k-i] \id_{(1\leq k-i \leq H)}
	$. 
%	More generally, consider time-varying policies $u_t= \sum_{k=1}^H M_t[k]w_{t-k}$, then we have
%$
	%	x_t=A_*^H x_{t-H}+\tilde x_t(\mb M_{t-H:t-1};\theta_*)	=A_*^H x_{t-H}+\sum_{k=1}^{2H}\tilde \Phi_k^x(\mb M_{t-H:t-1};\theta_*)w_{t-k}
% $,
% 	where
% 	$
% 	\tilde \Phi_k^x(\mb M_{t-H:t-1};\theta_*)= A_*^{k-1} \one_{(k\leq H)}+\sum_{i=1}^H A_*^{i-1} B_* M_{t-i}[k-i] \one_{(1\leq k-i \leq H)}
% 	$. 

%	(\mb M;\theta_*)$ and approximate state $\mathring{x}_t(\mb M;\theta_*)=	\sum_{k=1}^{2H}\mathring \Phi_k^x(\mb M;\theta_*)w_{t-k}$, where $\mathring \Phi_k^x(\mb M;\theta_*)=\Phi_k^x(\mb M;\theta_*)$.\lina{there were typos in the proposition. I corrected some. please check whether the fix is right or not}
\end{proposition}
Based on DAP and Proposition \ref{prop: state formula with known system}, we introduce the quadratic problem formulation to approximates the optimal linear policy for problem \eqref{equ: J(pi)}.

\nbf{Quadratic cost functions.} For fixed $\mb M$, define $f(\mb M;\theta_*)$ by the expected cost of  approximate state and action: $ f(\mb M; \theta_*)=\E_{w_k}\left[ l(\tilde x_t(\mb M, \theta_*), u_t(\mb M))\right]$, which is a convex quadratic  function of $\mb M$. %Since $l(x_t, u_t)$ is quadratic

\nbf{Polytopic constraints.} \citet{li2020online} reformulated the linear constraints on the states and actions to polytopic constraints on the policy parameters as follows.
\begin{align}
	&g_i^x(\mb M;\theta_*)\leq d_{x,i}\!-\!\epsilon_x,\ \forall 1\leq i \leq k_x \label{equ: gix}\\
	&g_j^u(\mb M)\leq d_{u,j}-\epsilon_u, \forall 1\leq j \leq k_u\label{equ: giu}
\end{align}
where $g_i^x(\mb M;\theta_*)$ and $g_j^u(\mb M)$ represents each line of the state and action constraints as follows:
\begin{align*}
		g_i^x(\mb M;\theta_*)& \coloneqq \sup_{w_k\in \W} D_{x,i}^\top \tilde x_t(\mb M;\theta_*)\\
		&=\sum_{k=1}^{2H}\|D_{x,i}^\top \Phi_k^x(\mb M;\theta_*)\|_1w_{\max}\\
		g_j^u(\mb M)&\coloneqq\sup_{w_k\in \W} D_{u,j}^\top u_t(\mb M)\\
		&= \sum_{k=1}^{H}\|D_{u,j}^\top M[k]\|_1w_{\max}.
\end{align*}
Notice that $g_i^x(\cdot)$ above is defined based on the approximate state $\tilde x$, so a constraint-tightening term $\epsilon_x$ in \eqref{equ: gix} is introduced to account for the approximation error.  \citet{li2020online} shows that it is sufficient to let $\epsilon_x \geq \epsilon_{H}(H)=\|D_x \|_\infty \kappa x_{\max}(1-\gamma)^{H}$. There is no action  approximation here so we can let $\epsilon_u=0$ here.

In summary, the optimal DAP to our problem \eqref{equ: J(pi)} can be approximated by the quadratic program  below.
\begin{align}\label{equ: time-invariant optimal DAP}
	\min_{\bm M\in \M_H} f(\bm M; \theta_*),\, \text{s.t.\,\eqref{equ: gix},\eqref{equ: giu},\,$\epsilon_x\!=\!\epsilon_{H}(H), \epsilon_u\!=\!0$.}
\end{align}
As shown in \citet{li2020online},   the optimal DAP $\bm M_*$ to \eqref{equ: time-invariant optimal DAP} is safe to implement  and its induced cost $J(\bm M_*)$ is close to the cost of the optimal linear policy $J(K^*)$  in the sense of $J(\mb M^*)-J^*\leq \tilde O(1/T)$ for large enough $H$. 
\end{mycitation}
\vspace{-5pt}
\subsection{Safe Policy Updates with  Known Models}\label{subsec: slow variation trick}
\begin{myspecial}
\red{Q: should I add a counter-example here? maybe.}
\end{myspecial}

Notice that the constraints in \eqref{equ: optimal DAP}  can only guarantee safety when a feasible policy $\mb M$ is implemented in a time-invariant fashion. It is known that time-varying policies $\{\mb M_t\}_{t\geq 0}$ may still violate the  constraints even if each $\bm M_t\in \Omega(\theta_*, \epsilon_x^*, \epsilon_u^*)$ \citep{li2020online}. 
Intuitively, this is because the   approximate state $ \tilde x_{t+1}$ is affected by not only $\mb M_t$ but also  $\mb M_{t-1},\mb M_{t-2}, \dots$, causing the safety constraints  coupled across   the history policies (see  \cite{li2020online} for more details).

%$H$-step history policies as shown in \citet{agarwal2019logarithmic,li2020online}, thus introducing  coupling constraints on  policy parameters in $H$ stages, which is not taken care of by the single-stage constraints  in \eqref{equ: time-invariant optimal DAP}. %$\Omega(\epsilon_H(H), \theta_*)$.
 
 When the model is known, \citet{li2020online} tackle the  coupled constraints by a \textit{slow-variation} trick. Roughly, this trick indicates that safety of a slowly varying sequence of policies $\{\mb M_t\}$ can be guaranteed if each $\mb M_t$ belongs to the set \eqref{equ: Omega def} with an additional constraint tightening term $\epsilon_v$ to allow  small policy variations.  
 %sequence to vary slowly and approximates the history policies by the current policy. By properly tightening the constraints in \eqref{equ: Omega def} to account for the policy variation, safety can be guaranteed for such slowly varying policies by the following lemma.
 %and, roughly speaking, approximates the $H$-step history policies  with the policy at the current stage (see \citet{li2020online} for more details). By  tightening the constraints in \eqref{equ: gix} with an additional term $\epsilon_v$ to account for  policy variation errors, it is shown below that the single-stage constraints \eqref{equ: gix}, \eqref{equ: giu} are sufficient to guarantee the safety of slowly varying DAPs.
\begin{lemma}[Slow variation trick with perfect model \citep{li2020online}]\label{lem: constraint satisfaction of slow varying DAP}
	Consider a slowly varying  DAP sequence $\{\mb M_t\}_{t\geq 0}$ with %
	%add comma
	%
	$\|\mb M_t-\mb M_{t-1}\|_F\leq \Delta_M$, where $\Delta_M$ is called the policy variation budget.  $\{\mb M_t\}_{t\geq 0}$ is safe to implement  if $\mb M_t\in \Omega(\theta_*, \epsilon_x, \epsilon_u)$ for all $t\geq 0$, where
	\begin{align*}
		\epsilon_x \geq \epsilon_H(H)+ \epsilon_v(\Delta_M,H), \ \epsilon_u\geq 0,
	\end{align*}
and $	\epsilon_{v}(\Delta_M, H)\!=\! \sqrt{mnH}\Delta_{M}\cdot \|D_x\|_\infty w_{\max}\frac{\kappa \kappa_B}{\gamma^2}$,  $\kappa_B=\max_{(A, B)\in \Theta_{\text{ini}}} \|B\|_2.$.%,c_1=\|\!D_x\!\|_\infty w_{\max}\frac{\kappa \kappa_B}{\gamma^2}$.
%	
%	 for each $t\geq 0$, we have $\mb M_t\in \Omega(\vec \epsilon, \theta_*, H)$ for
%	$$ \epsilon_x \geq \epsilon_H(H)+\epsilon_v(\Delta_M,H), \ \epsilon_u\geq 0$$
%	where 
%$	\epsilon_{v}(\Delta_M, H)= c_2\sqrt{mnH}\Delta_{M}$ for $c_2=\|D_x\|_\infty w_{\max}\kappa \kappa_B/\gamma^2$, then it is safe to implement the time-varying policies $\{\mb M_t\}_{t\geq 0}$ when the model is known.
\end{lemma}
%Lemma \ref{lem: constraint satisfaction of slow varying DAP} provides a way to ensure safety of slowly varying DAP controllers by tightening the safety constraints  established for time-invariant DAP in $\Omega(\vec \epsilon, \theta_*, H)$. The additional constraint tightening term   shows that for slowly varying DAP sequences $\{\mb M_t\}_{t\geq 0}$, if the policy at each stage satisfies the safety constraints established for the time-invariant policies  tightened by an additional term $\epsilon_v(\Delta_M, H)$ to account for the policy  variation $\Delta_M$, then the time-varying policy sequence is safe to implement when the model is known.
\begin{myspecial}
\red{can I define safe policy set, and a conservative safe policy set  tightened to address slow variation of the policies. Then, with model uncertainty, we further tighten the constraints. Define $\Omega(\theta, \epsilon_x, \epsilon_u)$. By choosing $\Delta_M$ small enough, approximate the optimal DAP. }
\end{myspecial}

\section{Safe Adaptive  Control Algorithm}

This section introduces our safe adaptive control algorithm for constrained LQR. Our algorithm design adopts a standard model-based approach as in, e.g., \citet{mania2019certainty,simchowitz2020naive}. That is, we improve  model estimations with newly collected data and update the policies by near-optimal   policies computed based on the current model estimation. 

However, the constraints in our problem bring additional challenges on, e.g., constraint satisfaction, feasibility, tradeoff among exploration, exploitation, and safety. To address these challenges, we design three  components in Algorithm \ref{alg: online algo} that are  different or irrelevant in the unconstrained case (see e.g., \citet{mania2019certainty,simchowitz2020naive}), i.e.,  (i) the \texttt{RobustCE} and \texttt{ApproxDAP} subroutines (Lines 3,6, 8,10), (ii)  safe policy updates by Algorithm \ref{alg: safe transit} (Line 4 \& Line 9), (iii) Phase 2: pure exploitation  (Line 8-11).

In the following, we explain the   components (i) and (ii) in detail  then discuss the overall algorithm. The component (iii) is motivated by our regret analysis and will be explained after our regret bound in Theorem \ref{thm: regret bound}. %Lastly, we briefly discuss the overall algorithm.

%\red{the name cautious is confusing, but robust CE means something else, think more.}

%\red{Q: how about robustly safe CE, or CE with robust constraint satisfaction?}

\nbf{(i) Approximate DAP and  robust CE.}
When the model $\theta_*$ is unknown but an estimation $\hat \theta$ is available, we consider the following approximate DA policy: 
\vspace{-5pt}
\begin{equation} \label{equ: DAP unknown model, etat}
	u_t = \sum_{k=1}^H M[k]\hat w_{t-k}+\eta_t, 
\end{equation}
where   we approximate the disturbances by
$ \hat w_t =\Pi_{\mathbb W}(x_{t+1}-\hat A x_{t}-\hat B u_{t} )$ and the projection on $\W$ benefits constraint satisfaction  but  introduces nonlinearity to the  policy \eqref{equ: DAP unknown model, etat} with respect to the history states. Besides, 
we add an excitation noise $\eta_t\overset{\text{ind.}}{\sim}\! \bar \eta\D_\eta$ in  \eqref{equ: DAP unknown model, etat} to encourage exploration, where   $\bar \eta\geq 0$ is called an excitation level, the distribution $\D_\eta$ has zero mean, $(s_\eta, p_\eta)$-anti-concentration for some $s_\eta, p_\eta$, and bounded support such that $\|\eta_t\|_\infty \leq \bar \eta$. Examples of $\D_\eta$ include truncated Gaussian, uniform distribution, etc.

\begin{algorithm2e}[htp]
	\caption{Safe Adaptive Control}
	\label{alg: online algo}
	\SetAlgoNoLine
	\DontPrintSemicolon
	\LinesNumbered
	\KwIn{$\Theta_{\text{ini}},\mathcal D_\eta$. $ T^{(e)}, H^{(e)},\bar \eta^{(e)},\Delta_M^{(e)}$, $T_D^{(e)},\forall e.$}
	\textbf{Initialize: }$\hat \theta^{(0)}=\hat \theta_{\text{ini}}, r^{(0)}=r_{\text{ini}}, \Theta^{(0)}=\Theta_{\text{ini}}$.  Define  $w_t=\hat w_t=0$ for $t<0$, $t_1^{(0)}=0$.
	
	\For{\textup{Episode} $e=0,1, 2, \dots$}{
		\tcp{Phase 1$:$ 
			exploration \& exploitation}
		$(\mb M_{\dagger}^{(\!e\!)}\!,\Omega_{\dagger}^{(\!e\!)}\!)\leftarrow$\FuncSty{RobustCE}($\Theta^{(\!e\!)}\!,\! H^{(\!e\!)}, \bar \eta^{(\!e\!)},\! \Delta_M^{(\!e\!)}$).\;

		Run Algorithm \ref{alg: safe transit} to \blue{safely update the policy} to $\mb M_\dagger^{(e)}$ and output $t_1^{(e)}$ if $e>0$.  The algorithm  inputs are $(\mb M^{(e-1)}, \Omega^{(e-1)}, \Theta^{(e)}, 0, \Delta_M^{(e-1)})$, $(\mb M_\dagger^{(e)},\Omega_\dagger^{(e)}, \Theta^{(e)},  \bar \eta^{(e)}, \Delta_M^{(e)})$,  $T^{(e)}$.
		% 		%\tcc{Safe policy update.
		
		%Run \FuncSty{Execute-Policy}($x_t,\bm M_\dagger^{(e)}, \hat \theta^{(e)}, \bar \eta^{(e)}$) for $T_D^{(e)}$ stages and collect data. \;
		
		\For{$t=t_1^{(e)},\dots, t_1^{(e)}+T_D^{(e)}-1$}{
			Run \FuncSty{ApproxDAP}($\bm M_\dagger^{(e)}, \hat \theta^{(e)}, \bar \eta^{(e)}$).\;% and collect data.\;
		}
		
		\tcp{Model Updates} 
		$\Theta^{(e+1)}\!\leftarrow$ \FuncSty{ModelEst} ($\{  x_k, u_k\}_{k=t_1^{(e)}}^{t_1^{(e)}+T_D^{(e)}}\!, \bar \eta^{(\!e\!)})$). \;

		\tcp{Phase 2: pure exploitation ($\eta_t=0$)}
		$(\mb M^{(\!e\!)}\!,\Omega^{(\!e\!)}\!)\leftarrow$\FuncSty{RobustCE}($\Theta^{(\!e+1\!)}\!,\! H^{(\!e\!)}, 0,\! \Delta_M^{(\!e\!)}$).\;
		%	Run \FuncSty{CCE}($\Theta^{(e+1)}\!,\! H^{(e)}, 0, \Delta_M^{(e)}$) to obtain $\mb M^{(e)},\Omega^{(e)}$.\;

		Run Algorithm \ref{alg: safe transit} to \blue{safely update the policy} to $\mb M^{(e)}$ and output $t_2^{(e)}$. Algorithm  \ref{alg: safe transit}'s inputs are $(\mb M_\dagger^{(e)}, \Omega_\dagger^{(e)}, \Theta^{(e)},  \bar \eta^{(e)}, \Delta_M^{(e)})$, $(\mb M^{(e)},\Omega^{(e)}, \Theta^{(e+1)},  0, \Delta_M^{(e)})$,  $t_1^{(e)}+T_D^{(e)}$.\;
		% 		Set the algorithm output as $t_2^{(e)}$.

		\For{$t=t_2^{(e)},\dots, T^{(e+1)}-1$}{
			Run \FuncSty{ApproxDAP}($\bm M^{(e)}, \hat \theta^{(e+1)}, 0$). %until $t=T^{(e+1)}$.\;
		}
		
	}
\end{algorithm2e}

\begin{algorithm}
	\LinesNumbered
	\SetKwFunction{FRCE}{RobustCE}
	\SetKwProg{Fn}{Subroutine}{:}{}
	\Fn{\FRCE{$\Theta, H, \bar \eta, \Delta_M$}}{
		%		Construct a robustly safe policy set:
		%		\begin{align*}
		%			\Omega=\{\mb M\!\in \!\M_H\!\!: 
		%			&g_i^x(\mb M;\hat \theta)\!\leq\! d_{x,i}\!-\!\epsilon_x, \forall 1\leq i \leq k_x \\
		%			&g_j^u(\mb M)\!\leq\! d_{u,j}-\epsilon_u, \forall 1\leq j \leq k_u\}
		%		\end{align*}
		%		where $\epsilon_u=\epsilon_{\eta,u}(\bar \eta)$ and
		%		\begin{center}
		%			$\epsilon_x=\epsilon_H(H)+\epsilon_v(\Delta_M,H)+\epsilon_\theta(r)+\epsilon_{\eta,x}(\bar \eta)$.
		%		\end{center}
		%		
		%		
		Compute the optimal policy $\mb M$ to \eqref{equ: CCE}.\;
		
		Construct the robustly safe policy set: 
		
		$\Omega= \Omega(\hat \theta, \epsilon_x, \epsilon_u)$, where $(\epsilon_x, \epsilon_u)$ are defined by \eqref{equ: CCE}.
		%Consider the problem  
	%	\eqref{equ: CCE} with \eqref{equ: CCE constraint tightening}.\;
		
	%	5Compute the optimal policy, denoted by $\mb M^*$.\;
		
		%  to obtain $\mb M^*$. \;
	%	\red{this is confusing, if space permits, use the complete notation}.\;
		
	%	Construct the feasible policy set, denoted by $\Omega$.
		%		\begin{center}
		%			$\mb M^*=\argmin_{\mb M\in \Omega_\dagger} f(\mb M;\hat \theta)$.
		%		\end{center}
		
		%Solve CE with robust constraint satisfaction:
		% %\begin{align*}
		%     &\min_{\mb M\in \M_H} f(\mb M;\hat \theta), \\
		%     \text{s.t. } \ 	&g_i^x(\mb M;\hat \theta)\leq d_{x,i}\!-\!\epsilon_x,\ \forall 1\leq i \leq k_x \\
		% 	&g_j^u(\mb M)\leq d_{u,j}-\epsilon_u, \forall 1\leq j \leq k_u
		% \end{align*}
		% for $\epsilon_u=\epsilon_{\eta,u}(\bar \eta)$ and
		% \begin{center}
		%     $\epsilon_x=\epsilon_H(H)+\epsilon_v(\Delta_M,H)+\epsilon_\theta(r)+\epsilon_{\eta,x}(\bar \eta)$.
		% \end{center} 
		% Let $\mb M_\dagger$ denote the optimal solution.\;
		
		%Run \FuncSty{Execute-Policy}($t_0,\bm M_\dagger, \hat \theta, \bar \eta, \D_\eta, T_D$.
		\Return{\textup{policy $\mb M$ and  robustly safe policy set $ \Omega.$}}}
	%\Return{$\{x_t, u_t, x_{t+1}\}_{t=t_0}^{t_0+T_D-1}$.}}
\end{algorithm}

\begin{algorithm}
	%\SetAlgoNoLine
	%\DontPrintSemicolon
	\setlength{\interspacetitleruled}{1pt}
	\setlength{\algotitleheightrule}{10pt}
	\LinesNumbered
	\SetKwFunction{FExc}{ApproxDAP}
	\SetKwProg{Fn}{Subroutine}{:}{}
	\Fn{\FExc{$\bm M, \hat \theta, \bar \eta$}}{
		%\For{$t=t_0, \dots, t_0+T_D-1$}{

		Implement 
		$
		u_t \!= \!\sum_{k=1}^H\! M[k]\hat w_{t\!-\!k}\!+\!\eta_t,
		$  with $\eta_t\!\overset{\text{ind.}}{\sim}\! \bar \eta\D_\eta$.\; %Observe $x_{t+1}$.%}\;
		Observe $x_{t\!+\!1}$ and	record $\hat w_{t} \!=\!\Pi_{\mathbb W}(x_{t\!+\! 1}\!-\!\hat A x_{t}\!-\!\hat B u_{t} \!)$.\;
		%\Return{$ (u_t, x_{t+1})$.}
	}
\end{algorithm}

\begin{myspecial}
	
	\red{when we say $\Theta$, we mean we know $\hat \theta$ and $r$.}
	
	\red{how about define $\Theta$ by Ball, and do not joint with $\Theta_{ini}$? the only trouble is that, Omega safe for Theta joint ini, not Theta itself. } 
\end{myspecial}
%\begin{figure}[t]
% \removelatexerror

% \setlength{\interspacetitleruled}{0pt}%
% \setlength{\algotitleheightrule}{0pt}%
\begin{algorithm}
	%\SetAlgoNoLine
	%\DontPrintSemicolon
	\setlength{\interspacetitleruled}{1pt}
	\setlength{\algotitleheightrule}{10pt}
	\LinesNumbered
	\SetKwFunction{FModelEst}{ModelEst}
	\SetKwProg{Fn}{Subroutine}{:}{}
	\Fn{\FModelEst{$\{x_k, u_k\}_{k=0}^{t},\bar \eta$}}{
		
		Estimate the model by  LSE  with projection on $\Theta_{\text{ini}}$:
		\begin{center}
			$ \tilde \theta=\argmin_{\theta}\sum_{k=0}^{t-1}\|x_{k+1}-Ax_k-Bu_k\|_2^2,\quad 	\hat \theta=\Pi_{\Theta_{\text{ini}}}(\tilde \theta).$  
		\end{center}
		
		%Set  confidence radius  $r\!=\!\tilde O(\frac{ \sqrt{n^2+nm}}{\sqrt{t} \bar \eta}\!)$ by Corollary \ref{cor: estimation error of our algo}.\;
		
		Set the model uncertainty set as $\Theta=\B(\hat{\theta},r)\cap \Theta_{\text{ini}}$, where  $r=\tilde O(\frac{ \sqrt{n^2+nm}}{\sqrt{t} \bar \eta})$ by Corollary \ref{cor: estimation error of our algo}.\;
		
		\Return{\textup{the model uncertainty set} $\Theta$.}}
	
\end{algorithm}
\vspace{-10pt}

To ensure safety of \eqref{equ: DAP unknown model, etat} without knowing the true model, we rely on robust constrainit satisfaction. Specifically, given a model uncertainty set $\Theta=\B(\hat{\theta},r)\cap \Theta_{\text{ini}}$, where $r$ is the confidence radius of the estimated model $\hat \theta$, we construct a robustly safe policy set 
  to include policies that are safe for any  model in $\Theta$ and any excitation noises $\|\eta_t\|_\infty \leq \bar \eta$. This robustly safe policy set can be determined by $\Omega(\hat \theta, \epsilon_x, \epsilon_u)$ for any
  \begin{equation}\label{equ: epsilonx, epsilonu, CCE, no DeltaM}
  	\epsilon_x\geq \epsilon_\theta(r)\!+\!\epsilon_{\eta,x}(\bar \eta)+\epsilon_H(H\!), \quad \epsilon_u\!\geq \!\epsilon_{\eta,u}(\bar \eta).
  \end{equation}
Compared with the safe policy set $\Omega(\theta_*, \epsilon_x, \epsilon_u)$ in \eqref{equ: Omega def}, the robustly safe policy set replaces the unknown model with the estimated model $\hat \theta$ and includes additional constraint tightening terms  $\epsilon_\theta(r), \epsilon_{\eta,x}, \epsilon_{\eta,u}$ in $\epsilon_x, \epsilon_u$'s lower bounds to be more conservative in the face of model uncertainties and excitation noises to guarantee safety. Formulas of these  constraint-tightening  terms can be determined below based on perturbation analysis, whose proof is provided in the supplementary.
\begin{lemma}[Definitions of $\epsilon_\theta, \epsilon_{\eta,x}, \epsilon_{\eta,u}$]\label{lem: constraint tightening terms}
	Define $	\epsilon_{ \theta}(r) = c_1\sqrt{mn} r$, 
	$	\epsilon_{\eta,x}(\bar \eta)=c_2 \sqrt m \bar \eta$,	$\epsilon_{\eta,u}(\bar \eta)=c_3 \bar \eta$, 
where $c_1, c_2, c_3$ are  polynomials of $\|D_x\|_\infty,\|D_u\|_\infty,\kappa,$ $\kappa_B,\gamma^{-1}, w_{\max},x_{\max}, u_{\max}$.\footnote{Formulas of  $c_1, \dots, c_3$ are  provided in Appendix \ref{append: constraint tightening}.} For any $\mb M\in \Omega(\hat \theta, \epsilon_x, \epsilon_u)$ where $(\epsilon_x, \epsilon_u)$ satisfies \eqref{equ: epsilonx, epsilonu, CCE, no DeltaM},  policy \eqref{equ: DAP unknown model, etat} is  safe for any model in  $ \Theta$,  any $\|\eta_t\|_\infty \leq \bar \eta$, and  any $w_t \in \W$.
	
\end{lemma}

Based on the robustly safe policy set, we can compute a near-optimal and robustly safe DAP by  the QP below:
\begin{equation}\label{equ: CCE}
	\begin{aligned}
		&\min_{\mb M} f(\mb M;\hat \theta)\\
		\text{s.t. } &\mb M\in \Omega(\hat \theta, \epsilon_x, \epsilon_u), \ \epsilon_u= \epsilon_{\eta,u}(\bar \eta), \\
		&\epsilon_x= \epsilon_\theta(r)\!+\!\epsilon_{\eta,x}(\bar \eta)+\epsilon_H(H)+\epsilon_v(\Delta_M,H).
	\end{aligned}
\end{equation}
Compared with the optimal DAP problem \eqref{equ: optimal DAP}, problem \eqref{equ: CCE} replaces the true model $\theta_*$ by the estimated model $\hat \theta$ (certainty equivalence) and ensures robust constraint satisfaction on $\Theta$ and $\| \eta_t\|\leq \bar \eta$. Hence, we  call \eqref{equ: CCE}   robust CE.\footnote{Our robust CE  does not consider min-max cost.} Notice that we include an additional constraint tightening term $\epsilon_v$ in \eqref{equ: CCE} to ensure safety of slowly varying policies according to Section \ref{subsec: slow variation trick}, which is crucial for our safe policy update algorithm below.

\begin{mycitation}

 safe small policy variations, which offers freedom for safe policy updates as discussed below.

% help ensure safe policy updates as shown below.

\begin{align}
	\Omega\coloneqq\{\mb M \in \M_H: & g_i^x(\mb M;\hat \theta)\!\leq\! d_{x,i}\!-\!\epsilon_x, \forall 1\!\leq \!i \!\leq \!k_x \\
	&g_j^u(\mb M)\!\leq\! d_{u,j}-\epsilon_u, \forall 1\!\leq\! j \!\leq\! k_u\}
\end{align}
where
\begin{equation}\label{equ: epsilonx, epsilonu, CCE, no DeltaM}
	\epsilon_x\geq \epsilon_\theta(r)\!+\!\epsilon_{\eta,x}(\bar \eta)+\epsilon_H(H\!), \quad \epsilon_u\!\geq \!\epsilon_{\eta,u}(\bar \eta)	.
\end{equation}

 on any possible model in the model uncertainty set, denoted

\red{revise the following paragraph, your goal is to make the preliminary and this RCE easy to compare, so it is easy to follow, not much extra efforts.}

To approximate the optimal DAP while suffering model uncertainties, we formulate the following problem:
\begin{equation}\label{equ: CCE}
	\begin{aligned}
		&\min_{\mb M\in \M_H} f(\mb M;\hat \theta)\\
		\text{s.t. } &g_i^x(\mb M;\hat \theta)\!\leq\! d_{x,i}\!-\!\epsilon_x, \forall 1\leq i \leq k_x \\
		&g_j^u(\mb M)\!\leq\! d_{u,j}-\epsilon_u, \forall 1\leq j \leq k_u
	\end{aligned}
\end{equation}
Compared with \eqref{equ: time-invariant optimal DAP}, the problem \eqref{equ: CCE} replaces the true model $\theta_*$ with the estimated model $\hat \theta$ since $\theta_*$ is unknown. This approach is called CE in  the control literature (cite??). Further, the constraint-tightening terms $\epsilon_x, \epsilon_u$ should be carefully designed to ensure  safety of our approximate DAP despite model uncertainties and excitation noises.

%despite the additional model estimation errors and excitation noises. We achieve this by establishing robust constraint satisfaction for any possible model in a given model uncertainty set $\Theta=\B(\hat \theta, r)\cap \Theta_{\text{ini}}$ and for any $\|\eta_t\|_\infty\leq \bar \eta$. 

%With a model uncertainty set $\Theta=\B(\hat \theta, r)\cap \Theta_{\text{ini}}$, where $r$ is a model estimation error bound, we are able to bound the impact of the model estimation errors on the state constraints by $\epsilon_\theta(r)$, whose formula will be provided in Definition \ref{def: constraint tightening terms} in Section ??. Further, we can also bound the impact the excitation noises on the state and action constraints by $\epsilon_{\eta,x}(\bar \eta), \epsilon_{\eta,u}(\bar \eta)$

Specifically, we can set
\begin{equation}\label{equ: CCE constraint tightening}
		\epsilon_x\!=\!\epsilon_\theta(r)\!+\!\epsilon_{\eta,x}(\bar \eta)+\epsilon_H(H\!)\!+\epsilon_v(\Delta_M,\!H)\!, \epsilon_u\!=\!\epsilon_{\eta,u}(\bar \eta)	
\end{equation}
where $\epsilon_\theta(r), \epsilon_{\eta,x}, \epsilon_{\eta, u}$ are additional constraint tightening terms to account for the model estimation errors bounded by $r$ and excitation noises bounded by $\bar \eta$, whose formulas will be provided in Definition \ref{def: constraint tightening terms} in Section ??. 
$\epsilon_H(H)$ is introduced in Section ?? and is needed even  with perfect model information,  $\epsilon_v(\Delta_M,H)$ is  defined in Lemma \ref{lem: constraint satisfaction of slow varying DAP} to ensure safety of slowly varying policy variations, which will help ensure safe policy updates as shown later.

 It can be further shown that any feasible policy to \eqref{equ: CCE}, \eqref{equ: CCE constraint tightening} is \textit{robustly safe} for any  model in the uncertainty set $\Theta=\B(\hat \theta, r)\cap \Theta_{\text{ini}}$ and any excitation noises $\|\eta_t\|_\infty \leq \bar \eta$ (see the proof of theorem \ref{thm: constraint satisfaction}). Therefore, we call feasible policy set, denoted by $\Omega$, as a robustly safe policy set, and call the problem \eqref{equ: CCE} with \eqref{equ: CCE constraint tightening} as Robustly Safe CE (RS-CE).
\end{mycitation}
\begin{mycitation}

\red{should I add $\Delta_M$ here, or only add it in the safe transit part?}
\end{mycitation}

\begin{algorithm2e}[htp]
	\caption{SafeTransit}
	\label{alg: safe transit}
	\SetAlgoNoLine
	\DontPrintSemicolon
	\LinesNumbered
	\KwIn{$(\mb M, \Omega,\Theta,  \bar \eta, \Delta_M)$.  $(\mb M', \Omega',\Theta',  \bar \eta', \Delta_M')$.  $ t_0$. }
	
	Set $\bar \eta_{\min}=\min(\bar \eta, \bar \eta')$, $\hat \theta_{\min}\!=\! \hat \theta \id_{(r\leq r')}\!+\! \hat \theta' \id_{(r> r')}$. \;
	Set an intermediate  policy  as ${\mb M}_{\text{mid}}\in \Omega\cap\Omega'$.\;
	\tcp{Step 1$:$ slowly move from $\mb M$ to ${\mb M}_{\text{mid}}$}
	%	\textit{Step 1: safe transition from $\mb M$ to ${\mb M}_{\text{mid}}$.} 
	Define $W_{1}=\max( \ceil{\frac{\|\mb M-{\mb M}_{\text{mid}}\|_F}{\min(\Delta_M, \Delta_M')}},H')$.\;
	
	%and $\bar \eta_{s_1}=\min(\bar \eta, \bar \eta')$.  \;
	%	\uIf{$r_\theta\leq r_{\theta}'$}{
	%	Define $\hat \theta_{s_1}=\hat \theta$.}
	%    \Else{Define $\hat \theta_{s_1}=\hat \theta'$.}
	\For{$t=t_0, \dots, t_0+W_{1}-1$}{
		Set
		$\mb M_t= \mb M_{t-1}+\frac{1}{W_{1}}({\mb M}_{\text{mid}}-\mb M).$ \;
		Run \FuncSty{ApproxDAP}($\bm M_t, \hat \theta_{\min}, \bar \eta_{\min}$).\;
		
		%	Generate $\eta_t\sim \bar \eta_{s_1} \mathcal D_\eta$ independently from the history. \;
		%	Implement \eqref{equ: DAP unknown model, etat} with policy $\mb M_t$, noise $\eta_t\overset{\text{i.i.d.}}{\sim}\! \bar \eta_{\min} \D_\eta$, and estimate $\hat w_t$ by $\hat \theta_{\min}$.
	}
	\tcp{Step 2$:$ slowly move from ${\mb M}_{\text{mid}}$ to $\mb M'$} Define $W_{2}= \ceil{\frac{\|\mb M'-{\mb M}_{\text{mid}}\|_F}{\Delta_M'}}$. \;%Define $\bar \eta_{s_2}=\bar \eta'$. Define $\hat \theta_{s_2}=\hat \theta'$.\;
	\For{$t=t_0+W_{1}, \dots, t_0+W_{1}+W_{2}-1$}{
		Set
		$\mb M_t=\mb M_{t-1}+\frac{1}{W_{2}}(\mb M'-{\mb M}_{\text{mid}}). $\;%\frac{k}{W_{s_2}}\mb M'+\frac{W_{s_2}-k}{W_{s_2}} \mb M_{\text{mid}}.$\;
		%	Generate $\eta_t\sim \bar \eta_{s_2} \mathcal D_\eta$ independently from the history. \;
		Run \FuncSty{ApproxDAP}($\bm M_t, \hat \theta', \bar \eta'$).\;
		%	Implement \eqref{equ: DAP unknown model, etat} with policy $\mb M_t$, noise $\eta_t\overset{\text{i.i.d.}}{\sim}\! \bar \eta' \D_\eta$, and estimate $\hat w_t$ by $\hat \theta'$.
	}
	\KwOut{Termination stage $t_1=t_0+W_{1}+W_{2}$.}
\end{algorithm2e}

\nbf{(ii) SafeTransit Algorithm.} Notice that at the start of each phase in Algorithm \ref{alg: online algo}, we compute a new policy to implement in this phase. However, directly changing the old policy to the new one may cause constraint violation as discussed in Section \ref{subsec: slow variation trick}.  

To address this, we design Algorithm \ref{alg: safe transit} to ensure safe policy updates at the start of each phase. The high-level idea of Algorithm \ref{alg: safe transit} is based on the slow-variation trick reviewed in Section \ref{subsec: slow variation trick}. That is, we construct a policy path connecting the old policy to the new policy such that this policy path is contained in some robustly safe policy set with  an additional constraint tightening term to allow slow policy variation, then by slowly varying the policies along this path, we are able to safely transit  to the new policy. 

\begin{myspecial}
    \red{(please review the Section ?? in the same language as this one)!}
\end{myspecial}

Next, we focus on constructing such a policy path. We illustrate our construction by Figure \ref{fig:illustration}. We follow the notations in Algorithm \ref{alg: safe transit}, i.e. the old policy is $\mb M$  in an old robustly safe policy set $\Omega$ and the new policy is $\mb M'$ in  $\Omega'$. Notice that the straight line from $\mb M$ to $\mb M'$ does not satisfy the requirements of the slow variation trick because some parts of the line are outside both robustly safe policy sets. To address this, Algorithm \ref{alg: safe transit} introduces an intermediate policy $\mb M_{\text{mid}}$ in $\Omega\cap\Omega'$, and  slowly moves the policy from the old one $\mb M$ to the intermediate one $\mb M_{\text{mid}}$ (Step 1),  then slowly moves from $\mb M_{\text{mid}}$ to the new policy $\mb M'$ (Step 2). In this way, all the path is included in at least one of the robustly safe policy sets, which allows 
% add "a" or "the"
safe transition from the old policy to the new policy. 

The choice of $\mb M_{\text{mid}}$ is not unique. In practice, we recommend selecting $\mb M_{\text{mid}}$ with a shorter path length for quicker policy transition. The existence of $\mb M_{\text{mid}}$ can be guaranteed if the first \texttt{RobustCE} program in Algorithm \ref{alg: online algo} (Phase 1 of episode 0) is strictly feasible. This is usually called recursive feasibility and will be formally proved in Theorem \ref{thm: feasibility}.
\begin{figure}[h]
	\centering
	\includegraphics[scale=0.7]{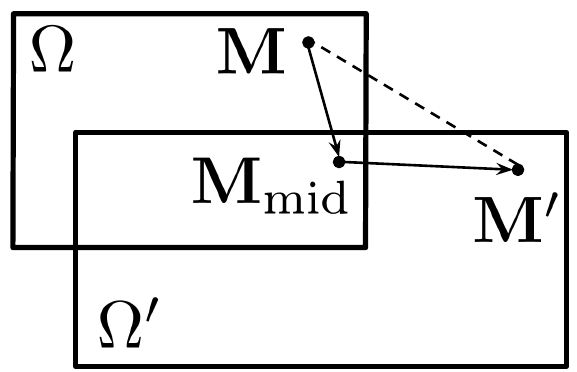}
	\caption{Illustration for Algorithm \ref{alg: safe transit}.}
	\label{fig:illustration}
\end{figure}

\begin{remark}
	If the robustly safe policy sets are monotone, e.g., if $\Omega\subseteq \Omega'$, then we can let $\mb M_{\text{mid}}=\mb M$ and the path constructed by Algorithm \ref{alg: safe transit} reduces to the straight line from $\mb M$ to $\mb M'$.  Hence, a non-trivial $\mb M_{\text{mid}}$ is only relevant when the robustly safe policy sets are not monotone, which can be caused by the non-monotone model uncertainty sets generated by  LSE (even though the error bound $r$ of LSE decreases with more data, the change in the point estimator $\hat \theta$ may cause $\Theta'\not\subseteq \Theta$). Though one can enforce decreasing uncertainty sets by taking joints over all the history  uncertainty sets, this approach leads to an  increasing number of constraints when determining the robustly safe policy sets in \textup{\texttt{RobustCE}}, thus being computationally demanding especially for large episode index $e$. 
	
	%because the number of the constraints in \eqref{equ: CCE} will increase linearly with the number of history  uncertainty sets and go to infinity as time goes to infinity.
	
%	Besides, we would like to mention that another way to ensure monotone model uncertainty sets in the literature is by using set-membership model estimation  (cite??), whose non-asymptotic error  bounds are still under-explored. %But to the best of our knowledge, the non-asymptotic error bounds of set-membership estimation is not well-understood, causing difficulties in establishing non-asymptotic regret guarantees for the algorithms based on the set-membership.
	
%	we are not aware of the non-asymptotic bound of the set-membership estimation error, thus leaving the non-asymptotic performance guarantees of algorithms based on  as open questions 
\end{remark}

%\nbf{(iii) Pure exploitation phases.}  The pure exploitation phases in  Algorithm \ref{alg: safe transit} are motivated from our regret analysis in Section \ref{sec: theory}: the algorithm still works without these phases but will generate a worse regret bound. More   discussions and intuitions will be provided after our regret bound in Theorem \ref{thm: regret bound}.

%Intuitively, we add these phases to encourage exploitation because the excitation noises can be shown to have a worse impact on the constrained control performance 

%includes a pure exploitation phase with no  excitation noises, which is not present in the unconstrained algorithms \citep{dean2018regret,mania2019certainty,simchowitz2020naive}. This phase is motivated by our regret analysis: the algorithm still works without this phase but will generate a worse regret bound. An intuitive explanation is that the excitation noises will degrade the control performance more in constrained LQR because the noises call for more conservative controllers to ensure safety, so we remove the noises for some phases to encourage exploitation for better overall performance. We will provide more discussions in Section \ref{sec: theory} after Theorem \ref{thm: regret bound}, where we will use analytical discussions to support our intuitions.

%\vspace{-5pt}
\nbf{Discussions on the overall algorithm.} 
Firstly, though  Algorithm \ref{alg: online algo} is implemented by episodes, it is still a 
% remove "a"
single-trajectory learning since no system starts are needed when new episodes start. The episodic design  reduces the computational burden since  a new policy is only computed once in a while. Further, the new policies are computed by solving \eqref{equ: CCE}, which is 
QP and enjoys polynomial-time solvers. Besides, $\mb M_{\text{mid}}$ in Algorithm \ref{alg: safe transit} can also be computed efficiently since  the set $\Omega\cap \Omega'$ is a polytope. As for the model updates, we can use all the history data in practice, though we only use part of history in \texttt{ModelEst} for simpler analysis.  \texttt{ModelEst}  also projects  the estimated model onto $\Theta_{\text{ini}}$ to ensure bounded estimation.  We also  note that Algorithm \ref{alg: safe transit} is not the unique way to ensure safe policy updates: other path design or MPC   can also work. More discussions are provided in the remarks below.

%The major difficulty is to ensure optimality/no regret while guaranteeing safe policy updates. For example, 

\begin{remark}[Comparison with literature]
Some constrained control methods   construct safe \textup{state} sets and safe \textup{action} sets for  the current  state, e.g.,  control barrier functions \citep{ames2016control}, reachability-set-based methods \citep{akametalu2014reachability}, regulation maps \citep{kellett2004smooth}, etc. In contrast, this paper constructs  safe \textup{policy} sets in the  space of policy parameters $\mb M$. This is possible because our policy structure (linear on history disturbances) allows  a transformation from linear constraints on the states and actions to polytopic constraints on the policy parameters. 
\end{remark}

\begin{remark}[More discussions on safe policy updates]
As mentioned in Section \ref{subsec: slow variation trick}, for time-varying policies, state constraints require  coupling policy constraints that  depend on not only the current policy but also the history ones, so the decoupled constraints in \eqref{equ: CCE} no longer guarantee safety. To address this, we adopt the slow-variation trick in \cite{li2020online} and slowly update the policies. Another possible idea is to introduce additional constraints to \eqref{equ: CCE} conditioning on the history policies and directly compute a new policy that is safe to implement given the current history policies. However, these additional constraints may reduce the sizes of the safe policy sets  and thus  sacrifice optimality.

% and we 

% the 

% the policy constraints  couples the current policy with the history ones, so the constraints in \eqref{equ: CCE} cannot guarantee safety. To address the coupling constraints, we adopt the slow-variation trick in \cite{li2020online} and slowly varies the policies towards the new one. Besides, it might be interesting to mention another idea to handle the coupling constraints in our setting: at the start of each episode, select a policy by  \eqref{equ: CCE} with  additional coupling constraints conditioning on the history policies. However, this idea may sacrifice optimality because these additional constraints can induce  smaller safe policy sets.

\end{remark}

%Last but not the least, 

\begin{myspecial}
	\red{abuse the name: not robust on cost, only robust on constraints.}
	
\end{myspecial}

\section{Theoretical Analysis}\label{sec: theory}

In this section, we provide theoretical guarantees of our algorithms including model estimation errors, feasibility, constraint satisfaction, and a regret bound. %We also formally define the constraint-tightening terms $\epsilon_\theta(r), \epsilon_{\eta,x}(\bar \eta), \epsilon_{\eta,u}(\bar \eta)$ in robustly safe policy sets (cite??) when estabilishing the feasibility and safety.

%This section presents   theoretical guarantees on  feasibility, safety, and regret of our adaptive control algorithm.  
%As preparations, we will first  quantify the estimation error bounds $r^{(e)}$ and the constraint-tightening terms $\epsilon_\theta(r), \epsilon_{\eta,x}(\bar \eta), \epsilon_{\eta,u}(\bar \eta)$ used in our algorithms before introducing the theoretical guarantees. 

%\subsection{Preparations}
\subsection{Estimation Error Bound}

The estimation error bounds for linear policies have been studied in the literature \cite{dean2018regret}. However, due to the projection in our disturbance approximation in \eqref{equ: DAP unknown model, etat}, the policies implemented by our algorithms can be sometimes nonlinear. To cope with this, we provide an  error bound below for  general policies.

\begin{theorem}[Estimation error bound]\label{thm: general estimation error bdd}
Consider  actions $u_t=\pi_t(x_0, \{w_k,\eta_k\}_{k=0}^{t-1})+\eta_t$, where $\|\eta_t\|_\infty\leq \bar \eta$ is generated as discussed after \eqref{equ: DAP unknown model, etat} and policies $\pi_t(\cdot)$  ensure bounded states and actions, i.e. $\|(x_t^\top, u_t^\top)\|_2\leq b_z$ for all $t\geq 0$.    Let $\tilde \theta_T=\min_{A,B} \sum_{t=0}^{T-1}\|x_{t+1}-Ax_t-Bu_t\|_2^2$ denote the model estimation.  %and $\hat \theta=\Pi_{\Theta_{\text{ini}}}\tilde \theta$. Suppose $\theta_*\in \Theta_{\text{ini}}$. 
	For any $0<\delta<1/3$, for $T\geq O(\log(1/\delta)+(m+n)\log(b_z/\bar \eta))$, 
	 with probability (w.p.) $1-3\delta$, we have
	$$\|\tilde \theta_T-\theta_*\|_2 \leq \! O\!\left(\!\sqrt{n\!+\!m}\frac{\sqrt{\log(b_z/\bar\eta+1/\delta)}}{\sqrt T \bar \eta}\!\right).$$
%	\frac{90 \sigma_{sub} }{p_z}\frac{\sqrt{n+(n+m)\log(10/p_z)+2(n+m)\log(\sqrt{b_x^2+b_u^2}/s_z)+\log(1/\delta)}}{\sqrt T s_z}$$
%	where  $s_z=\min(\frac{s_w}{4}, \frac{\sqrt 3}{2}s_\eta\bar \eta, \frac{s_ws_\eta}{4b_u} \bar \eta),p_z=\min(p_w, p_{\eta})$.
\end{theorem}

%define the variables used in our algorithms, i.e the estimation error bound $r^{(e)}$    in Algorithm \ref{alg: online algo}, and the  terms $\epsilon_\theta(r), \epsilon_{\eta,x}(\bar \eta), \epsilon_{\eta,u}(\bar \eta)$ in robustly safe policy sets (cite??).

%In this section, we provide theoretical guarantees of our algorithms. We show feasible, safe, and small regret under proper conditions. To prepare for the theoretical guarantees, we first formally define some variables used in our algorithms, i.e. the estimation error bound r^{(e)}, and the constraint tightening term $\epsilon_\theta, \epsilon_\eta$. 

\begin{myspecial}
	\red{Q: how to discuss the anti-concentration, and its comparison with Sigma>sigmaw?}
\end{myspecial}
\vspace{-10pt}

Theorem \ref{thm: general estimation error bdd} holds for both linear and nonlinear policies as long as the  induced states and actions are bounded, which can be guaranteed by the stability of the policies. Further, the  error bound in Theorem \ref{thm: general estimation error bdd}  is   $\tilde O(\frac{\sqrt{m+n}}{\bar \eta \sqrt T})$, which coincides with the error bound  for linear policies in terms of $T, \bar \eta, n, m$ in \cite{dean2018regret,dean2019safely}.

Based on Theorem \ref{thm: general estimation error bdd}, we obtain a formula for the estimation error bound $r^{(e)}$ in Line 7 of    Algorithm \ref{alg: online algo}.
\begin{corollary}[Formula of $r^{(e)}$]\label{cor: estimation error of our algo}
Suppose   $H^{(0)}\geq \log(2\kappa)/\log((1-\gamma)^{-1})$, $ T^{(e+1)}\geq t_1^{(e)}+T_D^{(e)}$ and $T_D^{(e)}$ satisfies the condition on $T$ in Theorem \ref{thm: general estimation error bdd}. For any $0<p<1$ and $e\geq 1$, with probability at least $1-\frac{p}{2e^2}$, we have $
	\|\hat \theta^{(e)}-\theta_*\|_F \leq r^{(e)}$, where
	\begin{equation}\label{equ: def of re}
		 r^{(e)}\!=\!%\\
		%&\leq \frac{90 \sigma_{sub} \sqrt n}{p_z c_9}\frac{\sqrt{n+(n+m)\log(10/p_z)+2(n+m)\log(\sqrt{c_{10}^2mn+u_{\max}^2}/(c_9 \bar\eta^{(e-1)}))+\log(6e^2/p)}}{\sqrt{ T^{(e-1)}_D} \bar \eta^{(e-1)}}\\
		O\!\left(\!\frac{(n\!+\!\sqrt{mn})\sqrt{\log(\sqrt{mn}/\bar \eta^{(e-1)}\!+\!e^2/p)}}{\sqrt{ T^{(e-1)}_D} \bar \eta^{(e-1)}}\!\right)\!.
	\end{equation}
\begin{myspecial}
		\red{Q: how does $r^e$ depends on p? I need this, otherwise the p in constraint satisfaction looks weird}
\end{myspecial}

%=r(T_D^{(e-1)}, \bar \eta^{(e-1)})

	%holds with probability $1-\frac{p}{2e^2}$, when $e\geq 1$.
\end{corollary}
\vspace{-8pt}

Corollary \ref{cor: estimation error of our algo} considers the $\|\cdot \|_F$ norm because Algorithm \ref{alg: online algo}   projects matrix $\tilde \theta^{(e)}$ onto $\Theta_\text{ini}$ and the $\|\cdot \|_F$ norm  is more convenient to analyze and implement when matrix projections are involved.  Due to the  change of norms, the  error bound has an additional $\sqrt n$ factor.

\subsection{Feasibility and Constraint Satisfaction}
This section provides feasibility and constraint satisfaction guarantees of our adaptive control algorithm.
\begin{myspecial}
	\red{explain how this compares with the bound when the model is known, and what's the use of the additional terms}
	
	\red{define $\epsilon_{rob,x}^{(0)} $ and $\epsilon_{rob,u}^{(0)} $ somewhere! I change this to be more concise, to make formula look cleaner and on the same line.}
\end{myspecial}

\begin{theorem}[Feasibility]\label{thm: feasibility}
	Algorithms \ref{alg: online algo} and  \ref{alg: safe transit} output feasible policies for all $t$ under the following conditions.
	\begin{enumerate}[nosep]\setlength{\itemsep}{0.5pt}

		\item[(i)] (Strict initial feasibility) There exists $\epsilon_0>0$ such that $\Omega(\hat \theta^{(0)},\epsilon_x^{(0)}+\epsilon_0, \epsilon_u^{(0)})\not=\emptyset $, where $\epsilon_x^{(0)},\epsilon_u^{(0)} $ are defined by \eqref{equ: CCE} with initial parameters $r^{(0)}, \bar \eta^{(0)}, H^{(0)}, \Delta_M^{(0)}$.
		
% 		There exists a policy $\bm M_F\in \M_{H^{(0)}}$ and $\epsilon_0>0$ s.t.
% 			\begin{align*}
% 		&  g_i^x(\bm M_F;\hat \theta^{(0)})\leq d_{x,i}-\epsilon_{rob}^{(0)} -\epsilon_0, \ \forall\, i\\
% 			& g_j^u(\bm M_F)\leq d_{u,j}-\epsilon_{rob, u}^{(0)}, \ \forall j.
% 		\end{align*}
		
	%	 with the first line of the constraints in  (cite??) being inactive, i.e.
	%	$g_i^x(\bm M;\hat \theta^{(0)})\leq d_{x,i}-\epsilon_{rob,x}^{(0)} -\epsilon_0$ for some $\epsilon_0>0$ for all $i$.
		
%		$\epsilon_0>0$ such that $\Omega_0\not=\emptyset$, where
%		\begin{align*}
%			\Omega_0=\{\bm M\in \M_H: & g_i^x(\bm M;\hat \theta^{(0)})\leq d_{x,i}-\epsilon_{rob,x}^{(0)} -\epsilon_0,\\
%			& g_j^u(\bm M)\leq d_{u,j}-\epsilon_{rob, u}^{(0)}, \ \forall \, i, j\}.
%		\end{align*}
		\item[(ii)] (Monotone parameters) $\bar \eta^{(e)}, H^{(e)}, T_D^{(e)}, \Delta_M^{(e)}$ are selected s.t.  $(H^{(e)})^{-1}, \sqrt{H^{(e)}}\Delta_M^{(e)}$, $\bar \eta^{(e)}$,  $r^{(e)}$ are all non-increasing with $e$, and  $r^{(1)}\leq \frac{\epsilon_0}{c_1\sqrt{mn}}$, where $r^{(e)}$   is defined in \eqref{equ: def of re}, $c_1$ is defined in  Lemma \ref{lem: constraint tightening terms}.
		%	\item $H^{(0)}\geq \log(2\kappa)/\log((1-\gamma)^{-1})$, and $T_D^{(0)}$ and $\bar \eta^{(0)}$ are selected such that $r^{(1)}\leq \frac{\epsilon_0}{c_4\sqrt{mn}}$.
	\end{enumerate}
	%	\begin{align}
	%		&T_D^{(e)}\geq T_D^{(e\!-\!1)}, \!\bar \eta^{(e)}\!\leq \!\bar \eta^{(e\!-\!1)}, H^{(e)}\!\geq H^{(e\!-\!1)}\!, \Delta_M^{(e)}\!\leq \!\sqrt{\frac{H^{(e\!-\!1)}}{H^{(e)}}} \Delta_M^{(e\!-\!1)}, r^{(e)}\!\leq\! r^{(e-1)},  \forall\, e\geq 1,\tag{I}\\
	%		&\epsilon_H(H^{(0)})+\epsilon_{P,x}(H^{(0)})+
	%		\epsilon_v(\Delta_M^{(0)},H^{(0)})+\epsilon_{\eta,x}(\bar \eta^{(0)})\leq \epsilon_{F,x}/2\tag{II}\\
	%		&\epsilon_{\eta,u}(\bar \eta^{(0)})+\epsilon_{P,u}(H^{(0)})\leq \epsilon_{F,u},\tag{III}
	%	\end{align}
	%	and
	%	$ H^{(0)} \geq \log(2\kappa)/\log((1-\gamma)^{-1})$, where $\epsilon_{H}(H)=c_1(1-\gamma)^{H}, 
	%	\epsilon_{\eta,x}(\bar \eta)=c_2\sqrt m \bar \eta, 
	%	\epsilon_{\eta,u}(\bar \eta)=c_3\bar \eta, 
	%	\epsilon_{ \theta}(r) = c_4\sqrt{mn} r,
	%	\epsilon_{v}(\Delta_M, H)= c_5\sqrt{mnH}\Delta_{M},
	%	\epsilon_{P,x}(H)=c_6 \sqrt n (1-\gamma)^H, \epsilon_{P,u}(H)=c_7 \sqrt n (1-\gamma)^H $, and $c_1,\dots, c_7$ are $\text{poly}(\|D_x\|_\infty,\|D_u\|_\infty,\kappa,\kappa_B,\gamma^{-1}, w_{\max},x_{\max}, u_{\max}) $.
	%	% and 
	%	%$T_D^{(e)}\geq T_D^{(e-1)}$, $\bar \eta^{(e)}\leq \bar \eta^{(e-1)}$,    $H^{(e)}\geq H^{(e-1)}$, $\Delta_M^{(e)}\leq (\sqrt{H^{(e-1)}}/\sqrt{H^{(e)}}) \Delta_M^{(e-1)}$. 
	%	
	
	Furthermore, under Assumption \ref{ass: KF epsilonF}, condition (i) above  is satisfied if 
	\begin{align}\label{equ: initial feasibility}
		&\epsilon_{x}^{(0)} \!+\!\epsilon_0\leq  \epsilon_{F,x}\!-\!\epsilon_\theta(r_{\textup{ini}})\!-\!\epsilon_{P}, \epsilon_{u}^{(0)}\leq \epsilon_{F,u}\!-\!\epsilon_{P},
	\end{align}
	where $\epsilon_{P}=c_4 \sqrt{mn} (1-\gamma)^{H^{(0)}}$ and $c_4$ is a polynomial of $\|D_x\|_\infty,\|D_u\|_\infty,\kappa,$ $\kappa_B,\gamma^{-1}, w_{\max},x_{\max}, u_{\max}$.
\end{theorem}
%The first claim in Theorem \ref{thm: feasibility} establishes \textit{recursive feasibility}, i.e., if the  initial policy set $\Omega_\dagger^{(0)}$ has a strictly feasible solution and if the algorithm parameters satisfy certain monotone conditions, then all the future steps are feasible.

Condition (i)  requires the initial policy set $\Omega_\dagger^{(0)}$ to contain a  policy that strictly satisfies the constraints on $g_i^x(\mb M; \hat \theta^{(0)})$. Condition (ii) requires monotonic  parameters in later phases, where the non-increasing estimation error $r^{(e)}$ requires an increasing number of exploration stages $T_D^{(e)}$. Conditions (i) and (ii) together establish the \textit{recursive feasibility}: if our algorithm is (strictly) feasible at the initial stage, then the algorithm is feasible in the future  under proper   parameters.

The last statement in Theorem \ref{thm: feasibility} provides conditions  \eqref{equ: initial feasibility} for strict \textit{initial feasibility}, which is based on the $\epsilon_F$-strictly safe policy  $u_t=-K_Fx_t$  in Assumption \ref{ass: KF epsilonF}. The  term $\epsilon_P$ captures the difference between the  policy $K_F$  and its approximate DAP. \eqref{equ: initial feasibility} requires large enough $H^{(0)}, T_D^{(0)}$ and small enough $\bar \eta^{(0)}, \Delta_M^{(0)}$. Further,  \eqref{equ: initial feasibility} implicitly requires a small enough initial uncertainty radius: we at least need $\epsilon_\theta(r_{\text{ini}})<\epsilon_{F,x}$. If the initial uncertainty set is too large  but a safe policy is available,  one can first explore the system with the safe policy to reduce the model uncertainty until the strict initial feasibility is obtained and then apply our algorithm. 

% by  the existence of $\epsilon_F$-strictly-safe  policy $u_t=K_Fx_t$  in Assumption \ref{ass: KF epsilonF}. The  term $\epsilon_P$ bounds the error between linear policy $K_F$ and its approximate DAP. Notice that \eqref{equ: initial feasibility} implicitly requires a small enough initial uncertainty $r_{\text{ini}}$: we at least need $\epsilon_{F,x}>\epsilon_\theta(r_{\text{ini}})$. If the initial uncertainty set is too large  but a safe policy is available,  one can first apply the safe policy to reduce the model uncertainty until the initial feasibility is obtained and then apply our algorithm. 

\begin{theorem}[Constraint Satisfaction]\label{thm: constraint satisfaction}
	Under the conditions in Theorem \ref{thm: feasibility} and Corollary \ref{cor: estimation error of our algo}, when $T^{(e+1)} \geq t_2^{(e)}$,  we have $u_t\in \U$ for all $t\geq 0$ w.p.1 and $x_t \in \X$ for all $t\geq 0$ w.p.  $1-p$, where $p$ is chosen  in Corollary \ref{cor: estimation error of our algo}.
	%given any admissible $\theta_*, w_k, \eta_k$.
	%Further,  $x_t \in \X$ holds for all $t\geq 0$ with probability at least $1-p$, where $p$ is defined in Corollary \ref{cor: estimation error of our algo}.
\end{theorem}

The control constraint satisfaction is always ensured by the projection onto $\W$ in \eqref{equ: DAP unknown model, etat}. Besides, we can show that the state constraints are satisfied if the true model is  inside the confidence sets $\Theta^{(e)}$ for all $e\geq 0$, whose  probability is at least $1-p$  by Corollary \ref{cor: estimation error of our algo}.

\subsection{Regret Bound}
Next, we show that our algorithm can achieve a  $\tilde O(T^{2/3})$ regret bound together with  feasibility and constraint satisfaction under proper conditions. Further, we explain the reasons behind the pure exploitation phase. 

%under proper conditions. We further explain how the full exploitation phase improves the regret bound.
%\red{how to add more explanation to the conditions to explain the conditions of the theorem?}

\begin{myspecial}
	\red{what's $T^1$'s dependence on epsilon0, or at least epsilonF?}
\end{myspecial}
\begin{theorem}[Regret bound]\label{thm: regret bound} Suppose $\epsilon_\theta(r_\textup{ini})\leq \epsilon_{F,x}/4$.
For any $0<p<1/2$, with parameters $T^{(1)}\geq \tilde O((\sqrt{n}m +n\sqrt m)^{3})$, $T^{(e+1)}=2T^{(e)}$, $T_D^{(e)}=(T^{(e+1)}-T^{(e)})^{2/3}$, $\Delta_M^{(e)}=O(\frac{\epsilon_F^x}{\sqrt{mn	H^{({e})} } (T^{(e+1)})^{1/3}})$,  $\bar\eta^{(e)}= O(\min( \frac{\epsilon^x_F}{ \sqrt m }, \epsilon_F^u))$, and  $H^{(e)}\geq O(\log(\max(T^{(e+1)},\frac{ \sqrt n}{\min(\epsilon_F)})))$, Algorithm \ref{alg: online algo}
	is feasible and satisfies $\{u_t\in\U\}_{t\geq 0} $ w.p.1 and $\{x_t\in\X\}_{t\geq 0} $ w.p. $(1-p)$. Further, w.p. $(1-2p)$, 
	$$\textup{Regret}\leq \tilde O((n^2m^{{2}}+n^{2.5}m^{{1.5}})\sqrt{mn+k_x+k_u}T^{2/3}).$$
\end{theorem}
\nbf{On $r_\text{ini}$.} As discussed after Theorem \ref{thm: feasibility}, the initial feasibility requires a small enough $r_\text{ini}$, otherwise more exploration is needed before applying our algorithm. For technical simplicity, Theorem \ref{thm: regret bound}  assumes $\epsilon_\theta(r_\text{ini})\leq \epsilon_{F,x}/4$ and establishes other conditions  accordingly. 

\nbf{On parameters.} Theorem \ref{thm: regret bound}  provides choices of parameters that ensure feasibility, safety, and the regret bound. Here, we choose exponentially increasing episode lengths $T^{(e)}$, and explore for $(T^{(e)})^{2/3}$ stages at each episode with a small enough constant excitation level $\bar \eta^{(e)}$. We select  large enough memory lengths $H^{(e)}\geq O(\log(T^{(e)}))$ and small enough variation budgets $\Delta_M^{(e)}\leq O((T^{(e+1)})^{-1/3})$.

\nbf{On  regret.} Though our regret bound $\tilde O(T^{2/3})$ is worse than the $\tilde O(\sqrt T)$ regret bound for \textit{unconstrained} LQR, it has the same order in terms of $T$ with the robust learning of  unconstrained  LQR in \cite{dean2018regret}. This motivates future work on   lower bounds of learning-based control with safety/robustness guarantees.

%\red{cut this one, add more intuition}
\nbf{Discussions on the pure exploitation phase.} Algorithm \ref{alg: online algo} includes a pure exploitation phase with no  excitation noise at each episode, which is not present in the unconstrained algorithms \citep{dean2018regret,mania2019certainty,simchowitz2020naive}. This phase is motivated by  our regret analysis: the algorithm can still work without this phase but will generate a worse regret bound. Specifically, consider a robust CE policy \eqref{equ: CCE} with respect to $\bar \eta$ and uncertainty radius $r$, the regret of this policy per stage can be roughly bounded by $\tilde O(\bar \eta+r)$  in the supplementary (we omit $\Delta_M$ here for simplicity). With no pure exploitation phases, the regret in episode $e$ can be roughly bounded by $\tilde O(T^{(e)}(\bar \eta^{(e)}+ r^{(e)}))$, where $r^{(e)}=\tilde O(\frac{1}{\sqrt{T^{(e-1)}} \bar \eta^{(e-1)}})$ by Corollary \ref{cor: estimation error of our algo} (hiding $n,m$). Therefore, the  total regret  can be  bounded by $\sum_e (\frac{1}{\sqrt{T^{(e-1)}} \bar \eta^{(e-1)}}+\bar \eta^{(e)}) T^{(e)}\approx \sum_e (\frac{1}{\sqrt{T^{(e)}} \bar \eta^{(e)}}+\bar \eta^{(e)}) T^{(e)} $, which is minimized at $\frac{1}{\sqrt{T^{(e)}} \bar \eta^{(e)}}=\bar \eta^{(e)}=(T^{(e)})^{-1/4}$, leading to a worse regret bound $\tilde O(T^{3/4})$.% regret bound that is worse that $\tilde O(T^{2/3})$. 

Notice that $\bar \eta^{(e)}$ chosen in Theorem \ref{thm: regret bound} is a constant, so our algorithm suffers slightly larger stage regret during the Phase 1 (exploration \& exploitation) compared with the no-pure-exploitation algorithm described above with $\eta^{(e)}=(T^{(e)})^{-1/4}\to 0$.  Nevertheless,  by constantly refining the models and reducing $\Delta_M^{(e)}$, the performance during the  phase 1 still improve over episodes.
\begin{myspecial}
    \red{please  add more intuitive explanation to the above discussion.}
\end{myspecial}

\begin{myspecial}
\red{caveat: improve the coefficients of the cosntraint tightening terms}
\end{myspecial}
\vspace{-3pt}

\section{Discussions and  Future Work}
\nbf{Comparing with linear dynamical policies.} Though we only considered linear static policies  as our regret benchmark, it is straightforward to include linear dynamical policies as the  benchmark and achieve similar regret bounds under proper conditions (see e.g., \cite{plevrakis2020geometric}).

\nbf{Comparing with RMPC in \cite{mayne2005robust}.} \cite{mayne2005robust} propose a tube-based RMPC algorithm for constrained LQR with perfect model information. Though this algorithm generates piecewise-affine policies, we are able to show that the infinite-horizon averaged cost of this algorithm can be characterized by $J(K)$, where $K$ is a pre-fixed  safe linear policy for tube construction in the tube-based RMPC (for more details, see \cite{mayne2005robust}).\footnote{Here, $K$ is required to be safe on the linear system starting from 0, but RMPC can be safe to implement starting from nonzero $x_0$.} In other words, our current regret benchmark $K^*$ performs not worse than RMPC in \cite{mayne2005robust} in terms of the infinite-horizon averaged cost, so the regret of our Algorithm \ref{alg: online algo} remains unchanged by including this RMPC  to the benchmark policy set. The major strength of RMPC in \cite{mayne2005robust} compared with our algorithm is the initial feasibility: 
RMPC can safely control the system starting from more distant $x_0$ while our algorithm can only guarantee safety for small $x_0$. More discussions and  formal proofs  are provided in Appendix \ref{append: discuss}.

% $J(K^*)$ is not greater than the infinite-horizon averaged cost of RMPC in \cite{mayne2005robust}, 

% the infinite-horizon averaged cost of RMPC in  Therefore, the regret of our algorithm compared with the RMPC in \cite{mayne2005robust} can also be bounded by $\tilde O(T^{2/3})$  

% that is safe to implement

% or tube construction  in \cite{mayne2005robust} and  $K$ is safe to implement on the linear system when starting from $x_0=0$ (for more details, see Appendix \ref{append: discuss} and \cite{mayne2005robust}). Therefore, we can obtain the same regret bound as in Theorem \ref{thm: regret bound} with RMPC in \cite{mayne2005robust} being the regret benchmark. More discussions and the formal results are provided in Appendix \ref{append: discuss}. 

%t is worth mentioning that  RMPC in \cite{mayne2005robust}

%include the RMPC algorithm  in \cite{mayne2005robust} as our  benchmark and obtain a similar regret bound, even though this RMPC generates piecewise-affine policies. This is because we are able to approximate the long-term cost of this RMPC algorithm by that of a linear policy.  The formal proof is provided in Appendix \ref{append: discuss}. 

%\paragraph{Conclusion.} This paper considers linearly constrained LQR with  model uncertainty. We design a single-trajectory adaptive control algorithm, which involves a SafeTransit algorithm for safe policy updates. We establish feasibility and constraint satisfaction during the process and provide a $\tilde O(T^{2/3})$ regret bound compared with the optimal safe linear policy.
 
\nbf{Future work.} There are many interesting future directions, e.g., (i) improving practical performance  by reducing the constraint tightenings,  (ii)  regret analysis compared with other  RMPC algorithms, (iii) algorithm design and analysis for large initial states, (iv)  fundamental regret lower bounds, (v) safe adaptive control with performance guarantees for  nonlinear systems, etc.

 \bibliographystyle{unsrtnat}
\bibliography{citation4safecontrol}
\newpage
 \vspace{20pt}

\appendix
%\subsection*{Summary}
\section*{Appendices}
The appendices  include the proofs for the theoretical results and more discussions  for the paper. 
\begin{itemize}
	\item Appendix \ref{append: prep} provides necessary lemmas that will be used throughout the appendices and defines the constraint-tightening factors promised in Lemma \ref{lem: constraint tightening terms}. 

\item Appendix \ref{append: estimate error} focuses on the estimation error and provides proofs for Theorem \ref{thm: general estimation error bdd} and Corollary \ref{cor: estimation error of our algo}.

\item Appendix \ref{append: feasible} studies  feasibility and provides a proof for Theorem \ref{thm: feasibility}.

\item Appendix \ref{append: constraint satisfaction} focuses on the constraint satisfaction and provides a proof for Theorem \ref{thm: constraint satisfaction}, which also includes a proof for Lemma \ref{lem: constraint tightening terms}.

\item Appendix \ref{append: regret} analyzes the regret and proves Theorem \ref{thm: regret bound}.

\item Appendix \ref{append: discuss} provides additional discussions on RMPC and non-zero $x_0$.

\item Appendix \ref{append: add proofs} provides proofs to the technical lemmas used in the appendices \ref{append: prep}-\ref{append: regret}.
\end{itemize}

\nbf{Additional Notations.}
Define $\X=\{x: D_x x\leq d_x\}$ and $\U=\{u: D_u u \leq d_u\}$.
Let $\upsilon_{\min}(A)$ and $\upsilon_{\max}(A)$ denote the minimum and the maximum eigenvalue of a symmetric matrix $A$ respectively. For two symmetric matrices $X$ and $Y$, we write $X\leq Y$ if $Y-X$ is positive semi-definite, we write $X<Y$ if $Y-X$ is positive definite. For two vectors $x, y \in \R^n$, we write $x\leq y$ is $(y-x)_i\geq 0$ for $1\leq i \leq n$, i.e. $x$ is smaller than $y$ elementwise. Consider a $\sigma$-algebra $\F_t$ and a random vector $y_t\in \R^n$, we write $y_t\in \F_t$ if the random vector $y_t$ is measurable in $\F_t$. We let $I_n$ denote the identity matrix in $\R^{n\times n}$. Denote an aggregated vector $z_t=(x_t^\top, u_t^\top)^\top$ for notational simplicity. Define $z_{\max}=\sqrt{x_{\max}^2+u_{\max}^2}$ as the maximum $l_2$ norm of $z_t$ for any $x_t,u_t$ satisfying the constraints in \eqref{equ: J(pi)}.  We use ``a.s." as an abbreviation for ``almost surely". Finally, for memory lengths $H_1<H_2$, notice that the set $\M_{H_1}$ can be viewed as a subset of $\M_{H_2}$ since we can append 0 matrices to any $\mb M\in \M_{H_1}$ to generate a corresponding matrix in $\M_{H_2}$, so we will slightly abuse the notation and write $\mb M\in \M_{H_2}$ for any  $\mb M\in \M_{H_1}$ with $H_1<H_2$.

\section{Preparations: State Approximation and Constraint-tightening Terms}\label{append: prep}
This section  provides  results that will be useful throughout the rest of the appendices. Specifically, the first subsection provides a state approximation lemma and a constraint-decomposition corollary for the approximate DAP, and the second subsection defines and discusses the constraint-tightening terms in Lemma \ref{lem: constraint tightening terms} and \eqref{equ: CCE} based on the upper bounds of the constraint decomposition terms.

\subsection{State Approximation and Constraint Decompositions}
We consider a more general form of approximate DAP below than that in \eqref{equ: def DAP}, i.e., an approximate DAP with time-varying policy matrices $\mb M_t$, time-varying excitation levels $\bar \eta_t$, and time-varying model estimations $\hat \theta_t$.
\begin{equation}\label{equ: time varying DAP uncertain}
	u_t=\sum_{t=1}^{H_t} M_t[k]\hat w_{t-k}+\eta_t, \quad \hat w_t=\Pi_{\mathbb W}(x_{t+1}-\hat \theta_t z_t), \quad \|\eta_t\|_\infty \leq \bar \eta_t, \quad t\geq 0,
\end{equation}
where $\mb M_t\in \M_{H_t}$ and $\{H_t\}_{t\geq 0}$ is non-decreasing.

When implementing the time-varying approximate DAP \eqref{equ: time varying DAP uncertain} to the system $x_{t+1}=A_* x_t +B_* u_t +w_t$, we have the following state approximation lemma.

\begin{lemma}[State approximation under time-varying approximate DAP]\label{lem: xt formula DAP uncertain}
	When implementing the time-varying approximate DAP \eqref{equ: time varying DAP uncertain} to the system $x_{t+1}=A_* x_t +B_* u_t +w_t$, we have the following state approximation result: 
	\begin{align*}
		x_t&\!=\!A_*^{H_t}x_{t\!-\!H_t}\!+\!\sum_{k=2}^{2H_t}\sum_{i=1}^{H_t}A_*^{i\!-\!1}B_* M_{t\!-\!i}[k\!-\!i]\hat w_{t-k}\one_{(1\leq k-i\leq H_{t\!-\!i})} \!+\! \sum_{i=1}^{H_t}A_*^{i-1}w_{t\!-\!i}\!+\!\sum_{i=1}^{H_t}A_*^{i\!-\!1}B_* \eta_{t\!-\!i}
	\end{align*}
\end{lemma}
The  lemma above is a straightforward extension from Proposition  \ref{prop: state formula with known system} reviewed in Section \ref{sec: prelim} for the case with perfect model information, thus the proof is omitted.
%\begin{proof}
%	\red{to be typed}
%\end{proof}

To simplify the exposition, we introduce the following notations for time-varying  DAP.
\begin{align}
	&\tilde \Phi_k^x(\mb M_{t-H_t: t};\theta)= A^{k-1}\one_{(k\leq H_t)}+ \sum_{i=1}^{H_t} A^{i-1} B M_{t-i}[k-i]  \one_{(1\leq k-i \leq H_t)}, \quad \forall \, 1\leq k \leq 2H_t,\\
	& \tilde g_i^x(\mb M_{t-H_t: t-1};\theta)\!=\!\sup_{\hat w_{k}\in \W}D_{x,i}^\top \sum_{k=1}^{2H_t} \tilde \Phi_k^x(\mb M_{t-H_t: t-1};\theta) \hat w_{t-k}\!=\! \sum_{k=1}^{2H_t} \|D_{x,i}^\top\tilde \Phi_k^x(\mb M_{t-H_t: t-1};\theta) \|_1 w_{\max},
\end{align}
where $1\leq i \leq k_x$ and we define  $\mb M_t=\mb M_0$ for $t\leq 0$ for notational simplicity. Notice that when $\mb M_t=\mb M$ and $H_t=H$ (the time-invariant case), the definitions of $\tilde \Phi_k^x(\mb M_{t-H_t: t};\theta)$ and $\tilde g_i^x(\mb M_{t-H_t: t-1};\theta)$  above reduce to the definitions of $ \Phi_k^x(\mb M;\theta)$ and $ g_i^x(\mb M;\theta)$ respectively in Section \ref{sec: prelim}. 

Based on Lemma \ref{lem: xt formula DAP uncertain} and the notations defined above, we can obtain the following corollary on the  decompositions of the state constraints $D_x x_t$ and action constraints $D_u u_t$. The decompositions are crucial when defining our constraint-tightening terms and developing the constraint satisfaction guarantees. 

%Next, we provide an upper bound of $D_{x,i}^\top x_t$ based on Lemma \ref{lem: xt formula DAP uncertain}. This upper bound will be helpful when proving state constraint satisfaction.
\begin{corollary}[Constraint decomposition]\label{cor: state constraint decomposition}
	When implementing the time-varying approximate DAP \eqref{equ: time varying DAP uncertain} to the system $x_{t+1}=A_* x_t +B_* u_t +w_t$, for each $1\leq i \leq k_x$ and $1\leq j \leq k_u$, we have the following decompositions:
	\begin{align*}
		D_{x,i}^\top x_t \leq \ & \underbrace{g_i^x(\mb M_t; \hat \theta_t^g)}_{\text{estimated state constraint function}}+ \underbrace{(g_i^x(\mb M_t;  \theta_*)-g_i^x(\mb M_t; \hat \theta_t^g))+ \sum_{k=1}^{H_t} D_{x,i}^\top A_*^{k-1}(w_{t-k}-\hat w_{t-k})}_{\text{model estimation errors}}\\
		& +\underbrace{\sum_{i=1}^{H_t} D_{x,i}^\top A_*^{i-1}B_*\eta_{t-i}}_{\text{excitation errors on the state}}+\underbrace{D_{x,i}^\top A_*^{H_t} x_{t-H_t}}_{\text{history truncation errors}}+\underbrace{ (\tilde g_i^x(\mb M_{t-H_t:t-1};  \theta_*)-g_i^x(\mb M_t;  \theta_*))}_{\text{policy variation errors}},\\
		D_{u,j}^\top u_t \leq \ & \underbrace{g_j^u(\mb M_t)}_{\text{action constraint function}}+\underbrace{ D_{u,j}^\top\eta_{t}}_{\text{excitation error on the action}}.
	\end{align*}
	where  $\hat \theta_t^g$ is an estimated model used to approximate the state constraint function, and we allow $\hat \theta_t^g\not = \hat \theta_t$ for generality.
\end{corollary}
\begin{proof}
	The proof is by the definitions of $g_i^x, g_j^u, \tilde g_i^x$, and Lemma \ref{lem: xt formula DAP uncertain}. For the action constraints, by the definition \eqref{equ: time varying DAP uncertain}, we have
	\begin{align*}
		D_{u,j}^\top u_t &= \sum_{t=1}^{H_t} D_{u,j}^\top M_t[k] \hat w_{t-k}+D_{u,j}^\top \eta_t \leq \sum_{t=1}^{H_t} \|D_{u,j}^\top M_t[k] \|_1 w_{\max}+D_{u,j}^\top \eta_t=g_j^u(\mb M_t)+D_{u,j}^\top \eta_t
	\end{align*}
	where the inequality is because $\hat w_{t-k}\in \W$ and the H\"{o}lder's inequality.  The state constraints can be similarly proved: notice that we apply  the H\"{o}lder's inequality on $\hat w_{t-k}$ instead of $w_{t-k}$. 
\end{proof}
%\begin{proof}
%	\red{to be typed}
%\end{proof}

\subsection{The Constraint-tightening Terms}\label{append: constraint tightening}
This subsection provides the definitions of the factors $c_1, c_2, c_3$ in the constraint-tightening terms introduced in Lemma \ref{lem: constraint tightening terms}. Further, this subsection provides explanations on all the constraint-tightening terms in \eqref{equ: CCE} by showing that each constraint-tightening term serves as an upper bound on an error term in the constraint decompositions in Corollary \ref{cor: state constraint decomposition}.

%The constraint-tightening terms $\epsilon_\theta, \epsilon_{\eta,x}, \epsilon_{\eta,u}$
%We first summarize the constraint-tightening terms below, then provide intuitions on each term by showing the 

\paragraph{Definition and explanation of $\epsilon_\theta(r)$.} The next lemma formally shows that the constraint-tightening term $\epsilon_\theta(r)$ is an upper bound on the model estimation errors in the state constraint decomposition in Corollary \ref{cor: estimation error of our algo}, where $r$ is the model estimation error bound.
\begin{lemma}[Definition of $\epsilon_\theta(r)$]\label{lem: Definition of epsilon theta(r)}
	Consider implementing the time-varying approximate DAP \eqref{equ: time varying DAP uncertain} to the system $x_{t+1}=A_* x_t +B_* u_t +w_t$. For a fixed $t$, suppose $\hat \theta_{t-k}, \hat \theta_t^g\in \Theta_{\text{ini}}$, $\|\hat \theta_{t}^g-\theta_*\|_F \leq r$, and $\|\hat \theta_{t-k} -\theta_*\|_F \leq r$ for all $1\leq k \leq H_t$. Further,  suppose $x_{t-k} \in \X, u_{t-k}\in \U$ for all $1\leq k \leq H_t$. Then, we have
	\begin{align*}
		&g_i^x(\mb M_t;  \theta_*)-g_i^x(\mb M_t; \hat \theta_t^g) \leq \epsilon_{\hat \theta}(r),\quad  \sum_{k=1}^{H_t} D_{x,i}^\top A_*^{k-1}(w_{t-k}-\hat w_{t-k}) \leq \epsilon_{\hat w}(r),\\
		&\underbrace{(g_i^x(\mb M_t;  \theta_*)-g_i^x(\mb M_t; \hat \theta_t^g))+ \sum_{k=1}^{H_t} D_{x,i}^\top A_*^{k-1}(w_{t-k}-\hat w_{t-k})}_{\text{model estimation errors}} 
		\leq  \epsilon_\theta(r)
	\end{align*}
	where $\epsilon_{\hat w}(r)= \|D_x\|_\infty z_{\max} \kappa/\gamma \cdot r =O(r)$, $\epsilon_{\hat \theta}(r)= 5\kappa^4 \kappa_B \|D_x\|_\infty w_{\max}/\gamma^3 \sqrt{mn} r =O(\sqrt{mn} r)$,  and 
	$\epsilon_\theta(r)=\epsilon_{\hat \theta}(r)+\epsilon_{\hat w}(r)=O(\sqrt{mn}) r$. We can let $c_1=\|D_x\|_\infty z_{\max} \kappa/\gamma+5\kappa^4 \kappa_B \|D_x\|_\infty w_{\max}/\gamma^3$.
\end{lemma}
The proof of Lemma \ref{lem: Definition of epsilon theta(r)} is based on the perturbation analysis and deferred to Appendix \ref{subsec: proof of Lemma epsilon theta}.
When proving Lemma \ref{lem: Definition of epsilon theta(r)}, we also establish the following lemma.
\begin{lemma}[Disturbance approximation error]\label{lem: what - w}
	Consider $\hat w_t =\Pi_{\mathbb W}(x_{t+1}-\hat \theta z_t)$ and $x_{t+1}=\theta_* z_t +w_t$. Suppose $\|z_t\|_2\leq b_z$ and $\|\theta_*-\hat \theta\|_F \leq r$, then 
	$$\|w_t -\hat w_t\|_2 \leq b_z r$$
\end{lemma}
\begin{proof}
	By non-expansiveness of projection, we have
	$
	\|w_t -\hat w_t\|_2\leq \| x_{t+1}-\theta_* z_t-(x_{t+1}-\hat \theta z_t)\|_2=\|(\hat \theta-\theta_*)z_t\|_2\leq b_z r.
	$
\end{proof}

\paragraph{Definition and explanation of $\epsilon_{\eta,x}(\bar \eta)$ and $\epsilon_{\eta,u}(\bar \eta)$} The next lemma formally shows that the terms $\epsilon_{\eta,x}(\bar \eta)$ and $\epsilon_{\eta,u}(\bar \eta)$ bounds the excitation errors on the state and action constraint decompositions in Corollary \ref{cor: state constraint decomposition}.

\begin{lemma}[Definition of $\epsilon_\eta(\bar \eta)$]\label{lem: Definition of epsilon eta}
	Consider implementing the time-varying approximate DAP \eqref{equ: time varying DAP uncertain} to the system $x_{t+1}=A_* x_t +B_* u_t +w_t$. For a fixed $t$, suppose $\|\eta_t\|_\infty \leq \bar \eta$  for all $0\leq k \leq H_t$. Then, 
	\begin{align*}
		&\underbrace{\sum_{i=1}^{H_t} D_{x,i}^\top A_*^{i-1}B_*\eta_{t-i}}_{\text{excitation errors on state}}
		\leq \epsilon_{\eta, x}(\bar \eta), \qquad \underbrace{D_{u,j}^\top \eta_t}_{\text{excitation errors on actions}}\leq \epsilon_{\eta, u}(\bar \eta),
	\end{align*}
	where 
	$	\epsilon_{\eta,x}=\|D_x\|_\infty \kappa\kappa_B/\gamma\sqrt m \bar \eta=O(\sqrt m \bar \eta)$, 
	$	\epsilon_{\eta,u}=\|D_u\|_\infty\bar \eta=O(\bar \eta)$, and we define $\epsilon_\eta=(\epsilon_{\eta,x}, \epsilon_{\eta,u})$, $c_2=\|D_x\|_\infty \kappa\kappa_B/\gamma$, $c_3=\|D_u\|_\infty$.
	
	%$\epsilon_\theta(r)=\|D_x\|_\infty z_{\max} \kappa/\gamma  r + 5\kappa^4 \kappa_B \|D_x\|_\infty w_{\max}/\gamma^3 \sqrt{mn} r =O(\sqrt{mn} r)$.
\end{lemma}
\begin{proof} The proof is provided below.
	\begin{align*}
		\|D_x\sum_{i=1}^{H_t} A_*^{i-1} B_* \eta_{t-i}\|_\infty &\leq \|D_x\|_\infty \sum_{i=1}^{H_t} \|A_*^{i-1} B_*\|_\infty \| \eta_{t-i}\|_\infty \leq \|D_x\|_\infty \sqrt m \sum_{i=1}^{H_t} \|A_*^{i-1} B_*\|_2 \| \eta_{t-i}\|_\infty\\
		& \leq \|D_x\|_\infty \sqrt m \sum_{i=1}^{H_t} \kappa (1-\gamma)^{i-1}\kappa_B \| \eta_{t-i}\|_\infty\leq  \|D_x\|_\infty \sqrt m\kappa \kappa_B/\gamma\bar \eta\\
		\|D_u \eta_t \|_\infty &\leq\|D_u\|_\infty \|\eta_t\|_\infty\leq \|D_u\|_\infty\bar \eta
	\end{align*}
\end{proof}

\paragraph{Definition of $\epsilon_H(H)$}
The  term $\epsilon_H(H)$  has been introduced in \cite{li2020online} for the known-model case to bound the history truncation errors in the state constraint decomposition. Here, we slightly improve its dependence on the problem dimensions and include our proof below.

\begin{lemma}[Definition of $\epsilon_H$]\label{lem: define epsilonH}
	For any $x_{t-H_t}\in \mathbb X$, we have
	$$\underbrace{D_{x,i}^\top A_*^{H_t} x_{t-H_t}}_{\text{history truncation errors}} \leq \epsilon_H(H_t)= \|D_x\|_\infty\kappa x_{\max}(1-\gamma)^{H_t}=O((1-\gamma)^{H_t}).$$
	
\end{lemma}
\begin{proof}
	\begin{align*}
		\|D_x A_*^{H_t} x_{t-H_t}\|_\infty& \leq \|D_x\|_\infty \|A_*^{H_t} x_{t-H_t}\|_\infty\leq \|D_x\|_\infty \|A_*^{H_t} x_{t-H_t}\|_2\\
		& \leq  \|D_x\|_\infty\| A_*^{H_t}\|_2 \| x_{t-H_t}\|_2\leq  \|D_x\|_\infty\kappa (1-\gamma)^{H_t} x_{\max}.
	\end{align*}
\end{proof}

\paragraph{Definition of $\epsilon_v(\Delta_M, H)$}
The error term $\epsilon_v(\Delta_M, H)$ has also been introduced in \cite{li2020online} for the known-model case to bound the policy variation error. Here, we also slightly improve its dependence on the problem dimensions and the memory length in the next lemma. The proof is based on the perturbation analysis and will be provided in Appendix \ref{subsec: proof of epsilonv}.

\begin{lemma}[Definition of $\epsilon_v(\Delta_M, H)$]\label{lem: define epsilon V}
	Under the conditions in Lemma \ref{lem: xt formula DAP uncertain}, suppose $\Delta_{M}\geq \max_{1\leq k \leq H_t}\frac{\|\mb M_t-\mb M_{t-k}\|_F}{k}$, then we have
	\begin{align*}
		\underbrace{ (\tilde g_i^x(\mb M_{t-H_t:t-1};  \theta_*)-g_i^x(\mb M_t;  \theta_*))}_{\text{policy variation errors}}
		\leq \epsilon_v(\Delta_M, H_t)%= |D_x\|_\infty w_{\max}\kappa\kappa_B/\gamma^2\sqrt{mnH_t}\Delta_{M},
	\end{align*}
	where $\epsilon_v(\Delta_M, H_t)=\|D_x\|_\infty w_{\max}\kappa\kappa_B/\gamma^2\sqrt{mnH_t}\Delta_{M}=O(\sqrt{mnH_t}\Delta_{M})$.
	%	We define $\epsilon_v(\Delta_M)=c_v\sqrt{mnH}\Delta_{M}$.
\end{lemma}

\section{Estimation Error Bounds}\label{append: estimate error}

This section provides a proof for Theorem \ref{thm: general estimation error bdd} and a proof for Corollary \ref{cor: estimation error of our algo}. When proving Corollary \ref{cor: estimation error of our algo}, we also establishes a.s. upper bounds on the state and action trajectories of our algorithm.

%In this appendix, we provide proofs of Theorem \ref{thm: general estimation error bdd} and Corollary \ref{cor: estimation error of our algo}.
\subsection{Proof of Theorem \ref{thm: general estimation error bdd}}
Our proof of Theorem \ref{thm: general estimation error bdd} relies on a recently developed least square estimation error bound for general time series satisfying a block matingale small-ball (BMSB) condition  \cite{simchowitz2018learning}.  The general error bound and the definition of BMSB are included below for completeness. In the literature \cite{dean2018regret,dean2019safely}, only linear policies are considered and shown to satisfy the BMSB condition. Our contribution is to show that even for general policies, BMSB still holds as long as the corresponding states and actions are bounded (which is usually the case if certain stability properties are satisfied). By general policies, we allow time-varying policies, nonlinear policies, policies that depend on all the history, etc., (i.e. we consider $u_t=\pi_t(x_0, \{w_k,\eta_k\}_{k=0}^{t-1})+\eta_t$). More rigorous discussions are provided below.

\begin{definition}[Block Martingale Small-Ball (BMSB) (Definition 2.1 \cite{simchowitz2018learning})]\label{def: bmsb}
	%	Let $\{Z_t\}_{t\geq 1}$ be an $\{\F_t\}_{t\geq 1}$-adapted random process taking values in $\R$. We say $\{Z_t\}_{t\geq 1}$ satisfies the $(k, v,p)$-block martingale small-ball (BMSB) condition if, for any $j\geq 0$, one has $\frac{1}{k}\sum_{i=1}^k \Pb(|Z_{j+i}|\geq v\mid \F_t)\geq p$ almost surely. 
	Let $\{X_t\}_{t\geq 1}$ be an $\{\F_t\}_{t\geq 1}$-adapted random process taking values in $\R^d$.
	%Given an $\{\F_t\}_{t\geq 1}$-adapted random process $\{X_t\}_{t\geq 1}$ taking values in $\R^d$, 
	We say that it satisfies the $(k, \Gamma_{sb},p)$-block martingale small-ball (BMSB) condition for $\Gamma_{sb}>0$ if, for any fixed $\lambda \in\R^d$ such that $\|\lambda\|_2=1$ and for any $j\geq 0$, one has $\frac{1}{k}\sum_{i=1}^k \Pb(|\lambda^\top X_{j+i}|\geq \sqrt{\lambda^\top\Gamma_{sb} \lambda}\mid \F_j)\geq p$ almost surely. 
\end{definition}

%\red{To do: check the proof correctness and check if the proof actually proves Theorem \ref{thm: general estimation error bdd}.}

\begin{theorem}[Theorem 2.4 in \cite{simchowitz2018learning}]\label{thm: theorem 2.4 general estimation error}
	Fix $\epsilon\in (0,1)$, $\delta \in (0, 1/3)$, $T\geq 1$, and $0<\Gamma_{sb}<\bar \Gamma$. Consider a random process $\{X_t, Y_t\}_{t\geq 1}\in (\R^d\times \R^n)^T$ and a filtration $\{\F_t\}_{t\geq 1}$. 
	Suppose the following conditions hold,
	\begin{enumerate}
		\item $Y_t=\theta_* X_t +\eta_t$, where $\eta_t\mid \F_t$ is $\sigma_{sub}^2$-sub-Gaussian and mean zero,
		\item $\{X_t\}_{t\geq 1}$ is an $\{\F_t\}_{t\geq 1}$-adapted random process satisfying the $(k, \Gamma_{sb},p)$-block martingale small-ball (BMSB) condition,
		\item $\Pb(\sum_{t=1}^T X_t X_t^\top \not \leq T \bar \Gamma)\leq \delta$.
	\end{enumerate}
	Define the (ordinary) least square estimator as
$
		\tilde \theta=\argmin_{\theta \in \R^{n \times d}} \sum_{t=1}^T \|Y_t-\theta X_t\|_2^2
$.
	Then if 
$$
		T\geq \frac{10k}{p^2}\left(\log(\frac{1}{\delta})+2d\log(10/p)+\log \det(\bar \Gamma \Gamma_{sb}^{-1})\right),
$$
	we have
	\begin{equation*}
		\|\tilde \theta-\theta_*\|_2 \leq \frac{90\sigma_{sub}}{p}\sqrt{\frac{n+d\log(10/p)+\log\det(\bar \Gamma \Gamma_{sb}^{-1})+\log(1/\delta)}{T \upsilon_{\min}(\Gamma_{sb})}}
	\end{equation*}
	with probability at least $1-3\delta$.
\end{theorem}

Next, we  present a proof for our Theorem \ref{thm: general estimation error bdd} by verifying the conditions in Theorem \ref{thm: theorem 2.4 general estimation error} for general nonlinear policies. 
\begin{proof}[Proof of Theorem \ref{thm: general estimation error bdd}]
%	To use Theorem \ref{thm: theorem 2.4 general estimation error}, we need to verify the three conditions. 
	
	Condition 1 is straightforward: $x_{t+1}=\theta_* z_t +w_t$, and $w_t\mid \F_t=w_t$ which is mean 0 and $\sigma_{sub}^2$-sub-Gaussian by Assumption \ref{ass: w Sigma}. 
	Condition 3 is also straightforward. Notice that 
	$\upsilon_{\max}(z_t z_t^\top )\leq \text{trace}(z_t z_t^\top) = \|z_t\|_2^2\leq b_z^2.$
	Therefore, we can define $\bar \Gamma=b_z^2I_{n+m}$, and then $\Pb(\sum_{t=1}^T z_t z_t^\top \not \leq T \bar \Gamma)=0\leq \delta$.
	
	The tricky part is Condition 2. 
	Next, we will show the BMSB condition holds for our system. Then, by Theorem \ref{thm: theorem 2.4 general estimation error}, we complete the proof.
	\begin{lemma}[Verification of BMSB condition]\label{lem: verify bmsb strictly safe perturbed policy}
		Define filtration $\F_t=\{w_0, \dots, w_{t-1}, \eta_0, \dots,\eta_t\}$. Under the conditions in Theorem \ref{thm: general estimation error bdd}, 	\begin{equation*}
			\{z_t\}_{t\geq 0} \ \text{ satisfies the $(1, s_z^2 I_{n+m}, p_z)$-BMSB condition,}
		\end{equation*}
		where $p_z=\min(p_w, p_{\eta})$, $s_z=\min(s_w/4, \frac{\sqrt 3}{2}s_\eta\bar \eta, \frac{s_ws_\eta}{4b_u} \bar \eta)$.
		
%		
%		
%		Define filtrations $\F_t^m=\F(w_0, \dots, w_{t-1},\eta_0, \dots, \eta_{t-1})$ and $\F_t=\{w_0, \dots, w_{t-1}, \eta_0, \dots,\eta_t\}$. Notice that $u_t=\pi_t(\F_t^m)+\eta_t$ as defined in Theorem \ref{thm: general estimation error bdd}. 
%		
%		
%		
%		Consider $x_{t+1}=A_*x_t+B_*u_t+w_t$, where $u_t=\pi_t(\F_t^m)+\eta_t$, and $\F_t^m=\F(w_0, \dots, w_{t-1},\eta_0, \dots, \eta_{t-1})$. Consider $w_t$ i.i.d. and $(s_w, p_w)$-anti-concentration. Consider $\eta_t\!\overset{\text{i.i.d.}}{\sim}\! \bar \eta \D_\eta$ and  $ \eta_t/\bar \eta$ satisfies the $(s_{\eta}, p_\eta)$-anti-concentration property. Suppose $w_t$ is $\sigma_{sub}^2$-subGaussian and has zero mean. Consider $\eta_t, w_t$ to be independent for all $t$. Consider general policies:
%		\begin{equation}\label{equ: general policy}
%			u_t=\pi_t(\F_t^m)+\eta_t, \quad t\geq 0.
%		\end{equation} 
%		Suppose we have $\|x_t\|_2\leq b_x$ and $\|u_t\|_2\leq b_u$ for some $b_x, b_u$ for all $t$ under policies \eqref{equ: general policy}. Define $\tilde  \theta=\min_\theta \sum_{t=0}^{T-1}\|x_{t+1}-Ax_t-Bu_t\|_2^2$. Define $\F_t=\{w_0, \dots, w_{t-1}, \eta_0, \dots,\eta_t\}$. Then we have
%		\begin{equation}
%			\{z_t\}_{t\geq 0} \ \text{ satisfies the $(1, s_z^2 I_{n+m}, p_z)$-BMSB condition,}
%		\end{equation}
%		where $p_z=\min(p_w, p_{\eta})$, $s_z=\min(s_w/4, \frac{\sqrt 3}{2}s_\eta\bar \eta, \frac{s_ws_\eta}{4b_u} \bar \eta)$.
	\end{lemma}

	\begin{proof}[Proof of Lemma \ref{lem: verify bmsb strictly safe perturbed policy}]
			Define filtration $\F_t^m=\F(w_0, \dots, w_{t-1},\eta_0, \dots, \eta_{t-1})$. Notice that the policy in Theorem \ref{thm: general estimation error bdd} can be written as $u_t=\pi_t(\F_t^m)+\eta_t$.
		Note that $z_t \in \F_t$ is by definition. Next,
		\begin{align*}
			z_{t+1}\mid \F_t=	
			\begin{bmatrix}
				x_{t+1}\\
				u_{t+1}
			\end{bmatrix}\mid \F_t= 
			\begin{bmatrix}
				\theta_* z_t+w_t \mid \F_t\\
				\pi_{t+1}(\F_{t+1}^m)+\eta_{t+1}\mid \F_t
			\end{bmatrix},
		\end{align*}
		where $\F_{t+1}^m=\F(w_0, \dots, w_{t},\eta_0, \dots, \eta_{t})$.
		
		Notice that conditioning on $\F_t$, the variable $\theta_* z_t$ is determined, but the variable $	\pi_{t+1}(\F_{t+1}^m)$ is still random due to the randomness of $w_t$. 
		{For the rest of the proof, we will always condition on $\F_t$, and omit the conditioning notation,  i.e., $\cdot\mid \F_t$,  for notational simplicity.}
		
		Consider any $\lambda=(\lambda_1^\top, \lambda_2^\top)^\top\in \R^{m+n}$, where $\lambda_1\in \R^n$, $\lambda_2\in \R^m$, $\|\lambda\|_2^2=\|\lambda_1\|_2^2+\|\lambda_2\|_2^2=1$. Define $k_0=\max(2/\sqrt 3, 4b_u/s_w)$. We consider three cases: (i) when $\|\lambda_2\|_2\leq 1/k_0$ and $\lambda_1^\top \theta_* z_t \geq 0$, (ii) when $\|\lambda_2\|_2\leq 1/k_0$ and $\lambda_1^\top \theta_* z_t < 0$, (iii) when $\|\lambda_2\|_2> 1/k_0$. We will show in all three cases,
		$$\Pb(|\lambda^\top z_{t+1}| \geq s_z)\geq p_z$$
		Consequently, by Definition 2.1 in \cite{simchowitz2018learning}, we have $\{z_t\}$ is $(1, s_z^2 I, p_z)$-BMSB.
		
		\nbf{Case 1: when $\|\lambda_2\|_2\leq 1/k_0$ and $\lambda_1^\top \theta_* z_t \geq 0$}
		
		\begin{align*}
			\lambda_1^\top w_t& \leq \lambda_1^\top(w_t+\theta_*z_t)\leq | \lambda_1^\top(w_t+\theta_*z_t)|\\
			&=| \lambda^\top z_{t+1}-\lambda_2^\top u_{t+1}|
			\leq | \lambda^\top z_{t+1}|+|\lambda_2^\top u_{t+1}|\leq | \lambda^\top z_{t+1}|+ \|\lambda_2\|_2 b_u\\
			& \leq  | \lambda^\top z_{t+1}|+b_u/k_0\leq  | \lambda^\top z_{t+1}|+s_w/4
		\end{align*}
		where the last inequality uses $k_0\geq 4b_u/s_w$.
		
		Further, notice that $k_0\geq 2/\sqrt 3$, so $\|\lambda_2\|_2^2\leq 1/k_0^2\leq 3/4$, thus, $\|\lambda_1\|_2^2 \geq 1/4$, which means $\|\lambda_1\|_2\geq 1/2$. Therefore,
		\begin{align*}
			\Pb(\lambda_1^\top w_t\geq s_w/2)&=\Pb(\frac{\lambda_1^\top w_t}{\|\lambda_1\|_2}\geq \frac{s_w}{2\|\lambda_1\|_2}) \geq \Pb(\frac{\lambda_1^\top w_t}{\|\lambda_1\|_2}\geq s_w)=p_w
		\end{align*}
		Then,
		\begin{align*}
			\Pb(|\lambda^\top z_{t+1}| \geq s_z)&\geq \Pb(|\lambda^\top z_{t+1}| \geq s_w/4)= \Pb(|\lambda^\top z_{t+1}|+s_w/4 \geq s_w/2)\\
			&\geq \Pb(\lambda_1^\top w_t\geq s_w/2)\geq p_w
		\end{align*}
		which completes case 1.
		
		\nbf{Case 2: when $\|\lambda_2\|_2\leq 1/k_0$ and $\lambda_1^\top \theta_* z_t < 0$.}
		\begin{align*}
			\lambda_1^\top w_t& \geq \lambda_1^\top(w_t+\theta_*z_t)\geq -| \lambda_1^\top(w_t+\theta_*z_t)|\\
			&=-| \lambda^\top z_{t+1}-\lambda_2^\top u_{t+1}|
			\geq -| \lambda^\top z_{t+1}|-|\lambda_2^\top u_{t+1}|\geq -| \lambda^\top z_{t+1}|- \|\lambda_2\|_2 b_u\\
			& \geq - | \lambda^\top z_{t+1}|-b_u/k_0\geq  -| \lambda^\top z_{t+1}|-s_w/4
		\end{align*}
		where the last inequality uses $k_0\geq 4b_u/s_w$.
		
		Further, notice that $k_0\geq 2/\sqrt 3$, so $\|\lambda_2\|_2^2\leq 1/k_0^2\leq 3/4$, thus, $\|\lambda_1\|_2^2 \geq 1/4$, which means $\|\lambda_1\|_2\geq 1/2$. Therefore,
		\begin{align*}
			\Pb(\lambda_1^\top w_t\leq -s_w/2)&=\Pb(\frac{\lambda_1^\top w_t}{\|\lambda_1\|_2}\leq -\frac{s_w}{2\|\lambda_1\|_2}) \geq \Pb(\frac{\lambda_1^\top w_t}{\|\lambda_1\|_2}\leq -s_w)=\Pb(\frac{-\lambda_1^\top w_t}{\|\lambda_1\|_2}\geq s_w)=p_w
		\end{align*}
		by $s_w/(2\|\lambda_1\|_2)\leq s_w$, and thus $-s_w/(2\|\lambda_1\|_2)\geq -s_w$, and Assumption \ref{ass: w Sigma}.
		
		Consequently,
		\begin{align*}
			\Pb(|\lambda^\top z_{t+1}| \geq s_z)&\geq \Pb(|\lambda^\top z_{t+1}| \geq s_w/4)= \Pb(-|\lambda^\top z_{t+1}|-s_w/4 \leq -s_w/2)\\
			&\geq \Pb(\lambda_1^\top w_t\leq -s_w/2)\geq p_w
		\end{align*}
		which completes case 2.
		
		\nbf{Case 3: when $\|\lambda_2\|_2> 1/k_0$.} Define $v=\bar \eta s_\eta/k_0=\min(\sqrt 3 \bar \eta s_\eta/2, s_w \bar \eta s_\eta/(4 b_u))$. Define
		\begin{align*}
			\Omega_1^\lambda&=\{w_t\in \R^n\mid \lambda_1^\top(w_t+\theta_*z_t)+\lambda_2^\top(\pi_{t+1}(\F_{t+1}^m))\geq 0\}\\
			\Omega_2^\lambda&=\{w_t\in \R^n\mid \lambda_1^\top(w_t+\theta_*z_t)+\lambda_2^\top(\pi_{t+1}(\F_{t+1}^m))< 0\}\\
		\end{align*}
		Notice that $\Pb( w_t \in \Omega_1^\lambda)+\Pb( w_t \in \Omega_2^\lambda)=1$.
		\begin{align*}
			\Pb(|\lambda^\top z_{t+1}| \geq s_z)&\geq \Pb(|\lambda^\top z_{t+1}| \geq v)= \Pb(\lambda^\top z_{t+1} \geq v)+ \Pb(\lambda^\top z_{t+1} \leq -v)\\
			& \geq \Pb(\lambda^\top z_{t+1} \geq v, w_t \in \Omega_1^\lambda)+ \Pb(\lambda^\top z_{t+1} \leq -v, w_t \in \Omega_2^\lambda)\\
			& \geq \Pb(\lambda_2^\top \eta_{t+1} \geq v, w_t \in \Omega_1^\lambda)+\Pb(\lambda_2^\top \eta_{t+1} \leq -v, w_t \in \Omega_2^\lambda)\\
			&=\Pb(\lambda_2^\top \eta_{t+1} \geq v)\Pb( w_t \in \Omega_1^\lambda)+\Pb(\lambda_2^\top \eta_{t+1} \leq -v)\Pb( w_t \in \Omega_2^\lambda)\\
			&\geq p_\eta
		\end{align*}
		where the last inequality is because of the following arguments.
		Notice that
		\begin{align*}
			\Pb(\lambda_2^\top \eta_{t+1} \geq v)&= \Pb(\lambda_2^\top \eta_{t+1}/\|\lambda_{2}\|_2 \geq v/\|\lambda_{2}\|_2)\\
			& =\Pb(\lambda_2^\top \tilde \eta_{t+1}/\|\lambda_{2}\|_2 \geq v/(\|\lambda_{2}\|_2\bar \eta))\\
			& \geq \Pb(\lambda_2^\top \tilde \eta_{t+1}/\|\lambda_{2}\|_2 \geq k_0v/(\bar \eta))\\
			& = \Pb(\lambda_2^\top \tilde \eta_{t+1}/\|\lambda_{2}\|_2 \geq  s_\eta)\geq p_\eta
		\end{align*}
		Then,
		$$	\Pb(\lambda_2^\top \eta_{t+1} \leq -v)=	\Pb(-\lambda_2^\top \eta_{t+1} \geq v)\geq p_\eta$$
		This completes the proof of Case 3.
	\end{proof}
Finally, we apply Theorem \ref{thm: theorem 2.4 general estimation error}. Notice that $d=m+n$ in our problem, and $\log\det(\bar \Gamma \Gamma_{sb}^{-1})=2(m+n)\log(b_z/s_z)=O((m+n)\log(b_z/\bar\eta))$ as $\bar \eta\to 0$,  $\upsilon_{\min}(\Gamma_{sb})=s_z^2=O(1/\bar\eta^2)$ as $\bar \eta\to 0$, and $p=p_z$ here. Therefore, for $T$ large enough, we have:
$$\|\tilde \theta_T-\theta_*\|_2 \leq \! O\!\left(\!\sqrt{n\!+\!m}\frac{\sqrt{\log(b_z/\bar\eta+1/\delta)}}{\sqrt T \bar \eta}\!\right).$$
\end{proof}

\subsection{Proof of Corollary \ref{cor: estimation error of our algo}}

%\red{haven't revised, basically I want to say: notice that there is an extra mn term in the log, this is because bz is not O(1), but O(mn), and we show this below. the formal proof is below: ut is bounded by algebraic,  xt can also be bounded also in a few lines, list here to avoid circular argument, then we are done.} 

Corollary \ref{cor: estimation error of our algo} follows directly from Theorem \ref{thm: general estimation error bdd}. We only need to verify the boundedness of the states and actions. In the following, we will show that $u_t \in \U$ for all $t$ and $\|x_t\|_2\leq O(\sqrt{mn})$ for all $t$. Notice that though we can further show a much smaller bound $\|x_t \|_2\leq x_{\max}$ with  probability $(1-p)$ in Theorem \ref{thm: constraint satisfaction}, Theorem \ref{thm: general estimation error bdd} requires an almost sure bound and thus we provide a larger bound $\|x_t\|_2\leq O(\sqrt{mn})$ here. %Based on the bounds on $x_t, u_t$, and by choosing $\delta=\frac{p}{6e^2}$, we have \eqref{equ: def of re}, where the additional $mn$ factor in $\log(\sqrt{mn}/\bar \eta^{(e-1)})$ comes from $b_z=O(\sqrt{mn})$.

In the following, we show that $u_t \in \U$ for all $t$ and $\|x_t\|_2\leq O(\sqrt{mn})$ for all $t$.

%We prove Corollary \ref{cor: estimation error of our algo} by verifying that Algorithm \ref{alg: online algo} satisfies the conditions in Theorem \ref{thm: general estimation error bdd}. The most tricky part is to provide  almost surely bounds on the generated trajectories $x_t$ and $u_t$. We are able to show that $u_t\in \U$ almost surely, but we cannot show $x_t\in\X$ almost surely. Nevertheless, we are able to show that $\|x_t\|_2\leq O(\sqrt{mn})$ almost surely by leveraging the condition $\mb M\in \M_H$ for proper $H$. In the following, we first show $u_t \in\U$ almost surely. Then, we show $\|x_t\|_2\leq O(\sqrt{mn})$ almost surely. Finally, we prove Corollary \ref{cor: estimation error of our algo}.

\begin{lemma}[Action constraint satisfaction]\label{lem: ut satisfies constraints}
	When applying Algorithm \ref{alg: online algo}, $u_t\in\U$ for all $t$ and for any $w_k\in\W$.
\end{lemma}
\begin{proof}Notice that  $u_t=\sum_{k=1}^{H^{(e-1)}}M_t[k]\hat w_{t-k} +\eta_t$. Hence, for any $1\leq j \leq k_u$, we have
			\begin{align*}
			D_{u,j}^\top u_t&=D_{u,j}^\top \sum_{k=1}^{H^{(e-1)}}M_t[k]\hat w_{t-k} +D_{u,j}^\top \eta_t \leq \sum_{k=1}^{H^{(e-1)}} \|D_{u,j}^\top M_t[k]\|_1w_{\max} +\|D_u\|_\infty \|\eta_t\|_\infty =g_j^u(\mb M_t)+\|D_u\|_\infty \|\eta_t\|_\infty
		\end{align*}
	Our goal is to show that $g_j^u(\mb M_t)+\|D_u\|_\infty \|\eta_t\|_\infty\leq d_{u,j}$ for all $j$ and for all $t\geq 0$. This is straightforward when $t_1^{(e)}\leq t \leq t_1^{(e)}+T_D^{(e)}-1$ and $t_2^{(e)}\leq t \leq T^{(e+1)}-1$. For example, when $t_1^{(e)}\leq t \leq t_1^{(e)}+T_D^{(e)}-1$, we have $\mb M_t=\mb M^{(e)}_\dagger$ and $\|\eta_t\|_\infty \leq \bar \eta^{(e)} $, which leads to $g_j^u(\mb M^{(e)}_\dagger)+\|D_u\|_\infty \|\eta_t\|_\infty=g_j^u(\mb M^{(e)}_\dagger)+c_3\|\bar \eta^{(e)}\|_\infty\leq d_{u,j}$ by \texttt{RobustCE} and Lemma \ref{lem: Definition of epsilon eta}. Similar results can be shown for $t_2^{(e)}\leq t \leq T^{(e+1)}-1$.
	
	Next, we focus on the safe policy transition stages. It suffices to show that $u_t=\sum_{k=1}^{H^{(e-1)}}M_t[k]\hat w_{t-k} +\eta_t$ in all stages of Algorithm \ref{alg: safe transit}. In the following, we will adopt the notations in Algorithm \ref{alg: safe transit}. In Step 1 of Algorithm \ref{alg: safe transit}, we have $\|\eta_t\|_\infty \leq\bar \eta_{\min}\leq \bar \eta$ and $\mb M_t\in \Omega$ by the convexity of $\Omega$. Therefore, we have $g_j^u(\mb M_t)+\|D_u\|_\infty \|\eta_t\|_\infty=g_j^u(\mb M_t)+c_3\|\bar \eta\|_\infty\leq d_{u,j}$, where we used the definition of $\Omega$ in \texttt{RobustCE}. In Step 2, we have $\|\eta_t\|_\infty \leq \bar \eta'$ and $\mb M_t\in \Omega'$ by the convexity of $\Omega'$. Therefore, we have $g_j^u(\mb M_t)+\|D_u\|_\infty \|\eta_t\|_\infty=g_j^u(\mb M_t)+c_3\|\bar \eta'\|_\infty\leq d_{u,j}$, where we used the definition of $\Omega'$ by \texttt{RobustCE} with input $\bar \eta'$. 
\end{proof}

\begin{lemma}[Almost sure upper bound on $x_t$]\label{lem: w.p.1 bdd on xt}
	Consider DAP policy $u_t=\sum_{k=1}^{H_t} M_t[k]\hat w_{t-k}+\eta_t$, where $\mb M_t\in \M_{H_t}$, $\{H_t\}_{t\geq 0}$ is non-decreasing,  and  $\|\eta_t\|_\infty \leq \eta_{\max}$. Suppose $H_0 \geq \log(2\kappa)/\log((1-\gamma)^{-1})$ and $\eta_{\max} \leq w_{\max}/\kappa_B$. Let $\{x_t, u_t\}_{t\geq 0}$ denote the trajectory generated by this policy on the  system with parameter $\theta_*$ and disturbance $w_t$.  Then, there exists $b_x=4\sqrt n \kappa w_{\max}/\gamma+ 4\sqrt{mn} \kappa^3\kappa_B w_{\max}/\gamma^2= O(\sqrt{mn})$ such that 
	$$\|x_t\|_2\leq b_x,  \quad \forall\, t\geq 0, \quad \forall \ w_k, \, \hat w_k \in \W.$$
	
	%	Further, if $\mb M_t$ also satisfies $  g_j^u(\mb M_t)\leq d_{u,j}-\epsilon_\eta^u(\eta_{\max,t})$ for all $t\geq 0$ for all $j$, then
	%	$u_t\in \mathbb U$ for all $t\geq 0$. Consequently, $\|u_t\|_2\leq u_{\max}$.
\end{lemma}
This lemma is a natural extension of Lemma 2 in \cite{li2020online} and the proof is deferred to Appendix \ref{subsec: as. bdd on xt}.

\begin{proof}[Proof of Corollary \ref{cor: estimation error of our algo}]
	%\red{use $c_9$ but define it, I used number-labeled constants in feasibility in the main.pdf . ???}

	By letting $\delta^{(e)}=\frac{p}{6e^2}$ for $e\geq 1$, we have that $\|\tilde \theta^{(e)}-\theta_*\|_2\leq O(\frac{(\sqrt{m+n})\sqrt{\log(\sqrt{mn}/\bar \eta^{(e-1)})+\log(e)}}{\sqrt{ T^{(e-1)}_D} \bar \eta^{(e-1)}})$ w.p. $1-p/(2e^2)$. Notice that 
$
		\| \hat \theta^{(e)}-\theta_*\|_F\leq  	\| \tilde  \theta^{(e)}-\theta_*\|_F\leq \sqrt n\| \tilde  \theta^{(e)}-\theta_*\|_2,
	$
	which completes the proof.
\end{proof}

\section{Feasibility}\label{append: feasible}
This appendix provides a proof for Theorem \ref{thm: feasibility}. We will first establish the recursive feasibility and then prove the initial feasibility. For notational simplicity, we define $\Omega_0\coloneqq \Omega(\hat \theta^{(0)}, \epsilon_x^{(0)}+\epsilon_0, \epsilon_u^{(0)})$.

	\nbf{Proof of recursive feasibility:} 
 To show that Algorithm 1 and 2 are feasible at all stages, we need to show that $\Omega_\dagger^{(e)}, \Omega^{(e)}, \Omega_\dagger^{(e)}\cap \Omega^{(e)}, \Omega_\dagger^{(e+1)}\cap \Omega^{(e)}$ are all non-empty for $e\geq 0$. Notice that it suffices to show that $\Omega_0\subseteq \Omega_\dagger^{(e)}$ and $\Omega_0\subseteq \Omega^{(e)}$ for all $e\geq 0$. 
 
 Consider $\Omega_\dagger^{(e)}$ for $e\geq 0$. Notice that $\Omega_0\subseteq \Omega_\dagger^{(0)}$ by definition, so we will focus on $e\geq 1$ below. We first consider the action constraints. 
For any $\mb M\in \Omega_0$, we have 
\begin{align*}
	g_j^u(\mb M)\leq d_{u,j}-\epsilon_{\eta, u}(\bar \eta^{(0)}), \quad \forall\, 1\leq j \leq k_u.
\end{align*}
Since $\bar \eta^{(0)}\geq \bar \eta^{(e)}$ by condition (ii) of Theorem \ref{thm: feasibility}, we have $\epsilon_{\eta, u}(\bar \eta^{(0)}) \geq \epsilon_{\eta, u}(\bar \eta^{(e)})$, so $\mb M$ satisfies the action constraints in $\Omega_\dagger^{(e)}$:
\begin{align*}
	g_j^u(\mb M)\leq d_{u,j}-\epsilon_{\eta, u}(\bar \eta^{(e)}) \quad \forall\, 1\leq j \leq k_u.
\end{align*}
Next, we consider the state constraints.
Notice that $\hat \theta^{(e)}\in \Theta_{ini}$ by \texttt{ModelEst}, so $\|\hat \theta^{(e)}-\hat \theta^{(0)}\|_F\leq r^{(0)}$ for $e\geq 1$. By Lemma \ref{lem: Definition of epsilon theta(r)},  for any $\mb M\in \Omega_0$, we have
\begin{align*}
	g_i^x(\mb M; \hat \theta^{(e)})&\leq 	g_i^x(\mb M; \hat \theta^{(0)})+\epsilon_\theta(r^{(0)})\\
	& \leq d_{x,i}-\epsilon_x^{(0)}-\epsilon_0+\epsilon_\theta(r^{(0)})\\
	& = d_{x,i}-\epsilon_\theta(r^{(0)})- \epsilon_{\eta,x}(\bar \eta^{(0)})-\epsilon_H(H^{(0)})-\epsilon_v(\Delta_M^{(0)}, H^{(0)}) -\epsilon_0+\epsilon_\theta(r^{(0)})\\
	&=d_{x,i}- \epsilon_{\eta,x}(\bar \eta^{(0)})-\epsilon_H(H^{(0)})-\epsilon_v(\Delta_M^{(0)}, H^{(0)}) -\epsilon_0
\end{align*}
Further, since $r^{(e)}\leq r^{(1)}\leq \epsilon_0/(c_1\sqrt{mn})$ by condition (ii) of Theorem \ref{thm: feasibility}, we have $\epsilon_\theta(r^{(e)})\leq \epsilon_\theta(r^{(1)})\leq \epsilon_0$, so 
\begin{align*}
	g_i^x(\mb M; \hat \theta^{(e)})&\leq d_{x,i}- \epsilon_{\eta,x}(\bar \eta^{(0)})-\epsilon_H(H^{(0)})-\epsilon_v(\Delta_M^{(0)}, H^{(0)}) -\epsilon_0\\
	& \leq d_{x,i}- \epsilon_{\eta,x}(\bar \eta^{(0)})-\epsilon_H(H^{(0)})-\epsilon_v(\Delta_M^{(0)}, H^{(0)}) -\epsilon_\theta(r^{(e)})\\
	& \leq d_{x,i}- \epsilon_{\eta,x}(\bar \eta^{(e)})-\epsilon_H(H^{(e)})-\epsilon_v(\Delta_M^{(e)}, H^{(e)}) -\epsilon_\theta(r^{(e)})
\end{align*}
 by condition (ii) of Theorem \ref{thm: feasibility}. So $\mb M\in \Omega_\dagger^{(e)}$ for $e\geq 1$. Similarly, we can show  $\mb M\in \Omega^{(e)}$ for $e \geq 0$. This completes the recursive feasibility.

\nbf{Proof of  initial feasibility.} By Lemma 4 and Corollary 2 in \cite{li2020online},  we can construct $\mb M_F$ with length $H^{(0)}$ based on $K_F$ in  Assumption \ref{ass: KF epsilonF} such that 
\begin{align*}
	&g_i^x(\mb M_F;\theta_*)\leq d_{x,i}- \epsilon_{F,x}+\epsilon_{P}(H^{(0)})\\
	&g_i^u(\mb M_F)\leq d_{u,j}- \epsilon_{F,u}+\epsilon_{P}(H^{(0)}).
\end{align*}
where $\epsilon_P$ corresponds to $\epsilon_1+\epsilon_3$ in \cite{li2020online}.\footnote{Notice that $\epsilon_1+\epsilon_3=O(n\sqrt m H (1-\gamma)^H)$ in \cite{li2020online}, but we  improve the bound to $\epsilon_P=O(\sqrt{mn}(1-\gamma)^H)$. Specifically, $\epsilon_3=O(\sqrt n (1-\gamma)^H)$ remains unchanged, but we can show $\epsilon_1(H)=O(\sqrt{mn}(1-\gamma)^H)$. This is because  the proof of Lemma 1 in \cite{li2020online} shows that $\epsilon_1=O(b_x(1-\gamma)^H)$, where $\|x_t\|_2\leq b_x$ a.s.. In Lemma \ref{lem: w.p.1 bdd on xt}, we show $ b_x=O(\sqrt{mn})$ in this paper, so we have $\epsilon_1(H)=O(\sqrt{mn}(1-\gamma)^H)$. \label{footnote: epsilonP}}

Therefore, by Lemma \ref{lem: Definition of epsilon theta(r)}, we have
\begin{align*}
	&g_i^x(\mb M_F;\hat \theta^{(0)})\leq d_{x,i}- \epsilon_{F,x}+\epsilon_{P}(H^{(0)})+\epsilon_\theta(r^{(0)})\quad g_i^u(\mb M_F)\leq d_{u,j}- \epsilon_{F,u}+\epsilon_{P}(H^{(0)}).
\end{align*}
Therefore, $\mb M_F\in \Omega_0$ if \eqref{equ: initial feasibility} in Theorem \ref{thm: feasibility} holds.

\section{Constraint Satisfaction}\label{append: constraint satisfaction}
This section provides a proof for  the constraint satisfaction guarantee in Theorem \ref{thm: constraint satisfaction}. Notice that  the control constraint satisfaction has already been established in Lemma \ref{lem: ut satisfies constraints}. Hence, we will focus on state constraint satisfaction in this appendix. Firstly, we present and prove a  general state constraint satisfaction lemma for time-varying approximate DAPs, which is more general than Lemma \ref{lem: constraint tightening terms}.  Secondly, we prove the state constraint satisfaction of our algorithms by showing that our algorithms satisfy the conditions in the general state constraint satisfaction lemma.
\subsection{A General State Constraint Satisfaction Lemma}
This subsection provides a   general state constraint satisfaction lemma for time-varying approximate DAPs, which includes Lemma \ref{lem: constraint tightening terms} as a special case.

\begin{lemma}[General Constraint Satisfaction Lemma]\label{lem: general constraint satisfaction theta*}
	Consider the time-varying approximate DAPs in \eqref{equ: time varying DAP uncertain}, where $\mb M_t \in \M_{H_t}$ for non-decreasing $\{H_t\}_{t\geq 0}$, $\hat \theta_t\in \Theta_{\textup{ini}}$. Define 
	\begin{align*}
		\epsilon_{H,t}&=(1-\gamma)^{H_t}\cdot\|D_x\|_\infty \kappa x_{\max}\\
				\epsilon_{v,t}&= \sqrt{mnH_t}\Delta_{M,t}\cdot \|D_x\|_\infty w_{\max}\kappa\kappa_B/\gamma^2, \quad \Delta_{M,t}=\max_{1\leq k\leq H_t}\frac{\|\mb M_t-\mb M_{t-k}\|_F}{k}, \\
								\epsilon_{\theta,t}& = c_1 \max_{0\leq k\leq H_t} \|\hat \theta_{t-k}-\theta_*\|_F\\
		\epsilon_{\eta,x,t}&=c_2\sqrt m \max_{1\leq k\leq H_t}\bar \eta_{t-k}, %\\
	%	\epsilon_{\eta,u,t}&=c_3\bar \eta_{t}, 
	\end{align*}
	where $c_1, c_2$ are defined in Lemma \ref{lem: Definition of epsilon theta(r)} and Lemma \ref{lem: Definition of epsilon eta}, and we let $\mb M_t=\mb M_0$, $\bar \eta_t=0$, $H_t=H_0$, $\hat \theta_t=\hat \theta_0$, $w_t=\hat w_t=x_t=0$,  for $t\leq -1$. 
	
	For any $t\geq 0$, if $x_s\in \X, u_s\in\U$ for all $s\leq t-1$ and
	\begin{equation}\label{equ: general state constraint satisfaction}
		g_i^x(\mb M_t;\hat \theta_t)\leq d_{x,i}-\epsilon_{H,t}-\epsilon_{\eta,x,t}-\epsilon_{\theta,t}-\epsilon_{v,t}, \quad \forall \, 1\leq i\leq k_x,
	\end{equation} 
	 then $x_t\in \X$. 
	 
	 Consequently, if \eqref{equ: general state constraint satisfaction} holds  and $u_t\in \U$ for all $t\geq 0$, then $x_t \in \X$ for all $t\geq 0$.
\end{lemma}
\begin{proof}
	%We prove this lemma by induction. 
%	When $t=0$, $x_0\in \X$. 
%	Suppose at $t-1\geq 0$, the claim in Lemma \ref{lem: general constraint satisfaction theta*} is true. 
Consider stage $t\geq 0$.
	By Lemma \ref{lem: xt formula DAP uncertain}, for any $1\leq i \leq k_x$, we have
	\begin{align*}
		D_{x,i}^\top x_t&=	D_{x,i}^\top A_*^{H_t}x_{t-H_t}\\
		&+\sum_{k=2}^{2H_t}\sum_{i=1}^{H_t}	D_{x,i}^\top A_*^{i-1}B_* M_{t-i}[k-i]\hat w_{t-k}\one_{(1\leq k-i\leq H_{t-i})} + \sum_{i=1}^{H_t}	D_{x,i}^\top A_*^{i-1}w_{t-i}+\sum_{i=1}^{H_t}	D_{x,i}^\top A_*^{i-1}B_* \eta_{t-i}\\
		&=	D_{x,i}^\top A_*^{H_t}x_{t-H_t}\\
		&+\sum_{k=1}^{2H_t}D_{x,i}^\top (A_*^{i-1}\one_{k\leq H_t}+\sum_{i=1}^{H_t}A_*^{i-1}B_* M_{t-i}[k-i]\one_{(1\leq k-i\leq H_{t-i})})\hat w_{t-k} + \sum_{i=1}^{H_t}	D_{x,i}^\top A_*^{i-1}(w_{t-i}-\hat w_{t-i})\\
		&+\sum_{i=1}^{H_t}	D_{x,i}^\top A_*^{i-1}B_* \eta_{t-i}\\
		&=	D_{x,i}^\top A_*^{H_t}x_{t-H_t}+\sum_{k=1}^{2H_t}D_{x,i}^\top \Phi_k^x(\mb M_{t-H_t:t-1};\theta_*)\hat w_{t-k} + \sum_{i=1}^{H_t}	D_{x,i}^\top A_*^{i-1}(w_{t-i}-\hat w_{t-i})\\
		&+\sum_{i=1}^{H_t}	D_{x,i}^\top A_*^{i-1}B_* \eta_{t-i}\\
		& \leq \|D_x\|_\infty \kappa(1-\gamma)^{H_t}x_{\max}+ g_i^x(\mb M_{t-H_t:t-1};\theta_*)+ \|D_x\|_\infty \kappa/\gamma \max_{1\leq k\leq H_t}\|\hat \theta_{t-k}-\theta_*\|_2z_{\max}\\
		&+\|D_x\|_\infty \kappa \kappa_B/\gamma \sqrt m \max_{1\leq k\leq H_t}\bar \eta_{t-k}\\
		& \leq \epsilon_{H,t}+ \mathring g_i^x(\mb M_{t};\theta_*)+ \epsilon_{v,t} +\|D_x\|_\infty \kappa/\gamma \max_{1\leq k\leq H_t}\|\hat \theta_{t-k}-\theta_*\|_2z_{\max}+\epsilon_{\eta,x,t}\\
		& \leq \epsilon_{H,t}+ \mathring g_i^x(\mb M_{t};\hat \theta_t)+\epsilon_{\hat \theta}(\|\theta_*-\hat \theta_t\|_F)+ \epsilon_{v,t} +\|D_x\|_\infty \kappa/\gamma \max_{1\leq k\leq H_t}\|\hat \theta_{t-k}-\theta_*\|_2z_{\max}+\epsilon_{\eta,x,t}\\
		& \leq \epsilon_{H,t}+ \mathring g_i^x(\mb M_{t};\hat \theta_t)+ \epsilon_{\theta,t}+  \epsilon_{v,t}+\epsilon_{\eta,x,t}\\
	&	\leq d_{x,i}
	\end{align*}
where we used Lemma \ref{lem: what - w}, Lemma \ref{lem: Definition of epsilon theta(r)},  $x_s\in \X, u_s\in\U$ for all $s\leq t-1$, and \eqref{equ: general state constraint satisfaction}. The last inequality guarantees $x_t\in \X$. Therefore, the proof can be completed by induction.
\end{proof}

\subsection{Proof of Theorem \ref{thm: constraint satisfaction}}

Define an event 
\begin{align}\label{equ: def of Esafe}
	\mathcal E_{\text{safe}}=\{\theta_*\in \bigcap_{e=0}^{N-1} \Theta^{(e)}\}.
\end{align}
Notice that
\begin{align*}
	\Pb(\Esafe)&=1-\Pb(\Esafe^c)\geq 1-\sum_{e=0}^N \Pb(\theta_*\not \in \Theta^{(e)})\geq 1-\sum_{e=1}^N p/(2e^2)\geq 1-p
\end{align*}
where we used Corollary \ref{cor: estimation error of our algo} and $\theta_*\in  \Theta^{(0)}=\Theta_{\text{ini}}$.
In the following, we will condition on the event $\Esafe$ and show $x_t\in \X$ for all $t\geq 0$ under this event. By Lemma \ref{lem: general constraint satisfaction theta*}, we only need to show \eqref{equ: general state constraint satisfaction} for any $t$.
%$$\mathring g_i^x(\mb M_t;\hat \theta_t)\leq d_{x,i}-\epsilon_{ \theta,t}-\epsilon_{H,t}-\epsilon_{\eta,x,t}-\epsilon_{v,t}, \quad \forall \, 1\leq i\leq k_x.$$ 

We discuss three possible cases based on the value of $t$. We introduce some notations for our case-by-case discussion: let $W_{1}^{(e)}, W_{2}^{(e)}$ denote the $W_1, W_2$ defined in Algorithm \ref{alg: safe transit} during the transition in Phase 1, and let $\tilde W_{1}^{(e)}, \tilde W_{2}^{(e)}$ denote the $ W_1,  W_2$ defined in Algorithm \ref{alg: safe transit} during the transition in Phase 2.

%We prove this by induction. When $t=0$, notice that $x_0=0\in \X$, which is a consequence of Assumption \ref{ass: KF epsilonF} when considering $t=0$ and $x_0=0$. Further, for $s<t=0$, $x_s=0\in \X$ by our definition. Next, we suppose at stage $t\geq 1$, we have $x_s\in \X$ for all $s<t$. We will show $x_t\in \X$ below. We discuss three possible cases based on the value of $t$. We introduce some notations for our case-by-case discussion: let $W_{1}^{(e)}, W_{2}^{(e)}$ denote the $W_1, W_2$ defined in Algorithm \ref{alg: safe transit} during the transition in Phase 1, and let $\tilde W_{1}^{(e)}, \tilde W_{2}^{(e)}$ denote the $ W_1,  W_2$ defined in Algorithm \ref{alg: safe transit} during the transition in Phase 2.

%By Lemma \ref{lem: general constraint satisfaction thetaghat}, we only need to show that for any $t$, there exists $\hat \theta_{g,t}\in \Theta_{ini}$, such that 
%$$\mathring g_i^x(\mb M_t;\hat \theta_{g,t}, H_t)\leq d_{x,i}-\epsilon_{\hat \theta,t}-\epsilon_{H,t}-\epsilon_{\eta,x,t}-\epsilon_{\hat w,t}-\epsilon_{v,t}, \quad \forall \, 1\leq i\leq k_x.$$

\nbf{(Case 1: when $T^{(e)}\leq t \leq T^{(e)}+W_{1}^{(e)}-1$.)} In this case, $\mb M_t\in \Omega^{(e-1)}$, so
\begin{align*}
	 g_i^x(\mb M_t;\hat \theta^{(e)})\leq d_{x,i}-\epsilon_H(H^{(e-1)})-\epsilon_{v}(\Delta_M^{(e-1)})-\epsilon_{ \theta}(r^{(e)})
\end{align*}
Notice that $H_t=H^{(e-1)}$, so $\epsilon_{H,t}=\epsilon_H(H^{(e-1)})$. Further, by our algorithm design, $\Delta_{M,t}\leq \Delta_M^{(e-1)}$. Since $\tilde W_{1}^{(e-1)}\geq H^{(e-1)}$ and $\eta_k=0$ for $t_1^{(e-1)}+T_D^{(e-1)}\leq k \leq t$, we have $\epsilon_{\eta,x,t}=\max_{1\leq k\leq H_t}c_2\sqrt m \bar \eta_{t-k}=0$. Next, since $r_\theta^{(e)}\leq r_\theta^{(e-1)}$ by Condition 2 of Theorem \ref{thm: feasibility} for $e\geq 1$, $\epsilon_{\theta,t}\leq \epsilon_\theta(r^{(e)})$. So we satisfy \eqref{equ: general state constraint satisfaction}.

\nbf{(Case 2: when $ T^{(e)}+W_{1}^{(e)} \leq t \leq  t_1^{(e)}+T_{D}^{(e)}+\tilde W_{1}^{(e)}-1$.)} We have $\mb M_t\in \Omega_\dagger^{(e)}$. So
\begin{align*}
	 g_i^x(\mb M_t;\hat \theta^{(e)})\leq d_{x,i}-\epsilon_H(H^{(e)})-\epsilon_{v}(\Delta_M^{(e)})-\epsilon_{\eta,x}(\bar \eta^{(e)})-\epsilon_{ \theta}(r^{(e)})
\end{align*}
Next, $H_t=H^{(e)}$, since $W_{1}^{(e)}\geq H^{(e)}$, we have$ \epsilon_{\theta,t}\leq \epsilon_{ w}(r^{(e)})$, and $\epsilon_{v,t}=\epsilon_{v}(\Delta_M^{(e)})$. Since we take minimum over potential $\bar \eta$ in Step 1 of Algorithm \ref{alg: safe transit} and $W_{1}^{(e)}\geq H^{(e)}$, we have $\epsilon_{\eta,x,t}\leq \epsilon_{\eta,x}(\bar \eta^{(e)})$. So we satisfy \eqref{equ: general state constraint satisfaction}.

\nbf{(Case 3: when $t_1^{(e)}+T_{D}^{(e)}+\tilde W_{1}^{(e)}\leq t \leq  T^{(e+1)}-1$.)} We have $\mb M_t\in \Omega^{(e)}$. So
\begin{align*}
	 g_i^x(\mb M_t;\hat \theta^{(e+1)})\leq d_{x,i}-\epsilon_H(H^{(e)})-\epsilon_{v}(\Delta_M^{(e)})-\epsilon_{ \theta}(r^{(e+1)})
\end{align*}
Next, $H_t=H^{(e)}$, since $W_{1}^{(e)}\geq H^{(e)}$, we have $ \epsilon_{\theta,t}\leq  \epsilon_{\theta}(r^{(e+1})$ by $r^{(e+1)}\leq r_\theta^{(e)}$, and $\epsilon_{v,t}\leq \epsilon_{v}(\Delta_M^{(e)})$. Since we take minimum over potential $\bar \eta$ in Step 1 of Algorithm \ref{alg: safe transit} and $\tilde W_{1}^{(e)}\geq H^{(e)}$, we have $\epsilon_{\eta,x,t}=0$. So we satisfy \eqref{equ: general state constraint satisfaction}.

In conclusion, we satisfy \eqref{equ: general state constraint satisfaction} for all $t\geq 0$. By Lemma \ref{lem: general constraint satisfaction theta*} and Lemma \ref{lem: ut satisfies constraints}, we can show state constraint satisfaction under $\Esafe$.

\section{Regret Analysis}\label{append: regret}
%\red{explain the typos in the theorem statement!??}

%\red{you need to carefully verify the conditions of Corollary 1, this is the most tricky!??}

In this section, we provide a proof for Theorem \ref{thm: regret bound}. Specifically, we first prove the regret bound and then  verify the conditions for  feasibility and constraint satisfaction. %Then, we prove the regret bound. Finally, we introduce the robust MPC in \citet{mayne2005robust} and discusses the regret bound with RMPC as the benchmark.

%\nbf{Typos in Theorem 4:} The conditions for $T^{(1)}$ and $\Delta_M^{(e)}$ should be $T^{(1)}\geq O(((\sqrt{nm} +n)\red{\sqrt m})^{3})$ and $\Delta_M^{(e)}=O(\frac{\epsilon_F^x}{\sqrt{mn	H^{(\red{e})} } (T^{(e+1)})^{1/3}})$, and the regret bound should be $\textup{Regret}\leq \tilde O((n^2m^{\red{2}}+n^{2.5}m^{\red{1.5}})\sqrt{mn+k_x+k_u}T^{2/3}).$

\subsection{Proof of the Regret Bound}
Our proof of the regret bound relies on decomposing  the regret into several parts and bounding each part. 

Firstly, we  decompose the $T$ stages into two parts and decompose the regret accordingly. 
For $e\geq 0$, define 
\begin{align*}
	\T_1^{(e)}&=\{T^{(e)}\leq t \leq t_2(e)+H^{(e)}-1\}, \quad	\T_2^{(e)}=\{t_2(e)+H^{(e)}\leq t \leq T^{(e+1)}-1\}.
\end{align*}
Then, decompose the regret by the stage decomposition below:
\begin{equation}\label{equ: regret decompo}
	\text{Regret}=\sum_{t=0}^{T-1}(l(x_t,u_t)-J^*)=\underbrace{\sum_{e=0}^{N-1} \sum_{t\in \T_1^{(e)}}(l(x_t,u_t)-J^*)}_{\text{First term}}+\underbrace{\sum_{e=0}^{N-1} \sum_{t\in \T_2^{(e)}}(l(x_t,u_t)-J^*)}_{\text{Second term}}
\end{equation}
The first term can be bounded straightforwardly by the fact that the single-stage regret is bounded and the total number of stages in $	\T_1^{(e)}$ for all $e$ can be bounded by $O(T^{2/3} )$. 
\begin{lemma}[Regret Bound of the First Term]\label{lem: regret in sum-e T1(e)}
	When the event $\Esafe$ defined in \eqref{equ: def of Esafe} happens, under the conditions in Theorem \ref{thm: regret bound}, we have
	$$ \sum_{e=0}^{N-1} \sum_{t\in \T_1^{(e)}}(l(x_t,u_t)-J^*)\leq O(T^{2/3})$$
\end{lemma}
\begin{proof}
When $\Esafe $ is true, by Theorem \ref{thm: constraint satisfaction}, we have $x_t\in \X$ and $u_t\in \U$, thus $\|x_t\|_2 \leq x_{\max}$ and $\|u\|_2\leq u_{\max}$ and  $l(x_t,u_t)-J^*\leq \|Q\|_2x_{\max}^2+\|R\|_2u_{\max}^2=O(1)$. 

Next, we bound the number of stages in $\T_1^{(e)}$. Under the conditions in Theorem \ref{thm: regret bound}, the number of the stages in $\T_1^{(e)}$ is $ T_D^{(e)} + H^{(e)}$ plus the safe policy transition stages in Phase 1 and Phase 2. Since $\M_{H^{(e)}}$ is a bounded set, the number of stages in SafeTransit between any two policies in $\M_{H^{(e)}}$ can be bounded by $O(\max(1/\Delta_M^{(e)}, H^{(e)}))=\tilde O(\sqrt{mn} (T^{(e+1)})^{1/3})$, where we used $H^{(e)}=O(\log(T^{(e+1)}))$ and $\Delta_M^{(e)}=O(\frac{\epsilon_F^x}{\sqrt{mn	H^{(e)} } (T^{(e+1)})^{1/3}})$. Further, by $T^{(e+1)}=2T^{(e)}$, $T_D^{(e)}=(T^{(e+1)}-T^{(e)})^{2/3}$, we have $ T_D^{(e)} =O((T^{(e+1)})^{2/3})$. Consequently, the total number of stages in $\T_1^{(e)}$ can be bounded by $O((T^{(e+1)})^{2/3})$ (notice that $T^{(e+1)})^{1/3} \geq \sqrt{mn}$ by our condition of $T^{(1)}$ in Theorem \ref{thm: regret bound}). 

Finally, with the help of the  algebraic fact in  Lemma \ref{lem: bound Te power a by T power a}, we are able to bound the total regret in all episodes by $O((T^{(e+1)})^{2/3})$.
% $O((T^{(e+1)})^{2/3})$. Therefore, by Lemma \ref{lem: bound Te power a by T power a}, we have the proof.
\end{proof}

Lemma \ref{lem: bound Te power a by T power a} is a technical fact that will be used throughout our regret proof.
\begin{lemma}[An algebraic fact]\label{lem: bound Te power a by T power a}
	When $T^{(e)}=2^{e-1}T^{(1)}$, and $T^{(N)}\geq T >T^{(N-1)}$, $N\leq O(\log T)$. Further, for any $\alpha>0$, we have
	$$\sum_{e=1}^{N}(T^{(e)})^\alpha =O(T^\alpha)$$
\end{lemma}
\begin{proof}
	By $T \geq T^{(N-1)}\geq 2^{(N-2)}$, we have $\log T \geq (N-2)\log(2)$, so $N\leq O(\log T)$.
	Further, 
$
		\sum_{e=1}^{N}(T^{(e)})^\alpha = \sum_{e=1}^{N}(2^{e-1})^\alpha (T^{(1)})^\alpha \leq O((2^{N})^\alpha (T^{(1)})^\alpha )\leq O(T^\alpha)
$.
\end{proof}

The second term in \eqref{equ: regret decompo} is more complicated to bound, so we further decompose it into four parts as follows.
\begin{align*}
	\sum_{e=0}^{N-1} \sum_{t\in \T_2^{(e)}}(l(x_t,u_t)-&J^*)=\underbrace{\sum_{e=0}^{N-1} \sum_{t\in \T_2^{(e)}}(l(x_t,u_t)-l(\hat x_t,\hat u_t))}_{\text{Part i}}+\underbrace{\sum_{e=0}^{N-1} \sum_{t\in \T_2^{(e)}}(l(\hat x_t,\hat u_t)-  f(\mb M^{(e)};\theta_*))}_{\text{Part ii}}\\
	&+\underbrace{\sum_{e=0}^{N-1} \sum_{t\in \T_2^{(e)}}(  f(\mb M^{(e)};\theta_*)-  f(\mb M^*_{H^{(e)}};\theta_*))}_{\text{Part iii}}+\underbrace{\sum_{e=0}^{N-1} \sum_{t\in \T_2^{(e)}}(  f(\mb M^*_{H^{(e)}};\theta_*)-J^*)}_{\text{Part iv}},
\end{align*}
where we introduced auxiliary states $\hat x_t$ and actions $\hat u_t$ defined as
\begin{align*}
	\hat x_t &= \sum_{k=1}^{2H^{(e)}} \Phi_k^x(\mb M^{(e)};\theta_*)w_{t-k}, \quad	\hat u_t=\sum_{k=1}^{H^{(e)}}M^{(e)}[k]w_{t-k},
\end{align*}
which are basically the approximate states and the  actions generated by the disturbance-action policy $\mb M^{(e)}$  computed in Phase 2 of Algorithm \ref{alg: online algo} when the actual disturbances $w_{t-k}$ are known. We also introduce an auxiliary policy $ \mb M^*_{H^{(e)}}$ in Part iii, which is defined as the optimal DAP policy  in \eqref{equ: def DAP} with a memory length $H=H^{(e)}$  under a known model, i.e. $ \mb M^*_{H^{(e)}}=\argmin_{\mb M\in \Omega_*^{(e)}}  f(\mb M;\theta_*)$, where $\Omega_{*}^{(e)}$ is defined by \eqref{equ: Omega def} with $H=H^{(e)}$.

The rest of the proof is to bound Parts i-iv. Establishing the bound on Part iii is the major part of the proof and the bound on Part iii is the dominating term in our regret bound, so we will present our bound on Part iii first. Then, we will establish bounds on Parts i, ii, and iv.

\subsubsection{Bound on Part iii}
Notice that $\mb M^{(e)}$ is the solution to the CCE program in \eqref{equ: CCE} and $\mb M^*_{H^{(e)}}$ is the solution to the optimal DAP program in \eqref{equ: def DAP}. Further,  the CCE program \eqref{equ: CCE}  can be viewed as a slightly perturbed version of  the optimal DAP program  \eqref{equ: def DAP} due to model estimation errors and constraint-tightening terms. Therefore, we can bound Part iii by the perturbation analysis. 

Specifically, we establish the following general perturbation bound. This bound is not only useful for our bound on Part iii but also helps the discussions after Theorem \ref{thm: regret bound} on the reasons for including the pure exploitation phases.
\begin{lemma}[Perturbation analysis for CCE]\label{lem: perturb cost}
	Consider a fixed memory length $H\geq \log(2\kappa)/\log((1-\gamma)^{-1})$ and $\theta_1, \theta_2\in \Theta_{\textup{ini}} $. Consider two CCE programs $\mb M_1=\argmin_{\mb M\in \Omega(\theta_1, \epsilon_{x1}, \epsilon_{u1})} f(\mb M; \theta_1)$ and $\mb M_2=\argmin_{\mb M\in \Omega(\theta_2, \epsilon_{x2}, \epsilon_{u2})} f(\mb M; \theta_2)$. Suppose there exists $\epsilon_g>0$ such that $\Omega(\theta_1, \epsilon_{x1}+\epsilon_g, \epsilon_{u1}+\epsilon_g)\cap \Omega(\theta_2, \epsilon_{x2}+\epsilon_g, \epsilon_{u2}+\epsilon_g)$ is non-empty. Then, we have
	\begin{align*}
		f(\mb M_1, \theta_2)-f(\mb M_2, \theta_2)\leq  O(mn r+(\sqrt{mn}+\sqrt{k_x+k_u})n\sqrt{mH} \max(|\epsilon_{x1}-\epsilon_{x2}|, |\epsilon_{u1}-\epsilon_{u2}|))/\epsilon_g
	\end{align*}
where $\|\theta_1-\theta_2\|_F\leq r$.
	
%	Consider model uncertainty set $\Theta=\{\theta: \|\theta-\hat \theta\|_F\leq r\}$ containing the true model $\theta_*$, exploration level $\bar \eta$, and variation budget $\Delta_M$. Consider a large enough $H$.
%	Consider a cautious certainty equivalent control  $\mb M_{cce}=\argmin_{\mb M\in \Omega(\hat \theta, \epsilon, H)} f(\mb M; \hat \theta)$  for $\epsilon=\epsilon_c(H, \bar \eta, r)+\epsilon_v(\Delta_M, H)$. Then,
%	$$f(\mb M_{cce};\theta_*)-J^*\leq \tilde O(r+\Delta_M+\bar \eta)$$
%	where $\tilde O(\cdot)$ hides polynomial factors of problem dimensions for  illustration purposes.
\end{lemma}
\begin{proof}
	Notice that both the objective functions and the constraints are different in the two CCE programs above, so we  introduce an auxiliary policy $\mb M_3=\argmin_{\mb M\in \Omega(\theta_1, \epsilon_{x1}, \epsilon_{u1})} f(\mb M; \theta_2)$ to discuss the perturbation bounds on cost differences and constraint differences separately. 
	We will first bound 	$f(\mb M_1, \theta_2)-f(\mb M_3, \theta_2)$ and then bound 	$f(\mb M_3, \theta_2)-f(\mb M_2, \theta_2)$ below.
	
	\nbf{Perturbation on the cost functions.}
	\begin{align*}
		f(\mb M_1, \theta_2)-f(\mb M_3, \theta_2)& = 	f(\mb M_1, \theta_2)-f(\mb M_1, \theta_1)+	f(\mb M_1, \theta_1)-f(\mb M_3, \theta_1)+	f(\mb M_3, \theta_1)-f(\mb M_3, \theta_2)\\
		& \leq 	f(\mb M_1, \theta_2)-f(\mb M_1, \theta_1)+ 	f(\mb M_3, \theta_1)-f(\mb M_3, \theta_2)\\
		& \leq O({mn}\|\theta_1-\theta_2\|_F)
	\end{align*}
where the first inequality is because $\mb M_1$ and $\mb M_3$ are in the same set and $\mb M_1$ minimizes the cost $f(\mb M, \theta_1)$ in this set, and the second inequality is because of the following perturbation lemma on the cost functions.
\begin{lemma}[Perturbation bound on $  f$ with respect to $\theta$]\label{lem: f(theta)-f(theta')}
	For any $H\geq \log(2\kappa)/\log((1-\gamma)^{-1})$, $\mb M\in \M_H$, any $\theta,  \hat \theta \in \Theta_{\textup{ini}}$ , we have
	$$|  f(\mb M;\theta)-  f(\mb M;\hat \theta)|\leq O(mn r)$$
	where $\|\theta-\hat \theta\|_F\leq r$.
\end{lemma}
	\begin{proof}
	We let $\tilde x_t(\theta)$ and $\tilde x_t( \theta_2)$ denote the approximate states defined by Proposition \ref{prop: state formula with known system}, and we omit $\mb M$ in this proof for notational simplicity. Notice that
		\begin{align*}
		\|	\tilde x(\theta)-\tilde x( \hat \theta)\|_2&=\|\sum_{k=1}^{2H }(  \Phi_k^x( \theta)-   \Phi_k^x(\hat  \theta))w_{t-k}\|_2 \leq \sum_{k=1}^{2H }\|(  \Phi_k^x( \theta)-   \Phi_k^x( \hat \theta))w_{t-k}\|_2\\
		& \leq \sum_{k=1}^{H } \|(A^{k-1}-\hat A^{k-1}) w_{t-k}\|_2+ \sum_{k=1}^{2H }\sum_{i=1}^H\| (A^{i-1}B-\hat A^{i-1}\hat B)M[k-i]\id_{(1\leq k-i\leq H)} w_{t-k}\|_2 \\
				& \leq \sum_{k=1}^{H } \|(A^{k-1}-\hat A^{k-1})\|_2 \sqrt{n} w_{\max} + \sum_{k=1}^{2H }\sum_{i=1}^H\| (A^{i-1}B-\hat A^{i-1}\hat B)\|_2\sqrt{m}\|M[k-i]\id_{(1\leq k-i\leq H)} w_{t-k}\|_\infty \\
				& \leq  \sum_{k=1}^{H } \sqrt{n} w_{\max} O(k(1\!-\!\gamma)^{k\!-\!1} \|\theta-\hat\theta\|_F)\!+  \! \sum_{k=1}^{2H }\sum_{i=1}^H\!O(\sqrt{m}\| (A^{i\!-\!1}B\!-\!\hat A^{i-1}\hat B)\|_2 \|M[k\!-\!i]\|_\infty \id_{(1\leq k-i\leq H)} )\\
				& \leq O(\sqrt n r_\theta)+  \sum_{k=1}^{2H }\sum_{i=1}^H \|\theta-\hat\theta\|_F\sqrt{mn} O((1-\gamma)^{k-i-1}(i-1) (1-\gamma)^{i-2}\id_{(i\geq 2)}\id_{(1\leq k-i\leq H)}) \\
		& = O(\sqrt n \|\theta-\hat\theta\|_F )+ O(\sum_{i=1}^H \sum_{j=1}^H  \|\theta-\hat\theta\|_F\sqrt{mn} (i-1) (1-\gamma)^{i-2}(1-\gamma)^{j-1}  )\\
		& \leq  O(\sqrt n \|\theta-\hat\theta\|_F)+ O(\sqrt{mn} \|\theta-\hat\theta\|_F)	= O(\sqrt{mn} \|\theta-\hat\theta\|_F)	
	\end{align*}
where we used Lemma \ref{lem: perturb Ak AkB} in the third and fourth inequality.

	\end{proof}
		\nbf{Perturbation on the constraints.} Our bound on  	$f(\mb M_3, \theta_2)-f(\mb M_2, \theta_2)$ is established based on Proposition 2 and Lemma 9 in \cite{li2020online}. The major difference is that \cite{li2020online} only considers changes in the right-hand-side of the constraint inequalities but in our setting, the left-hand-side of the constraint inequalities also change. To handle this difference, we introduce two auxiliary sets  $\Omega_1$ and $\Omega_2$ with the same  left-hand-sides in the constraint inequalities  such that $\Omega_1=\Omega(\theta_1, \epsilon_{x1}, \epsilon_{u1})$ and $\Omega_2=\Omega(\theta_2, \epsilon_{x2}, \epsilon_{u2})$, which is achieved by adding inactive inequalitiy constraints. Specifically, define
		 \begin{align*}
		 	&\Omega_1=\Omega(\theta_1, \epsilon_{x1}, \epsilon_{u1})\cap \Omega(\theta_2, \epsilon_{x1}-\epsilon_{\hat \theta}(r), \epsilon_{u1}), \quad \Omega_2=\Omega(\theta_2, \epsilon_{x2}, \epsilon_{u2})\cap \Omega(\theta_1, \epsilon_{x2}-\epsilon_{\hat \theta}(r), \epsilon_{u2}).
		 \end{align*}
	 Notice that the constraints in  $\Omega(\theta_2, \epsilon_{x1}-\epsilon_{\hat \theta}(r), \epsilon_{u1})$ and $\Omega(\theta_1, \epsilon_{x2}-\epsilon_{\hat \theta}(r), \epsilon_{u2})$ are inactive due to Lemma \ref{lem: Definition of epsilon theta(r)}. Further, notice that $\Omega_1, \Omega_2$ are both polytopes.
	 
	 Another difference between our setting and that in Proposition 2 of \cite{li2020online} is that $\Omega_1$ may not be a subset of $\Omega_2$ and vice versa, so we need to generalize Proposition 2 to this setting as follows.
	 
	 \begin{lemma}[Extension from Proposition 2  \cite{li2020online}]\label{lem: cost diff lemma for linear constrained opt}
	 	Consider two polytopes: $\Omega_1=\{x: Cx\leq h-\Delta_1\}, \Omega_2=\{x: Cx\leq h-\Delta_2\}$, where $\Delta_1, \Delta_2$ are two vectors. Define $\Delta_0=\min(\Delta_1, \Delta_2)$ elementwise. Define $\Delta_3=\max(\Delta_1, \Delta_2)$ elementwise. Suppose the $l_2$-diameter of $\Omega_0$ is $d_{\Omega_0}$, $f(x)$ is $L$-Lipschitz continuous, and there exists $x_F\in \Omega_3$, then we have
	 	$$|\min_{\Omega_1}f(x)-\min_{\Omega_2}f(x)| \leq \frac{2Ld_{\Omega_0} \|\Delta_1-\Delta_2\|_\infty}{\min_{\{i: (\Delta_1)_i\not=(\Delta_2)_i\}} (h-\Delta_3-Cx_F)_i}$$
	 \end{lemma}
	 The proof is deferred to Appendix \ref{subsec: perturb QP}.
	 
	 Now, we are ready to bound $f(\mb M_3, \theta_2)-f(\mb M_2, \theta_2)$. By our definitions and discussions above, we have $f(\mb M_3, \theta_2)=\min_{\mb M\in \Omega_1}f(\mb M;\theta_2)$ and $f(\mb M_2, \theta_2)=\min_{\mb M\in \Omega_2}f(\mb M;\theta_2)$. Further, notice that $\Omega_1, \Omega_2$ are both polytopes and can be written in the form of 
	$\{\vec P:  C\vec P\leq h-\Delta_1\}, \{\vec P:  C\vec P\leq h-\Delta_2\}$ by introducing auxiliary variables to represent the absolute values (see Lemma 10 in \cite{li2020online} for more details). Therefore, we can apply Lemma \ref{lem: cost diff lemma for linear constrained opt} to obtain the bound on $f(\mb M_3, \theta_2)-f(\mb M_2, \theta_2)$ by bounding the corresponding constants $d_{\Omega_0}, L, \|\Delta_1-\Delta_2\|_\infty$, and $\min_{\{i: (\Delta_1)_i\not=(\Delta_2)_i\}} (h-\Delta_3-Cx_F)_i$. Same as the proof of Lemma 9, we can show the $l_2$-diameter of $\Omega_0$ is $d_{\Omega_0}=O(\sqrt{mn}+\sqrt{k_x+k_u})$ and $\|\Delta_1-\Delta_2\|_\infty=\max(|\epsilon_{x1}-\epsilon_{x2}|, |\epsilon_{u1}-\epsilon_{u2}|)$. Further, the Lipschitz factor $L$ can be obtained from the gradient bound $L= G_f=O( \sqrt{n^2mH})$ as provided below, whose proof is provided in Appendix \ref{subsec: Gf bound}.
	\begin{lemma}[Gradient bound of $f(\mb M;\theta)$]\label{lem: bound on Gf}
		For any $H\geq 1$, $\mb M\in \M_H$, $\theta\in \Theta_{ini}$,  we have
		$\|\nabla f(\mb M;\theta)\|_F\leq G_f=O(\sqrt{n^2mH})$.
	\end{lemma}
Next, since $\Omega(\theta_1, \epsilon_{x1}+\epsilon_g, \epsilon_{u1}+\epsilon_g)\cap \Omega(\theta_2, \epsilon_{x2}+\epsilon_g, \epsilon_{u2}+\epsilon_g)$ is non-empty, there exists  $\vec P_F\in \Omega_3$ such that  $\min_{\{i: (\Delta_1)_i\not=(\Delta_2)_i\}} (h-\Delta_3-Cx_F)_i\geq \epsilon_0$. Lastly, by applying the constants above to Lemma \ref{lem: cost diff lemma for linear constrained opt}, we can show $f(\mb M_3, \theta_2)-f(\mb M_2, \theta_2)\leq O((\sqrt{mn}+\sqrt{k_x+k_u})\sqrt{n^2mH})\max(|\epsilon_{x1}-\epsilon_{x2}|, |\epsilon_{u1}-\epsilon_{u2}|)/\epsilon_g$.

The proof of Lemma \ref{lem: perturb cost} is completed by combining the bounds on $	f(\mb M_1, \theta_2)-f(\mb M_3, \theta_2)$ and $f(\mb M_3, \theta_2)-f(\mb M_2, \theta_2)$.

\end{proof}

Based on Lemma \ref{lem: perturb cost}, we can show the following bound on Part iii when $\Esafe$ is true.
\begin{equation}\label{equ: bound on Part iii}
\text{Part iii}=	\sum_{e=0}^{N-1} \sum_{t\in \T_2^{(e)}}(  f(\mb M^{(e)};\theta_*)-  f(\mb M^*_{H^{(e)}};\theta_*))\leq  O((n^2m^{2}+n^{2.5}m^{1.5})\sqrt{mn+k_x+k_u}T^{2/3})
\end{equation}
The proof for \eqref{equ: bound on Part iii} is provided below. For each $e\geq 0$, notice that $\mb M^{(e)}=\argmin_{\mb M\in \Omega(\hat \theta^{(e+1)}, \hat \epsilon_x^{(e)},  0)} f(\mb M; \theta^{(e+1)})$, where we define $ \hat \epsilon_x^{(e)}=\epsilon_\theta(r^{(e+1)})+\epsilon_H(H^{(e)})+\epsilon_v(\Delta_M^{(e)}, H^{(e)})$; and $\mb M^*_{H^{(e)}}=\argmin_{\mb M\in \Omega(\theta_*, \epsilon_H(H^{(e)}), 0)}f(\mb M;\theta_*)$. Therefore, we can apply Lemma \ref{lem: perturb cost}. We first verify the conditions of Lemma \ref{lem: perturb cost}. Notice that $\theta^{(e+1)}, \theta_*\in \Theta_{\text{ini}}$, $H^{(e)}$ is large enough. Further, we have $\epsilon_g =min(\epsilon_{F,x}, \epsilon_{F,u})/4>0$ such that $\Omega(\hat \theta^{(e+1)}, \hat \epsilon_x^{(e)}+\epsilon_g,  \epsilon_g)\cap \Omega(\theta_*, \epsilon_H(H^{(e)})+\epsilon_g, \epsilon_g)$ is not empty due to the  feasibility conditions in Theorem \ref{thm: feasibility}. Therefore, when $\Esafe$ is true, by Lemma \ref{lem: perturb cost}, under our choices of parameters in Theorem \ref{thm: regret bound}, we have
\begin{align*}
	 f(\mb M^{(e)};\theta_*)&-  f(\mb M^*_{H^{(e)}};\theta_*) \leq O(mn r^{(e+1)}+(\sqrt{mn}+\sqrt{k_x+k_u})n\sqrt{mH^{(e)}} (\epsilon_\theta(r^{(e+1)})+\epsilon_v(\Delta_M^{(e)}, H^{(e)})))\\
	 & \leq \tilde O(mn  ({n\sqrt m +m \sqrt n}) (T^{(e)})^{-1/3}+ (\sqrt{mn}+\sqrt{k_x+k_u})n\sqrt{m}\sqrt{mn} ({n\sqrt m +m \sqrt n}) (T^{(e+1)})^{-1/3} )\\
	 & \leq \tilde O\left( (n^2m^{2}+n^{2.5}m^{1.5})\sqrt{mn+k_c}(T^{(e+1)})^{-1/3}\right)
\end{align*}
Consequently, by Lemma \ref{lem: bound Te power a by T power a}, we prove the bound \eqref{equ: bound on Part iii}.

%As discussed in Appendix \ref{subsec: condition verification for regret bound}, our algorithm parameters ensure the feasibility conditions in Theorem \ref{thm: feasibility}, so there exists an $\epsilon_g =O(\min(\epsilon_{F,x}, \epsilon_{F,u}))>0$ such that $\Omega(\hat \theta^{(e+1)}, \hat \epsilon_x^{(e)}+\epsilon_g,  \epsilon_g)\cap \Omega(\theta_*, \epsilon_H(H^{(e)})+\epsilon_g, \epsilon_g)$ is not empty. 

\subsubsection{Bound on Part i}
When $\Esafe$ is true, we are able to show
\begin{equation}\label{equ: bound on Part i}
\text{Part i}=	\sum_e \sum_{t\in \T_2^{(e)}} l(x_t, u_t)-l(\hat x_t, \hat u_t)\leq \tilde O(n \sqrt m \sqrt{m+n}  T^{2/3}).
\end{equation}
The proof is provided below. 

Firstly, under $\Esafe$, we have $x_t \in \X$ and $u_t \in \U$, so $\|x_t\|_2\leq O(1), \|u_t\|_2\leq O(1)$. Further, by the definitions of $\hat x_t, \hat u_t$ and the proof of Theorem \ref{thm: constraint satisfaction}, we can also verify that $\hat x_t\in \X$ and $\hat u_t \in \U$ based on the proof of Lemma \ref{lem: ut satisfies constraints} and Lemma \ref{lem: general constraint satisfaction theta*}. Therefore, we have$\|\hat x_t\|_2\leq O(1), \|\hat u_t\|_2\leq O(1)$

Next, by Lemma \ref{lem: xt formula DAP uncertain}, for $t\in \T_2^{(e)}$, we can write $x_t, \hat x_t, u_t, \hat u_t$ as
\begin{align*}
	x_t & = \!A_*^{H_t}x_{t\!-\!H_t}\!+\!\sum_{k=2}^{2H_t}\sum_{i=1}^{H_t}A_*^{i\!-\!1}B_* M_{t\!-\!i}[k\!-\!i]\hat w_{t-k}\one_{(1\leq k-i\leq H_{t\!-\!i})} \!+\! \sum_{i=1}^{H_t}A_*^{i-1}w_{t\!-\!i}\\
	\hat x_t &= \!A_*^{H_t}x_{t\!-\!H_t}\!+\!\sum_{k=2}^{2H_t}\sum_{i=1}^{H_t}A_*^{i\!-\!1}B_* M_{t\!-\!i}[k\!-\!i] w_{t-k}\one_{(1\leq k-i\leq H_{t\!-\!i})} \!+\! \sum_{i=1}^{H_t}A_*^{i-1}w_{t\!-\!i}\\
	u_t &=\sum_{t=1}^{H_t} M_t[k]\hat w_{t-k},\quad \hat u_t =\sum_{t=1}^{H_t} M_t[k] w_{t-k}.
\end{align*}
Hence,  we can bound $\|x_t-\hat x_t\|_2$ and $\|u_t-\hat u_t\|_2$ by Lemma \ref{lem: what - w} below:
\begin{align*}
	\|x_t-\hat x_t\|_2& \leq O((1-\gamma)^{H^{(e)}}+  \sqrt{mn}  r_\theta^{(e+1)})=\tilde O(n m\sqrt{m+n} (T^{(e+1)})^{-1/3})\\
	\|u_t-\hat u_t\|_2&\leq O(\sqrt{mn} r_\theta^{(e+1)})=\tilde O(n  m \sqrt{m+n} (T^{(e+1)})^{-1/3})
\end{align*}
Consequently, by applying Lemma \ref{lem: bound Te power a by T power a} and the quadratic structure of $l(x,u)$, we can bound the Part i by
$$\sum_e \sum_{t\in \T_2^{(e)}} (l(x_t, u_t)-l(\hat x_t, \hat u_t))\leq \sum_e T^{(e+1)} \tilde O(n  m \sqrt{m+n} (T^{(e+1)})^{-1/3})\leq \tilde O(n  m \sqrt{m+n}  T^{2/3}).$$

\subsubsection{Bound on Part ii}

\begin{lemma}[Bound on Part ii]\label{lem: bound on Part ii}
	With probability $1-p$, $\text{Part ii}\leq \tilde O(mn \sqrt T )$.
\end{lemma}
Notice that this part is not a dominating term in the regret bound. The proof relies on a martingale concentration analysis and is very technical, so we defer it to Appendix \ref{subsec: proof for bound on Part ii}.

\subsubsection{Bound on Part iv}
In the following, we will show that 
\begin{equation}\label{equ: bound on Part iv}
	\text{Part iv}=\sum_{e=0}^{N-1} \sum_{t\in \T_2^{(e)}}(  f(\mb M^*_{H^{(e)}};\theta_*)-J^*)=\tilde O(n\sqrt m\sqrt{mn+k_c} \sqrt n)
\end{equation}
The proof is provided below.

Remember that $J^*$ is generated by the optimal safe linear policy $K^*$. By Lemma 4 and Corollary 2 in \cite{li2020online}, for a memory length $H^{(e)}$, we can define $\mb M_{H^{(e)}}(K^*)\in \Omega(\theta_*, -\epsilon_P^{(e)},0)$, where $\epsilon_P^{(e)}=\sqrt{mn}(1-\gamma)^{H^{(e)}}$ corresponds to $\epsilon_1+\epsilon_3$ in \cite{li2020online} with $H=H^{(e)}$.\footnote{As discussed in footnote \ref{footnote: epsilonP}, our $\epsilon_P^{(e)}$ has smaller dependence on $n,m,H$ compared with \cite{li2020online}.} Further, by Lemma 6 in  \cite{li2020online}, we have
$$f(\mb M_{H^{(e)}}(K^*);\theta_*)-J^*=\lim_{T\to +\infty}\frac{1}{T}\sum_{t=0}^{T-1}f(\mb M_{H^{(e)}}(K^*);\theta_*)-\E(l(x_t^*, u_t^*))\leq O(n^2m (H^{(e)})^2 (1-\gamma)^{H^{(e)}})$$

In addition, we have
\begin{align*}
	 f(\mb M^*_{H^{(e)}};\theta_*)-f(\mb M_{H^{(e)}}(K^*);\theta_*)& \leq f(\mb M^*_{H^{(e)}};\theta_*) -\min_{\mb M\in \Omega(\theta_*,-\epsilon_P^{(e)},0)}f(\mb M; \theta_*)\\
	 & \leq \tilde O(n\sqrt m \sqrt{mn +k_x+k_u} (\epsilon_P^{(e)}+\epsilon_H(H^{(e)})))\\
	 & \leq \tilde O(n\sqrt m\sqrt{mn+k_c} \sqrt{mn} (1-\gamma)^{H^{(e)}})
\end{align*}
By combining the bounds above and by choosing $H^{(e)}\geq \log(T^{(e+1)})/\log((1-\gamma)^{-1})$, we have 
\begin{align*}
	f(\mb M^*_{H^{(e)}};\theta_*)-J^*&=f(\mb M^*_{H^{(e)}};\theta_*)-f(\mb M_{H^{(e)}}(K^*);\theta_*)+(\mb M_{H^{(e)}}(K^*);\theta_*)-J^*\\
	&\leq \tilde O(n\sqrt m\sqrt{mn+k_c} \sqrt{mn} /T^{(e+1)}),
\end{align*}
which directly leads to the bound \eqref{equ: bound on Part iv} on Part iv.

\subsubsection*{Completing the proof of the regret bound.}

By  combining \eqref{equ: bound on Part i}, \eqref{equ: bound on Part iii}, \eqref{equ: bound on Part iv}, and Lemma \ref{lem: bound on Part ii}, we obtain our regret bound in Theorem \ref{thm: regret bound}. 
Notice that \eqref{equ: bound on Part i}, \eqref{equ: bound on Part iii}, \eqref{equ: bound on Part iv} all condition on  $\Esafe$, and $\Esafe$ holds w.p. $1-p$. But Lemma \ref{lem: bound on Part ii} conditions on a different event and that event also holds with probability $1-p$. Putting them together, we have that our regret bound holds w.p. $1-2p$.

\subsection{Condition Verification for Feasibility and Constraint Satisfaction}\label{subsec: condition verification for regret bound}
In this subsection, we briefly show that there exist  parameters characterized by Theorem \ref{thm: regret bound} that satisfy the conditions for feasibility and constraint satisfaction in Theorem \ref{thm: feasibility} and 
Theorem \ref{thm: constraint satisfaction}, which include: $T_D^{(e)}$ satisfying the condition in Theorem \ref{thm: general estimation error bdd} (Corollary \ref{cor: estimation error of our algo}'s condition),  condition \eqref{equ: initial feasibility}, condition (ii) of Theorem \ref{thm: feasibility},  and $T^{(e+1)}\geq t_2^{(e)}$. 

Firstly, in Corollary \ref{cor: estimation error of our algo}, we need $T_D^{(e)}\geq O(\log(2e^2/p)+(m+n)\log(1/\bar \eta))$ for $e\geq 0$. By our choices, we have $T_D^{(e)}=(T^{(\min(e,1)})^{2/3}$ and $T^{(e+1)}=2T^{(e)}$, so $T_D^{(e)}$ increases exponentially. Therefore, $T_D^{(e)}\geq O(\log(2e^2/p)+(m+n)\log(1/\bar \eta))$ can be guaranteed if $T_D^{(1)}\geq O(\log(1/p)+(m+n)\log(1/\bar \eta))$ with some sufficiently large constant factor, which requires $T^{(1)}\geq O((m+n)^{3/2})$. 

Secondly, for condition \eqref{equ: initial feasibility}, we set $\epsilon_0=\epsilon_{F,x}/4$ and let $\epsilon_P+\epsilon_H(H^{(0)})\leq \epsilon_{F,x}/12$, $\epsilon_{\eta,x}\leq \epsilon_{F,x}/12$,  $\epsilon_v(\Delta_M^{(0)}, H^{(0)})\leq  \epsilon_{F,x}/12$, and $\epsilon_P \leq \epsilon_{F,u}/4$, $\epsilon_{\eta,u}\leq \epsilon_{F,u}/4$. These conditions require $H^{(0)}\geq O(\log(\sqrt{mn}/\min(\epsilon_F)))$, $\bar\eta^{(e)}= O(\min( \frac{\epsilon^x_F}{ \sqrt m }, \epsilon_F^u))$, and $\Delta_M^{(e)}=O(\frac{\epsilon_F^x}{\sqrt{mn	H^{({e})} } (T^{(e+1)})^{1/3}})$. 

Thirdly, for the condition (ii) of Theorem \ref{thm: constraint satisfaction}, the monotonicity for $H^{(e)}, \sqrt{H^{(e)}}\Delta_M^{(e)}$ are satisfied and $\bar \eta^{(e)}$ is a constant, so its monotonicity condition is also satisfied. With exponentially increasing $T_D^{(e)}$, the decreasing $r^{(e)}$ is also satisfied. We only need to verify that $r^{(1)}\leq r_{\text{ini}}$. This requires $T^{(1)}\geq \tilde O((\sqrt{n}m +n\sqrt n)^{3})$.

Lastly, for $T^{(e+1)}\geq t_2^{(e)}$, notice that Phase 1 only takes $(T^{(e+1)}-T^{(e)})^{2/3}$ stages, and the safe transitions only takes $\tilde O((T^{(e+1)})^{1/3}) $ stages, so $T^{(e+1)}\geq t_2^{(e)}$ for all $e$ for large enough initial $T^{(1)}$.

\section{More Discussions}\label{append: discuss}
In this appendix, we briefly introduce RMPC in \citet{mayne2005robust} and show that its infinite-horizon averaged cost can be captured by  $J(\mathbb K)$ for some safe linear policy $\mathbb K$. Therefore, algorithms with small regret compared with optimal safe linear policies can also achieve comparable performance with RMPC in \citet{mayne2005robust} for long horizons, which further motivates our choice of regret benchmarks as safe linear policies.
Further, we  discuss the implementation of our algorithm for non-zero $x_0$.
\subsection{A brief review of RMPC in \cite{mayne2005robust}}

RMPC is a popular method to handle constrained system with disturbances and/or other system uncertainties. Since we will include RMPC in the benchmark policy class, we assume the model $\theta_*$ is available here, but RMPC can also handle model uncertainties. Many different versions of RMPC have been proposed in the literature, (see \cite{rawlings2009model} for a review). In this appendix, we will focus on a tube-based RMPC defined in \cite{mayne2005robust}. The RMPC method  in \cite{mayne2005robust} enjoys desirable theoretical guarantees, such as robust exponential stability, recursive feasibility, constraint satisfaction, and  is thus commonly adopted. RMPC usually considers $x_0\not =0$. When considering RMPC for regulation problems, one goal of RMPC is to quickly and safely steer the states to a neighborhood of origin (due to the system disturbances, one cannot steer the state to the origin exactly).

Next, we briefly introduce the tube-based RMPC scheme. In most tube-based RMPC schemes (not just \cite{mayne2005robust}), it is required to know a  linear static controller $u_t=-\Kb x_t$ such that this controller is strictly safe if the system starts from the origin. A  disturbance-invariant set for the closed-loop system $x_{t+1}=Ax_t-B\Kb x_t+w_t$ is also needed. 
\begin{definition}
	$\Xi$ is called a disturbance-invariant set for $x_{t+1}=Ax_t-B\Kb x_t+w_t$ if for any $x_t \in \Xi$, and $w_t\in \W$, we have $x_{t+1}\in \Xi$. 
\end{definition}
For computational purposes, a polytopic approximation of disturbance-invariant set is usually employed. Further, the implementation of RMPC also requires the knowledge of a terminal set $X_f$ such that for any $x_0\in X_f$, implementing the controller $u_t=-\Kb x_t$ is safe, as well as a terminal cost function $V_f(x)=x^\top P x$ satisfying certain conditions (see \cite{mayne2005robust} for more details).

\nbf{RMPC scheme in \cite{mayne2005robust}.}
Now, we are ready to define the tube-based RMPC proposed in \cite{mayne2005robust}. 
At each stage $t$, consider a planning window $t+k|t$ for $0\leq k \leq W$, RMPC in \cite{mayne2005robust} solves the following optimization:
\begin{equation}\tag{RMPC \cite{mayne2005robust}}
	\begin{aligned}
		\min_{ x_{t|t},  u_{t+k|t}} \ &\sum_{k=0}^{W-1} l( x_{t+k|t},  u_{t+k|t})+ V_f( x_{t+W|t})\\
		\text{s.t. \ } &  x_{t+k+1|t}=A_* \bar x_{t+k|t}+B_* u_{t+k|t}, \quad k\geq 0\\
		&  x_{t|t}\in x_{t}\oplus\Xi\\
		&  x_{t+k|t}\in 	\mathbb X\ominus\Xi, \forall 0\leq k \leq W-1\\
		&  u_{t+k|t}\in 	\mathbb U\ominus\mathbb K\Xi, \forall 0\leq k \leq W-1\\
		&  x_{t+W|t}\in X_f \subseteq 	\mathbb X\ominus\Xi
	\end{aligned}
\end{equation}
Then, implement control:
$$ u_t=-\mathbb K (x_t- x_{t|t}^*)+ u_{t|t}^*.$$
Notice that $ x_{t|t}^*, u_{t|t}^*$ are functions of $x_t$. Further, by \cite{bemporad2002explicit}, $u_t$ is a piece-wise affine (PWA) function of the state $x_t$ when $\Xi$ is a polytope. 
Define the set of feasible initial values as
$$X_N=\{x_0: \text{(RMPC \cite{mayne2005robust}) is feasible when $x_t=x_0$}\}.$$

The RMPC scheme in \cite{mayne2005robust} is a variant of the traditional RMPC schemes by allowing more freedom when choosing $x_{t|t}$, i.e., in the scheme above, $x_{t|t}$ is also an optimization variable as long as  $x_{t|t}\in x_{t}\oplus\Xi$, but in traditional RMPC schemes, $x_{t|t}=x_t$ is fixed. With this adjustment, the RMPC scheme in \cite{mayne2005robust} enjoys robust exponential stability.
\begin{theorem}[Theorem 1 in \cite{mayne2005robust}]\label{thm: rob exp stable}
	The set $\Xi$ is robustly exponentially stable for the closed-loop system with \textup{(RMPC \cite{mayne2005robust})} for $w_k\in \W$ with an attraction region $X_N$, i.e., there exists $c>0, \gamma_1\in (0,1)$, such that for any $x_0\in X_N$, for any $w_k\in \W$.
	$$\textup{dist}(x_t, \Xi)\leq c \gamma_1^t \textup{dist}(x_0, \Xi).$$
	
\end{theorem}
Theorem \ref{thm: rob exp stable} suggests that (RMPC \cite{mayne2005robust}) can quickly reduce the distance between $x_t$ and $\Xi$, i.e. it can drive a large initial state $x_0\not=0$ quickly to a neighborhood around $\Xi$, which is also a neighborhood around the origin.

Based on the robust exponential stability, we can build a connection between the infinite horizon averaged cost of RMPC  and that of the safe linear policy $\mathbb K$.
\begin{theorem}[Connection between RMPC in \cite{mayne2005robust} and linear control's infinite-horizon costs]\label{thm: inf cost of rmpc=lqr cost of K}
	Consider (RMPC \cite{mayne2005robust}) defined above with $\Kb$ satisfying the requirements in \cite{mayne2005robust}.  For any $x_0\in X_N$, the infinite-horizon averaged cost of RMPC in \cite{mayne2005robust}  equals the infinite-horizon averaged cost of $\Kb$, i.e.
	$$J(\textup{RMPC in \cite{mayne2005robust}})=J(\Kb),$$
	
\end{theorem}
The proof is deferred to the end of this appendix.

Notice that $\mathbb K$ is a pre-fixed safe linear policy, so by Theorem \ref{thm: inf cost of rmpc=lqr cost of K}, we have $J(K^*)\leq J(\text{RMPC in \cite{mayne2005robust}})$, where  $K^*$ is our regret benchmark, i.e., the optimal safe linear policy. This suggests  that RMPC in \cite{mayne2005robust} achieves similar or worse performance than the optimal safe linear policy in the long run. Since our adaptive control algorithm enjoys a sublinear regret compared to the optimal safe linear policy, Theorem \ref{thm: inf cost of rmpc=lqr cost of K} suggests that our algorithm achieves the same regret bound even if we include RMPC in \cite{mayne2005robust} to the benchmark policy set. Further, if $\mathbb K\not=K^*$, our adaptive algorithm can even achieve better performance than RMPC in \cite{mayne2005robust} at around the equilibrium point 0.

Nevertheless, one major strength of RMPC in \cite{mayne2005robust} compared with our algorithm is that RMPC can guarantee safety for large nonzero $x_0$  and can drive a large state exponentially to a small neighborhood of 0. Therefore, an interesting and natural idea is to combine RMPC in \cite{mayne2005robust} with our algorithm to achieve the strengths of both methods: quickly and safely drive a large initial state to a neighborhood around 0, and learning to optimize the performance around 0.\footnote{Though RMPC in \cite{mayne2005robust} requires a known model, there are standard approaches to extend RMPC to handle model uncertainties, e.g., \cite{kohler2019linear,lu2019robust}.} We leave more studies on this combination as future work.

% in the long run by our sublinear regret bound, our algorithm perform

% With our $\tilde O(T^{2/3})$ reg
 
%  $K^*$, per

%  we can show that our Algorithm \ref{alg: online algo} achieves $\tilde O(T^{2/3})$ regret even when compared with (RMPC \cite{mayne2005robust}). 
% \begin{corollary}
% 	Under the conditions in Theorem \ref{thm: regret bound}, for any  (RMPC \cite{mayne2005robust}) with admissible parameters required by \cite{mayne2005robust}, we have
% 	$$ \sum_{t=0}^{T-1} l(x_t,u_t)-TJ(\textup{RMPC in  \cite{mayne2005robust}})\leq \tilde O(T^{2/3}),$$
% 	where $x_t, u_t$ are generated by Algorithm \ref{alg: online algo}.
% \end{corollary}

\begin{remark}
	Since our proof relies on the robust exponential stability property of RMPC in \cite{mayne2005robust}, for other RMPC schemes without this property, we still cannot include them to our benchmark policy class and generate a sublinear regret. We leave the regret analysis compared with other RMPC schemes without robust exponential stability as future work. Further, we note that there are a few papers on the regret analysis with RMPC as the benchmark, e.g., \cite{wabersich2018linear,muthirayan2020regret}. However, \cite{wabersich2018linear} allows constraint violation during the learning process and allows restarts when policies are updated, and \cite{muthirayan2020regret} does not consider state constraints and the proposed algorithm involves an intractable oracle. In conclusion, the regret analysis with RMPC as the benchmark is largely under-explored and is an important direction for future research.
\end{remark}

\begin{proof}[Proof of Theorem \ref{thm: inf cost of rmpc=lqr cost of K}]

	%\red{?? revise the proof, the reviewers will read it!}

	To prove Theorem \ref{thm: inf cost of rmpc=lqr cost of K}, we introduce some necessary results from the existing literature and some  lemmas based on these existing results.
	
	Firstly, we review the structure of constrained LQR's solution proved in \cite{bemporad2002explicit}.
	
	\begin{proposition}[Corollary 2 and Theorem 4 and Section 4.4 in \cite{bemporad2002explicit}]\label{prop: CLQR solution structure}
		Consider (CLQR) with p.d. quadratic costs and polytopic constraints below:
		\begin{equation}\tag{CLQR}
			\begin{aligned}
				\min_{ u_{t+k|t}} \ &\sum_{k=0}^{W-1} l( x_{t+k|t},  u_{t+k|t})+  x_{t+W|t}^\top P  x_{t+W|t}^\top\\
				\text{s.t. \ } &  x_{t+k+1|t}=A_*  x_{t+k|t}+B_*  u_{t+k|t}, \quad k\geq 0\\
				& D_x  x_{t+k|t}\leq d_x, \quad \forall 0\leq k \leq W-1\\
				& D_u  u_{t+k|t}\leq d_u, \quad  \forall 0\leq k \leq W-1\\
				& D_{term} x_{t+W|t}\leq d_{term}\\
				& x_{t|t}=x
			\end{aligned}
		\end{equation}
		Denote the optimal policy as $\pi_{CLQR}(x)= u_{t|t}^*$, and denote the feasible region as $X_N$. Then, $X_N$ is convex, and $\pi_{CLQR}(x)$ is continuous and PWA on a finite number of closed convex polytopic regions. that is, 
		$$\pi_{CLQR}(x)=K_i x + b_i, \quad G_i x\leq h_i, \quad i=0, 1, \dots, N_{clqr}.$$
		Further, the number of different gain matrices can bounded by a constant $\bar N_{clqr-gain}$ that only depends on the dimensionality of the problem.
	\end{proposition}
	
	Based on Proposition \ref{prop: CLQR solution structure}, we have that $\pi_{CLQR}(x)
	$ is Lispchitz continous with Lipschitz factor $L_{CLQR}= \max_{i} \|K_i\|_2$ since $\pi_{CLQR}(x)$ is continuous and piecewise-affine with respect to $x$.

	Next, we will use the exponential convergence results of RMPC in \cite{mayne2005robust}.
	\begin{proposition}[See the proof of Theorem 1 in \cite{mayne2005robust}]\label{prop: RMPC exp convergence}
		There exists $c_1>0$ and $\rho\in (0,1)$ such that for any $x_0\in X_N$, and for any admissible disturbances $w_k$, we have
		$$ \|x_{t|t}^*(x_t)\|_2\leq c_1 \rho^t \|x_{0|0}^*(x_0)\|_2.$$
	\end{proposition}

	Based on this, we can also show the exponential decay of $u_{t|t}^*(x_t)$.
	\begin{lemma}\label{lem: ut|t exp. convergence}
		There exists $c_2>0$ and $\rho\in (0,1)$ such that for any $x_0\in X_N$, and for any admissible disturbances $w_k$, $u_{t|t}^*(x_{t|t}^*)$ is Lipschitz continous with a finite factor denoted as $L_{rmpc}$ on a convex feasible set. Further, we have
		$ \|u_{t|t}^*(x_t)\|_2\leq c_2 \rho^t ,$
		where $c_2= L_{rmpc} c_1 x_{\max}$.
	\end{lemma}
	\begin{proof}
		First of all, we point out that for the (RMPC \cite{mayne2005robust}) optimization, when $x_{t|t}^*$ is fixed, then $u_{t|t}^*$ can be viewed as $u_{t|t}^*=\pi_{CLQR}(x_{t|t}^*)$ for a (CLQR) problem with the same polytopic constraints and strongly convex quadratic cost functions with (RMPC \cite{mayne2005robust}). Therefore, $u_{t|t}^*(x_{t|t}^*)$ is Lipschitz continous with a finite factor denoted as $L_{rmpc}$ on a convex feasible set. 
		
		Further, notice that $u_{t|t}^*(0)=0$. Therefore,
		\begin{align*}
			\|u_{t|t}^*(x_{t|t}^*)\|_2=\|u_{t|t}^*(x_{t|t}^*)-u_{t|t}^*(0)\|_2\leq L_{rmpc}\|x_{t|t}^*\|_2\leq L_{rmpc} c_1\rho^t \|x_{0|0}^*(x_0)\|_2\leq c_2\rho^t
		\end{align*}
		where $c_2= L_{rmpc} c_1 x_{\max}$.
	\end{proof}
	
	Lastly, a technical lemma of a standard results. The proof is very straightforward.
	\begin{lemma}
		Consider $y^+=A_{\mathbb K} y+w$, where $y_0=x_0\in \mathbb X$ and $p=-\mathbb K y$. Since $\mathbb K$ is $(\kappa, \gamma)$ strongly convex, both $y$ and $p$ are bounded by
		$$ \|y_t\|_2\leq \|w\|_2 \kappa^2/\gamma + \kappa^2 x_{\max}=y_{\max}, \|p_t\|_2\leq \|w\|_2 \kappa^3/\gamma + \kappa^2 x_{\max}=p_{\max}.$$
	\end{lemma}

	Now, we are ready for the proof of Theorem \ref{thm: inf cost of rmpc=lqr cost of K}.
	\begin{proof}[Proof of Theorem \ref{thm: inf cost of rmpc=lqr cost of K}]
		The closed-loop system of (RMPC \cite{mayne2005robust}) is 
		$$ x_{t+1}=A_*x_t+B_* \pi_{RMPC}(x_t)+w_t=A_*x_t-B_*\mathbb K x_t + B_*( \mathbb K  x_{t|t}^*(x_t)+u_{t|t}^*(x_t))+w_t.$$
		
		Consider a possibly unsafe system: 
		$$ y_{t+1}=A_*y_t + B_* p_t +w_t, \quad p_t=-\mathbb K y_t$$
		with the same sequence of disturbances and $y_0=x_0$.

		The dynamics of the error $e_t=x_t-y_t$ is
		$$ e_{t+1}=A_{\mathbb K} e_t + \upsilon_t$$
		where $A_{\mathbb K}=A_*-B_*\Kb$, and $\upsilon_t=B_*( \mathbb K  x_{t|t}^*(x_t)+u_{t|t}^*(x_t))$.
		Notice that by Proposition \ref{prop: RMPC exp convergence} and Lemma \ref{lem: ut|t exp. convergence}, we have
		$$\|\upsilon_t\|_2 \leq \|B_*\|_2 (\kappa c_1 \rho^{t}x_{\max}+ c_2 \rho^t) = c_3\rho^t,$$
		where $c_3=  \|B_*\|_2 (\kappa c_1 x_{\max}+ c_2) $.
		
		Therefore, 
		\begin{align*}
			\|	e_t\|_2&= \|\upsilon_{t-1}+A_{\mathbb K} \upsilon_{t-2}+A_{\mathbb K}^{t-1}\upsilon_{0}\|_2\\
			& \leq c_3 \rho^{t-1}+ \kappa^2 (1-\gamma) c_3 \rho^{t-2}+ \dots\\
			& \leq c_3 \kappa^2 t \max(\rho, 1-\gamma)^{t-1}=c_4 t\rho_0^{t-1}
		\end{align*}
		where $\rho_0=\max(\rho, 1-\gamma)\in (0,1)$ and $c_4=  c_3 \kappa^2$.
		Further, 
		$$\|u_t-p_t\|_2 =\|-\mathbb K  e_t+\upsilon_t\|_2 \leq \kappa c_4 t \rho_0^{t-1}+ c_3 \rho^t \leq c_5 t \rho_0^{t-1},$$
		where $c_5=c_4\kappa+c_3/\rho$. 
		
		Therefore, the stage cost difference is
		\begin{align*}
			|l(x_t,u_t)-l(y_t,p_t)|&\leq \|Q\|_2 \|e_t\|_2 (x_{\max}+ y_{\max})+ \|R\|_2 \|u_t-p_t\|_2 \|u_{\max}+p_{\max}\|_2\\
			& \leq \|Q\|_2 (x_{\max}+ y_{\max})c_4 t \rho_0^{t-1}+ \|R\|_2 \|u_{\max}+p_{\max}\|_2c_5 t \rho_0^{t-1} = c_6 t \rho_0^{t-1}
		\end{align*}
		where $c_6=  \|Q\|_2 (x_{\max}+ y_{\max})c_4+  \|R\|_2 \|u_{\max}+p_{\max}\|_2c_5$.
		
		Therefore, 
		\begin{align*}
			\left|\frac{1}{T}\E \sum_{t=0}^{T-1}l(x_t,u_t)-l(y_t,p_t) \right|\leq \frac{1}{T}\sum_{t=0}^{T-1}\E |l(x_t,u_t)-l(y_t,p_t) | \leq \frac{1}{T} c_6 /(1-\rho_0)^2
		\end{align*}
		By taking $T\to +\infty$, we have 
		$ \lim_{T\to +\infty}\frac{1}{T}\E \sum_{t=0}^{T-1}l(x_t,u_t)-l(y_t,p_t) =0$.
		Since $ \lim_{T\to +\infty}\frac{1}{T}\E l(y_t,p_t) =J(\mathbb K)$, we have
		$\lim_{T\to +\infty}\frac{1}{T}\E \sum_{t=0}^{T-1}l(x_t,u_t)= J(\mathbb K).$
	\end{proof}
\end{proof}

%\section{Discussions on Robust MPC in \citet{mayne2005robust}}

\section{Additional Proofs}\label{append: add proofs}

\subsection{Proof of Lemma \ref{lem: Definition of epsilon theta(r)}}\label{subsec: proof of Lemma epsilon theta}

The proof relies on the following two lemmas.

\begin{lemma}[Definition of $\epsilon_{\hat w}$]\label{lem: define epsilon hat w}
	Under the conditions in Lemma \ref{lem: Definition of epsilon eta}, 
	\begin{align*}
		\sum_{k=1}^{H_t} D_{x,i}^\top A_*^{k-1}(w_{t-k}-\hat w_{t-k}) \leq \epsilon_{\hat w}(r)
		%	\|D_x\sum_{k=2}^{2H}\sum_{i=1}^H A_*^{i-1} B_* M_{t-i}[k-i] (\hat w_{t-k}-w_{t-k})\one_{(1\leq k-i \leq H)}\|_\infty \leq \epsilon_{\hat w,t}=\tilde c_{\hat w} \sqrt{mn} \max_{2\leq k \leq 2H} \|\hat w_{t-k}-w_{t-k}\|_2
	\end{align*}
	%		Further, if  $\|z_t\|_2\leq \sqrt{x_{\max}^2+u_{\max}^2}$, 
	%		\begin{align*}
	%			\|D_x\sum_{k=2}^{2H}\sum_{i=1}^H A_*^{i-1} B_* M_{t-i}[k-i] (\hat w_{t-k}-w_{t-k})\one_{(1\leq k-i \leq H)}\|_\infty \leq \epsilon_{\hat w,t}= c_{\hat w} \sqrt{mn} \max_{2\leq k \leq 2H}  \|\hat \theta_{t-k}-\theta_*\|_2.
	%		\end{align*}
	%		We define $\epsilon_{\hat w}(r)= c_{\hat w} \sqrt{mn} r$.
\end{lemma}

\begin{proof}
	%	For notational simplicity, we omit the subscript $t$ in $H_t$ in the proof.
	\begin{align*}
		\|D_x  \sum_{k=1}^{H_t}A_*^{k-1}(w_{t-k}-\hat w_{t-k}) \|_\infty & \leq \|D_x\|_\infty \sum_{k=1}^{H_t} \|A_*^{k-1}(w_{t-k}-\hat w_{t-k}) \|_\infty\\
		& \leq  \|D_x\|_\infty \sum_{k=1}^{H_t} \|A_*^{k-1}(w_{t-k}-\hat w_{t-k}) \|_2\\
		& \leq  \|D_x\|_\infty \sum_{k=1}^{H_t}  \kappa (1-\gamma)^{k-1} r z_{\max}\\
		& \leq \|D_x\|_\infty \kappa/\gamma z_{\max} r=\epsilon_{\hat w}(r)
	\end{align*}
	
\end{proof}

\begin{lemma}[Definition of $\epsilon_{\hat \theta}$]\label{lem: define epsilon hat theta}
	For any $\mb M\in \M$, any $\hat \theta, \theta \in \Theta^{(0)}$ such that $\|\hat \theta-\theta\|_F\leq r$, we have
	
	$$	|g_i^x(\mb M;\hat \theta)-	g_i^x(\mb M;\theta)|\leq \epsilon_{\hat \theta}(r)$$
	where $\epsilon_{\hat \theta}(r)=c_{\hat \theta} r\sqrt{mn}$.
	
\end{lemma}
%	\red{Note: I don't have $L_2$ estimation error bound because we have to do a projection and $L_2$ norm is lost! But is this projection necessary for LS? Can we make it work without it?}

\begin{proof}
	Firstly, we show that it suffices to prove an upper bound of a simpler quantity. 
	\begin{align*}
		|g_i^x(\mb M;\hat \theta)-	g_i^x(\mb M;\theta)|& = 	|\sum_{k=1}^{2H} \|D_{x,i}^\top \Phi_k^x(\mb M;\hat \theta)\|_1-\|D_{x,i}^\top \Phi_k^x(\mb M; \theta)\|_1|w_{\max}\\
		& \leq \sum_{k=1}^{2H}|\|D_{x,i}^\top \Phi_k^x(\mb M;\hat \theta)\|_1-\|D_{x,i}^\top \Phi_k^x(\mb M; \theta)\|_1|w_{\max}\\
		& \leq \sum_{k=1}^{2H}\|D_{x,i}^\top \Phi_k^x(\mb M;\hat \theta)-D_{x,i}^\top \Phi_k^x(\mb M; \theta)\|_1w_{\max}\\
		& \leq \sum_{k=1}^{2H}\|D_{x}\|_\infty \| \Phi_k^x(\mb M;\hat \theta)- \Phi_k^x(\mb M; \theta)\|_\infty w_{\max}\\
	\end{align*}
	thus, it suffices to bound $\sum_{k=1}^{2H} \|\Phi_k^x(\mb M;\hat \theta)- \Phi_k^x(\mb M; \theta)\|_\infty$. To bound this, we need several small lemmas  below.
	\begin{lemma}
		When  $\|\theta-\hat \theta\|_F\leq r$,
		we have	$\max(\|\hat A-A\|_2, \|\hat B-B\|_2)\leq \max(\|\hat A-A\|_F, \|\hat B-B\|_F)\leq r$
	\end{lemma}
	This is quite straightforward so the proof is omitted.
	\begin{lemma}\label{lem: perturb Ak AkB}
		For any $k\geq 0$, any $\hat \theta, \theta \in \Theta^{(0)}$ such that $\|\hat \theta-\theta\|_F\leq r$, we have
		\begin{align*}
			\|A^k-\hat A^k\|_2&\leq k \kappa^2 (1-\gamma)^{k-1}r \one_{(k\geq 1)}\\
			\|A^kB-\hat A^k\hat B\|_2& \leq k \kappa^2\kappa_B (1-\gamma)^{k-1}r \one_{(k\geq 1)}+\kappa(1-\gamma)^k r
		\end{align*}
	\end{lemma}
	\begin{proof}
		When $k=0$, $\|A^0-\hat A^0\|_2=0$. When $k\geq 1$,
		\begin{align*}
			\|\hat A^k- A^k\|_2&=\|\sum_{i=0}^{k-1} \hat A^{k-i-1}(\hat A-A)A^{i}\|_2\\
			&\leq \sum_{i=0}^{k-1} \| \hat A^{k-i-1}\|_2\|\hat A-A\|\|A^{i}\|_2\\
			& \leq \sum_{i=0}^{k-1}\kappa (1-\gamma)^{k-i-1}\epsilon \kappa (1-\gamma)^i\\
			&=k \kappa^2 r (1-\gamma)^{k-1}\\
			\|\hat A^k\hat B - A^kB\|_2&\leq \|\hat A^k\hat B-A^k \hat B\|_2+ \| A^k \hat B-\hat A^k \hat B\|_2\\
			& \leq k \kappa^2\kappa_B r (1-\gamma)^{k-1}\one_{(k\geq 1)}+ \kappa(1-\gamma)^kr
		\end{align*}
	\end{proof}
	Now, we can bound $\sum_{k=1}^{2H} \|\Phi_k^x(\mb M;\hat \theta)- \Phi_k^x(\mb M; \theta)\|_\infty$.  For any $1\leq k \leq 2H$, 
	\begin{align*}
		&\|\Phi_k^x(\mb M;\hat \theta)- \Phi_k^x(\mb M; \theta)\|_\infty \\
		= \ & \| \hat A^{k\!-\!1} \one_{(k\leq H)} \!+\!\sum_{i=1}^H \hat A^{i-1}\hat B M_{t-i}[k-i] \one_{(1\leq k-i \leq H)}\!- \!A^{k\!-\!1} \one_{(k\leq H)} \!-\! \sum_{i=1}^H  A^{i\!-\!1} B M_{t\!-\!i}[k\!-\!i] \one_{(1\leq k-i \leq H)}\|_\infty\\
		\leq \ & \|  \hat A^{k-1}- A^{k-1}\|_\infty \one_{(k\leq H)} + \sum_{i=1}^H \|  (\hat A^{i-1}\hat B-A^{i-1}B) M_{t-i}[k-i]\|_\infty \one_{(1\leq k-i \leq H)}\\
		\leq \ & \sqrt n \|  \hat A^{k-1}- A^{k-1}\|_2 \one_{(k\leq H)} + \sqrt m\sum_{i=1}^H \|\hat A^{i-1}\hat B-A^{i-1}B\|_2 2\sqrt n \kappa^2 (1-\gamma)^{k-i-1}\one_{(1\leq k-i \leq H)}
		%	\leq \ & \sqrt n (k-1) \kappa^2 (1-\gamma)^{k-2}r \one_{(2\leq k\leq H)}+ \sqrt m\sum_{i=1}^H \|\hat A^{i-1}\hat B-A^{i-1}B\|_2 2\sqrt n \kappa^2 (1-\gamma)^{k-i-1}\one_{(1\leq k-i \leq H)}\\
	\end{align*}
	There are two terms in the last right-hand-side of the inequality above. We sum each term over $k$ below.
	\begin{align*}
		\sum_{k=1}^{2H}\sqrt n \|  \hat A^{k-1}- A^{k-1}\|_2 \one_{(k\leq H)} & \leq 	\sum_{k=1}^{2H} \sqrt n (k-1) \kappa^2 (1-\gamma)^{k-2}r \one_{(2\leq k\leq H)}\leq \sqrt n \kappa^2 r/\gamma^2
	\end{align*}
	\begin{align*}
		&	\sum_{k=1}^{2H}	\sqrt m\sum_{i=1}^H \|\hat A^{i-1}\hat B-A^{i-1}B\|_2 2\sqrt n \kappa^2 (1-\gamma)^{k-i-1}\one_{(1\leq k-i \leq H)}\\
		\leq \ & 	\sum_{k=1}^{2H}	\sqrt m\sum_{i=1}^H (i-1)\kappa^2\kappa_B (1-\gamma)^{i-2}r \one_{(i\geq 2)} 2\sqrt n \kappa^2 (1-\gamma)^{k-i-1}\one_{(1\leq k-i \leq H)}\\
		&+		\sum_{k=1}^{2H}\sqrt m\sum_{i=1}^H \kappa(1-\gamma)^{i-1} r 2\sqrt n \kappa^2 (1-\gamma)^{k-i-1}\one_{(1\leq k-i \leq H)}\\
		=\ &2 \sqrt{mn}  \kappa^4\kappa_Br \sum_{i=1}^H\sum_{j=1}^H(i-1)(1-\gamma)^{i-2}(1-\gamma)^{j-1}+2\sqrt{mn}\kappa^3r\sum_i \sum_j(1-\gamma)^{i-1} (1-\gamma)^{j-1} \\
		=\ & 2 \sqrt{mn}  \kappa^4\kappa_Br/\gamma^3+ 2\sqrt{mn}\kappa^3r/\gamma^2
	\end{align*}
\end{proof}

\subsection{Proof of Lemma \ref{lem: define epsilon V}}\label{subsec: proof of epsilonv}

%	\red{change $H$ to $H_t$ and change $  \Phi $ to no, change no to $\tilde \Phi$.}
For notational simplicity, we omit the subscript $t$ in $H_t$ in this proof.
Remember that 	$g_i^x(\mb M_{t-H:t-1};\theta)=\sum_{s=1}^{2H}  \|D_{x,i}^\top  \Phi^x_{s}(\mb M_{t-H:t-1};\theta)\|_1  w_{\max}$.
\begin{align*}
	|\tilde g_i^x(\mb M_{t-H:t-1};\theta)&-  g_i^x(\mb M;\theta)|=\left| \sum_{k=1}^{2H}\|D_{x,i}^\top  \tilde \Phi^x_{k}(\mb M_{t-H:t-1};\theta)\|_1-\|D_{x,i}^\top  \Phi^x_{k}(\mb M_{t};\theta)\|_1\right|w_{\max}\\
	& \leq \sum_{k=1}^{2H}\left|\|D_{x,i}^\top  \Phi^x_{k}(\mb M_{t-H:t-1};\theta^*)\|_1-\|D_{x,i}^\top   \Phi^x_{k}(\mb M_{t};\theta)\|_1\right|w_{\max}\\
	& \leq \sum_{k=1}^{2H}\|D_{x,i}^\top ( \tilde\Phi^x_{k}(\mb M_{t-H:t-1};\theta)-  \Phi^x_{k}(\mb M_{t};\theta))\|_1w_{\max}\\
	& \leq \sum_{k=1}^{2H}\|D_x\|_\infty \|\tilde \Phi^x_{k}(\mb M_{t-H:t-1};\theta)-  \Phi^x_{k}(\mb M_{t};\theta)\|_\infty w_{\max}\\
	&\leq \sum_{k=1}^{2H}\|D_x\|_\infty \|\sum_{i=1}^H A^{i-1}B(M_{t-i}[k-i]-M_t[k-i])\|_\infty \one_{(1\leq k-i \leq H)}w_{\max}\\
	& \leq  \sum_{k=1}^{2H}\|D_x\|_\infty \sum_{i=1}^H \|A^{i-1}B\|_\infty \|M_{t-i}[k-i]-M_t[k-i]\|_\infty \one_{(1\leq k-i \leq H)}w_{\max}\\
	& \leq \|D_x\|_\infty\sqrt m w_{\max} \sum_{k=1}^{2H} \sum_{i=1}^H \kappa (1-\gamma)^{i-1}\kappa_B \|M_{t-i}[k-i]-M_t[k-i]\|_\infty \one_{(1\leq k-i \leq H)}\\
	& = \|D_x\|_\infty\sqrt m w_{\max}\kappa \kappa_B \sum_{i=1}^H\sum_{j=1}^H (1-\gamma)^{i-1} \|M_{t-i}[j]-M_t[j]\|_\infty\\
	& \leq \|D_x\|_\infty\sqrt m w_{\max}\kappa \kappa_B \sqrt{nH}\sum_{i=1}^H(1-\gamma)^{i-1} \|\mb M_{t-i}-\mb M_t\|_F\\
	& \leq \|D_x\|_\infty\sqrt{mnH} w_{\max}\kappa \kappa_B \sum_{i=1}^H(1-\gamma)^{i-1} i \Delta_{M}\\
	& \leq \|D_x\|_\infty\sqrt{mnH} w_{\max}\kappa \kappa_B /\gamma^2\Delta_{M}
\end{align*}
where the third last inequality is because $M[j]\in \R^{m\times n}$
$$ \sum_{j=1}^H \|M[j]\|_\infty \leq  \sum_{j=1}^H \|M[j]\|_2 \sqrt{n}\leq \sum_{j=1}^H \|M[j]\|_F \sqrt{n}\leq \|\mb M\|_F \sqrt{n}\sqrt{H}
$$

\subsection{Proof of Lemma \ref{lem: w.p.1 bdd on xt}}\label{subsec: as. bdd on xt}

For notational simplicity, we define $y_t= \sum_{i=1}^{H_t} A_*^{i-1} w_{t-i}+ \sum_{k=2}^{2{H_t}}\sum_{i=1}^{H_t} A_*^{i-1} B_* M_{t-i}[k-i] \hat w_{t-k}\one_{1\leq k-i \leq {H_t}}+ \sum_{i=1}^{H_t} A_*^{i-1}B_* \eta_{t-i}$.  Since $A_*$ is $(\kappa, \gamma)$-stable, we have
\begin{align*}
	&	\|y_t\|_2 \leq \sum_{i=1}^{H_t} \|A_*^{i-1}\|_2 \|w_{t-i}\|_2 + \sum_{k=2}^{2{H_t} }\sum_{i=1}^{H_t} \|A_*^{i-1} B_* M_{t-i}[k-i] \hat w_{t-k}\|_2\one_{1\leq k-i \leq {H_t}}+ \sum_{i=1}^{H_t} \|A_*^{i-1}B_* \eta_{t-i}\|_2\\
	&\leq \sum_{i=1}^{H_t} \kappa(1-\gamma)^{i-1}\sqrt n w_{\max} + \sum_{k=2}^{2{H_t} }\sum_{i=1}^{H_t}\|A_*^{i-1} B_*\|_2 \| M_{t-i}[k-i] \hat w_{t-k}\|_2\one_{1\leq k-i \leq {H_t}}+ \sum_{i=1}^{H_t} \|A_*^{i-1}B_*\|_2 \| \eta_{t-i}\|_2\\
	& \leq \kappa\sqrt n w_{\max} /\gamma + \sum_{k=2}^{2{H_t} }\sum_{i=1}^{H_t} \kappa (1-\gamma)^{i-1}\kappa_B \sqrt m \| M_{t-i}[k-i] \hat w_{t-k}\|_\infty \one_{1\leq k-i \leq {H_t}}+ \sum_{i=1}^{H_t}\kappa (1-\gamma)^{i-1}\kappa_B \sqrt n \eta_{\max}\\
	& \leq  \kappa\sqrt n w_{\max} /\gamma + \sum_{k=2}^{2{H_t} }\sum_{i=1}^{H_t} \kappa (1-\gamma)^{i-1}\kappa_B \sqrt m 2\sqrt n \kappa^2(1-\gamma)^{k-i-1} w_{\max } \one_{1\leq k-i \leq {H_t}}+ \kappa\kappa_B/\gamma \sqrt n\eta_{\max}\\
	& \leq  \kappa\sqrt n w_{\max} /\gamma + \kappa\kappa_B/\gamma \sqrt n\eta_{\max}+\kappa^3\kappa_B 2\sqrt{mn}w_{\max}\sum_{i=1}^{H_t}\sum_{j=1}^{H_t}(1-\gamma)^{i-1}(1-\gamma)^{j-1}\\
	& \leq \sqrt n( \kappa w_{\max} +\kappa\kappa_B\eta_{\max} )/\gamma+ \kappa^3\kappa_B 2\sqrt{mn}w_{\max}/\gamma^2\\
	& \leq 2\sqrt n \kappa w_{\max}/\gamma +  \kappa^3\kappa_B 2\sqrt{mn}w_{\max}/\gamma^2\leq c_{bx}\sqrt{mn}
\end{align*}
Remember that $	x_t=A_*^{H_t} x_{t-{H_t}}+y_t$ and 	 and $\|x_t\|_2=0\leq b_x$ for $t\leq 0$. We prove the bound on $x_t$ by induction. Suppose at $t\geq 0$, $\|x_{t-{H_t}}\|_2 \leq b_x$, then 
\begin{align*}
	\|x_t\|_2 &\leq \|A_*^{H_t} \|_2 \|x_{t-{H_t}}\|_2 + \|y_t\|_2 \leq \kappa (1-\gamma)^{H_t} b_x + 2\sqrt n \kappa w_{\max}/\gamma +  \kappa^3\kappa_B 2\sqrt{mn}w_{\max}/\gamma^2\\
	&\leq b_x/2+ 2\sqrt n \kappa w_{\max}/\gamma +  \kappa^3\kappa_B 2\sqrt{mn}w_{\max}/\gamma^2=b_x
\end{align*}
where the last inequality is by $ \kappa (1-\gamma)^{H_t} \leq 1/2$ when ${H_t}\geq \log(2\kappa)/\log((1-\gamma)^{-1})$. This completes the proof.

\subsection{Proof of Lemma \ref{lem: bound on Gf}}\label{subsec: Gf bound}
\begin{proof}
	We omit $\theta$ in this proof for simplicity of notations.
	
	For any $H\geq 1$, 
	define $\M_{out,H}=\{\mb M\in \R^{mnH}: \|M[k]\|_\infty \leq 4\kappa^2 \sqrt n (1-\gamma)^{k-1}\}$. Notice that $\M_H\subseteq interior(\M_{out,H})$. Therefore, for any $\mb M\in \M_H$,
	\begin{align*}
		\|	\nabla f(\mb M;\theta)\|_F&=\sup_{\Delta \mb M\not= 0, \mb M+\Delta \mb M\in \M_{out,H}}\frac{\langle \nabla f(\mb M;\theta), \Delta \mb M\rangle}{\|\Delta \mb M\|_F}\\
		&\leq \sup_{\Delta \mb M\not= 0, \mb M+\Delta \mb M\in \M_{out,H}}\frac{ f(\mb M+\Delta\mb M)-f(\mb M)}{\|\Delta \mb M\|_F}
		%=\lim_{\Delta \mb M\to 0, \mb M+\Delta \mb M\in \M_{out,H}}\frac{f(\mb M+\Delta \mb M;\theta)- f(\mb M;\theta)}{\|\}
	\end{align*}
	For $\mb M, \mb M'\in \M_{out,H}$, we bound the following.
	\begin{align*}
		\|\tilde x-\tilde x'\|_2 & \leq \sum_{k=1}^{2H}\|(  \Phi_k^x(\mb M)-  \Phi_k^x(\mb M'))w_{t-k}\|_2\\
		& \leq \sum_{k=1}^{2H}\|\sum_{i=1}^{H^{(e)}} A^{i-1}B(M[k-i]-M'[k-i])\one_{(1\leq k-i\leq H)} w_{t-k}\|_2\\
		& \leq \sum_{j=1}^{H} O(\sqrt n) \|M[j]-M'[j]\|_2\\
		& \leq  \sum_{j=1}^{H} O(\sqrt n) \|M[j]-M'[j]\|_F\\
		&\leq O(\sqrt n\sqrt{H}) \|\mb M-\mb M'\|_F\\
		\|\tilde u-\tilde u'\|_2 &\sum_{k=1}^{H}\|M[k]-M'[k]\|_2\sqrt n w_{\max}\leq O(\sqrt n\sqrt{H}) \|\mb M-\mb M'\|_F
	\end{align*}
	where the third inequality uses $\theta \in \Theta_{ini}$.
	
	Further, even though we make $\M_{out,H}$ larger, but we don't change the dimension, so by Lemma 24, $\|\tilde x\|_2\leq \sqrt{mn}$. Further, even when we don't have additional conditions on $\mb M$, we still have $\|\tilde u\|_2\leq O(\sqrt{mn})$. Therefore, for $\mb M, \mb M'\in \M_{out,H}$, 
	\begin{align*}
		|f(\mb M)-f(\mb M')|& \leq O(\sqrt{mn}\sqrt n\sqrt{H})\|\mb M-\mb M'\|_F
	\end{align*}

	Therefore,
	\begin{align*}
		\|	\nabla f(\mb M;\theta)\|_F&\leq \sup_{\Delta \mb M\not= 0, \mb M+\Delta \mb M\in \M_{out,H}}\frac{ f(\mb M+\Delta\mb M)-f(\mb M)}{\|\Delta \mb M\|_F}\\
		& \leq  \sup_{\Delta \mb M\not= 0, \mb M+\Delta \mb M\in \M_{out,H}}\frac{O(\sqrt{mn}\sqrt n\sqrt{H})\|\Delta \mb M\|_F}{\|\Delta \mb M\|_F}\leq O(n\sqrt m \sqrt H)
		%=\lim_{\Delta \mb M\to 0, \mb M+\Delta \mb M\in \M_{out,H}}\frac{f(\mb M+\Delta \mb M;\theta)- f(\mb M;\theta)}{\|\}
	\end{align*}
\end{proof}
\subsection{Proof of Lemma \ref{lem: cost diff lemma for linear constrained opt}}\label{subsec: perturb QP}
\begin{proof}
	Notice that $\Omega_1$ and $\Omega_3$ satisfies the conditions in Proposition 2 in \cite{li2020online}. Therefore,
	$$|\min_{\Omega_1}f(x)-\min_{\Omega_3}f(x)| \leq \frac{Ld_{\Omega_0} \|\Delta_1-\Delta_3\|_\infty}{\min_{\{i: (\Delta_1)_i>(\Delta_3)_i\}} (h-\Delta_1-Cx_F)_i}$$
	Notice that
	\begin{align*}
		(\Delta_3)_i=\begin{cases}
			(\Delta_1)_i, & \text{ if } (\Delta_1)_i \geq (\Delta_2)_i\\
			(\Delta_2)_i, & \text{ if } (\Delta_1)_i < (\Delta_2)_i
		\end{cases}
	\end{align*}
	therefore, $\|\Delta_1-\Delta_3\|_\infty\leq \|\Delta_1-\Delta_2\|_\infty$. Further, $\{i: (\Delta_3)_i>(\Delta_1)_i\}=\{i: (\Delta_2)_i>(\Delta_1)_i\}\subseteq \{i: (\Delta_1)_i\not=(\Delta_2)_i\}$. So $\min_{\{i: (\Delta_3)_i>(\Delta_1)_i\}} (h-\Delta_1-Cx_F)_i\geq  \min_{\{i: (\Delta_1)_i\not=(\Delta_2)_i\}} (h-\Delta_1-Cx_F)_i\geq  \min_{\{i: (\Delta_1)_i\not=(\Delta_2)_i\}} (h-\Delta_3-Cx_F)_i$. Therefore,
	$$|\min_{\Omega_1}f(x)-\min_{\Omega_3}f(x)| \leq \frac{Ld_{\Omega_0} \|\Delta_1-\Delta_3\|_\infty}{\min_{\{i: (\Delta_1)_i>(\Delta_3)_i\}} (h-\Delta_1-Cx_F)_i}\leq \frac{Ld_{\Omega_0} \|\Delta_1-\Delta_2\|_\infty}{\min_{\{i: (\Delta_1)_i\not =(\Delta_2)_i\}} (h-\Delta_3-Cx_F)_i}$$
	
	Similarly, 
	$$|\min_{\Omega_2}f(x)-\min_{\Omega_3}f(x)| \leq \frac{Ld_{\Omega_0} \|\Delta_2-\Delta_3\|_\infty}{\min_{\{i: (\Delta_2)_i>(\Delta_3)_i\}} (h-\Delta_2-Cx_F)_i}\leq \frac{Ld_{\Omega_0} \|\Delta_1-\Delta_2\|_\infty}{\min_{\{i: (\Delta_1)_i\not =(\Delta_2)_i\}} (h-\Delta_3-Cx_F)_i}$$
	which completes the bound.
\end{proof}

\subsection{Proof of Lemma \ref{lem: bound on Part ii}}\label{subsec: proof for bound on Part ii}
In this subsection, we provide a proof for our bound on Part ii by martingale concentration inequalities.
%\red{this is unreadable, please slightly organize it!??}
\begin{lemma}
	In our Algorithm \ref{alg: online algo}, $\mb M^{(e)}\in \F(w_0, \dots, w_{t_1^{(e)}+T_D^{(e)}-1},\eta_0, \dots, \eta_{t_1^{(e)}+T_D^{(e)}-1} )=\F_{t_1^{(e)}+T_D^{(e)}}^m\subseteq  \F_{t_2^{(e)}-H^{(e)}}$. 
\end{lemma}
\begin{proof}
	By definition, we have the following fact:
	$\mb M^{(e)}\in \F(\hat \theta^{(e+1)})=\F(\{z_k, x_{k+1}\}_{k=t^{(e)}_1}^{t_1^{(e)}+T_D^{(e)}-1})=\F(w_0, \dots, w_{t_1^{(e)}+T_D^{(e)}-1},\eta_0, \dots, \eta_{t_1^{(e)}+T_D^{(e)}-1} )=\F_{t_1^{(e)}+T_D^{(e)}}^m$.
	By $\tilde W_1^{(e)}\geq H^{(e)}$, we have $t_1^{(e)}+T_D^{(e)}+H^{(e)}\leq t_2^{(e)}$, and since $\F_t^m\subseteq \F_t$, we have the last claim.
\end{proof}
\begin{lemma}
	When $t\in \T_2^{(e)}$, $w_{t-2H^{(e)}}\ind \F_{t_2^{(e)}-H^{(e)}}$
\end{lemma}
\begin{proof}
	When $t\in \T_2^{(e)}$, $t\geq t_2^{(e)}+H^{(e)}$, so $t-2H^{(e)}\geq t_2^{(e)}-H^{(e)}$. Since $\F_t$ contains up to $w_{t-1}$, we have $w_{t-2H^{(e)}}\ind \F_{t_2^{(e)}-H^{(e)}}$.
\end{proof}
\begin{lemma}\label{lem: cond exp of l(tilde x)=mf(M)}
	In our Algorithm \ref{alg: online algo}, when $t\in \T_2^{(e)}$, we have
	$\E[ l(\hat x_t, \hat u_t)\mid \F_{t_2^{(e)}-H^{(e)}}]=f(\mb M^{(e)};\theta_*)$.
\end{lemma}
\begin{proof}
	By our lemmas above, $\mb M^{(e)}\in \F_{t_2^{(e)}-H^{(e)}}$, but $w_{t-2H^{(e)}}\ind \F_{t_2^{(e)}-H^{(e)}}$. Then, by our definition of $\hat x_t, \hat u_t$ and $f(\mb M;\theta_*)$, we have the result.
\end{proof}

\begin{definition}[Martingale]\label{def: martingale}
	$\{X_t\}_{t\geq 0}$ is a martingale wrt $\{\F_t\}_{t\geq 0}$ if 
	(i) $\E|X_t|<+\infty$, (ii) $X_t\in \F_t$, (iii) $\E(X_{t+1}\mid \F_t)=X_t$ for $t\geq 0$.
	
\end{definition}
\begin{proposition}[Azuma-Hoeffding Inequality]\label{prop: azuma}
	$\{X_t\}_{t\geq 0}$ is a martingale with respect to $\{\F_t\}_{t\geq 0}$. If (i) $X_0=0$, (ii) $|X_t-X_{t-1}|\leq \sigma$ for any $t\geq 1$, then,
	for any $\alpha>0$, 
	any $t\geq 0$,
	$$ \Pb(|X_t|\geq \alpha)\leq 2\exp\left( -\alpha^2/(2t\sigma^2)\right)$$
\end{proposition}
\begin{corollary}\label{cor: azuma whp bdd}
	$\{X_t\}_{t\geq 0}$ is a martingale wrt $\{\F_t\}_{t\geq 0}$. If (i) $X_0=0$, (ii) $|X_t-X_{t-1}|\leq \sigma$ for any $t\geq 1$, then,	for any $\delta\in (0,1)$, 
	$$ |X_t|\leq \sqrt{2t}\sigma \sqrt{\log(2/\delta)}$$
	w.p. at least $1-\delta$.
	
\end{corollary}
\begin{proof}
	The proof is by letting $\alpha= \sqrt{2t \sigma^2 \log(2/\delta)}$ in Proposition \ref{prop: azuma}.
\end{proof}

\begin{lemma}\label{lem: bdd on qt}
	Define $q_t=l(\hat x_t, \hat u_t)-f(\mb M^{(e)};\theta_*)$. Then, $|q_t|\leq O(mn)$ w.p.1.
\end{lemma}
\begin{proof}
	We can show that $\|\hat x_t\|_2\leq O(\sqrt{mn})$ a.s. and $\hat u_t \in \U$ a.s. by the proofs of  Lemmas \ref{lem: ut satisfies constraints} and \ref{lem: w.p.1 bdd on xt}. Therefore,  we have $|l(\hat x_t, \hat u_t)|=O({mn})$. Since $f(\mb M^{(e)};\theta_*)=\E [l(\hat x_t, \hat u_t)\mid \F_{t_2^{(e)}-H^{(e)}}]$, we have $|f(\mb M^{(e)};\theta_*)|=O(mn)$. This completes the proof.
\end{proof}

\paragraph{Notations and definitions.} Define, for $0\leq h \leq 2H^{(e)}-1$, that
\begin{align}
	\T_{2,h}^{(e)}&=\{t\in \T_2^{(e)}: t \equiv h \mod (2H^{(e)})\}\eqqcolon\{t_h^{(e)}+2H^{(e)}, \dots,t_h^{(e)}+2H^{(e)}k_h^{(e)} \}
\end{align}
\begin{lemma}\label{lem: bdd on t_h^e, and k_h(e)}
	$t_h^{(e)}\geq t_2^{(e)}-H^{(e)}$  and $k_h^{(e)}\leq T^{(e+1)}/(2H^{(e)})$
\end{lemma}
\begin{proof}
Notice that	$t_h^{(e)}+2H^{(e)}\geq t_2^{(e)}+H^{(e)}$, so the first inequality holds.
	Besides, notice that	$2H^{(e)}k_h^{(e)}\leq t_h^{(e)}+2H^{(e)}k_h^{(e)}\leq T^{(e+1)}$, so the second inequality holds.
\end{proof}

Define
\begin{align}
	&	\tilde q^{(e)}_{h,j}=q_{t_h^{(e)}+j (2H^{(e)})}\quad \forall 1\leq j \leq k_h^{(e)}\\
&	S^{(e)}_{h,j}=
	\sum_{s=1}^j  	\tilde q^{(e)}_{h,s} \quad \forall 0\leq j \leq k_h^{(e)},\\
	& \F^{(e)}_{h,j}=
	\F_{t_h^{(e)}+j (2H^{(e)})} \quad \forall 0\leq j \leq k_h^{(e)},
\end{align}
where we define $\sum_{s=1}^0a_s=0$.
By Lemma \ref{lem: bdd on t_h^e, and k_h(e)}, we have $	\F^{(e)}_{h,0}=\F_{t_h^{(e)}}\supseteq \F_{t_2^{(e)}-H^{(e)}}$.

\begin{lemma}\label{lem: Se hj is a martingale}
	$S^{(e)}_{h,j}$ is a martingale wrt $\F^{(e)}_{h,j}$ for $j\geq 0$. Further, $S^{(e)}_{k,0}=0$, $|S^{(e)}_{h,j+1}-S^{(e)}_{h,j}|\leq O(mn)$.
\end{lemma}
\begin{proof}
	Since $|q_t|\leq O(mn)$, $\E|S^{(e)}_{h,j}|\leq O(Tmn)<+\infty$. Notice that, for $t\in \T_2^{(e)}$, $w_{t-1}, \dots, w_{t-2H^{(e)}}\in \F_t$. and $\mb M^{(e)}\in \F_t$, so $q_t \in \F_t$, so $S^{(e)}_{h,j}\in \F^{(e)}_{h,j}$. Next, $\E[S^{(e)}_{h,j+1}\mid \F^{(e)}_{h,j}]=S^{(e)}_{h,j}+\E[q^{(e)}_{h,j+1}\mid \F^{(e)}_{h,j}]=S^{(e)}_{h,j}$. So this is done. The rest is by definition, and $q_t$'s bound.
\end{proof}
\begin{lemma}
	Consider our choice of $H^{(e)}$ in Theorem \ref{thm: constraint satisfaction}.
	Let $\delta=\frac{p}{2\sum_{e=0}^{N-1} H^{(e)}}$, w.p. $1-\delta$, we have
	$|S^{(e)}_{h,k_h^{(e)}}|\leq \tilde O\left( \sqrt{k_h^{(e)}} mn\right)$.
\end{lemma}
\begin{proof}
	By Lemma \ref{lem: Se hj is a martingale}, we can apply Corollary \ref{cor: azuma whp bdd}, and obtain the bound, where we used $\log(2/\delta)=\tilde O(1)$.
\end{proof}
\begin{lemma}\label{lem: part ii bound first step}
	Consider our choice of $H^{(e)}$ in Theorem \ref{thm: constraint satisfaction}.
	For any $e$, w.p. $1-2H^{(e)}\delta$, where $\delta=\frac{p}{2\sum_{e=0}^{N-1} H^{(e)}}$, 
	$$|\sum_{h=0}^{2H^{(e)}-1}S^{(e)}_{h,k_h^{(e)}}|\leq \tilde O\left( \sqrt{T^{(e+1)}} mn\right)$$
\end{lemma}
\begin{proof}
	Define event
	$$\mathcal E_h^{(e)}=\{|S^{(e)}_{h,k_h^{(e)}}|\leq \tilde O\left( \sqrt{k_h^{(e)}} mn\right)\}$$
	When $\cap_h \mathcal E_h^{(e)}$ holds, 
	\begin{align*}
		|\sum_{t\in \T_2^{(e)}} q_t|=|\sum_{h=0}^{2H^{(e)}-1}S^{(e)}_{h,k_h^{(e)}}|%& \leq 	\sum_{h=0}^{2H^{(e)}-1} |S^{(e)}_{h,k_h^{(e)}}|\\
		& \tilde O(mn \sqrt{\sum_h k_h^{(e)}} \sqrt{2H^{(e)}} )\leq \tilde O(mn T^{(e+1)})
	\end{align*}
	where we used Lemma \ref{lem: bdd on t_h^e, and k_h(e)} and Cauchy Schwartz.
	
	Then, we have
	\begin{align*}
		\Pb(\cap_h \mathcal E_h^{(e)})&=1-\Pb(\cup_h (\mathcal E_h^{(e)})^c)\geq 1-\sum_h \Pb((\mathcal E_h^{(e)})^c)\geq 1-2H^{(e)}\delta
	\end{align*}
\end{proof}

Now, we can prove Lemma \ref{lem: bound on Part ii}. 
By Lemma \ref{lem: part ii bound first step}, w.p. $1-p$,   we have $|\sum_{h=0}^{2H^{(e)}-1}S^{(e)}_{h,k_h^{(e)}}|\leq \tilde O\left( \sqrt{T^{(e+1)}} mn\right)$ for all $e$. Then, by Lemma \ref{lem: bound Te power a by T power a}, we completed the proof.

\end{document}